\documentclass[12pt]{amsart}
\usepackage{dirtytalk}
\usepackage[utf8]{inputenc}
\usepackage[english]{babel}
\usepackage[margin=2.5cm]{geometry}	
\usepackage{color}
\usepackage{xcolor}
\usepackage{enumerate}
\usepackage{enumitem}
\usepackage{amsmath}
\usepackage{amsthm}
\usepackage{amssymb}
\usepackage{dsfont}
\usepackage{mathtools}
\mathtoolsset{showonlyrefs}
\usepackage{rotating}
\usepackage{float}
\usepackage{comment}
\usepackage{longtable}
\usepackage{bbm}
\usepackage{caption}
\usepackage{multirow}
\usepackage{tabularx}

\newcommand{\sub}[1]{_{\mathrm{#1}}}

\newcommand{\Id}{\mathds{1}}

\newcommand{\N}{\mathbb{N}}
\newcommand{\Z}{\mathbb{Z}}

\newcommand{\R}{\mathbb{R}}
\newcommand{\C}{\mathbb{C}}
\newcommand{\T}{\mathbb{T}}

\newcommand{\K}{\mathcal{K}}
\newcommand{\Hi}{\mathcal{H}}

\newcommand{\U}{\mathcal{U}}

\newcommand{\UZ}{\mathcal{U}\sub{BFZ}}

\newcommand{\W}{\operatorname{I}^{S}}
\newcommand{\PI}{\operatorname{I}^{C}}
\newcommand{\G}{\operatorname{I}^{DIII}}
\newcommand{\I}{{\operatorname{I}^{\textup{Zak}}}}

\newcommand{\inn}[2]{\left\langle #1, #2 \right\rangle}

\numberwithin{equation}{section}

\DeclareMathOperator{\tr}{tr}

\DeclareMathOperator{\Proj}{Proj}
\DeclareMathOperator{\Imm}{Range}
\DeclareMathOperator{\Pf}{Pf}

\DeclareMathOperator{\Ch}{Ch}
\DeclareMathOperator{\FKM}{FKM}
\DeclareMathOperator{\TK}{TK}

\theoremstyle{plain}
  \newtheorem{theorem}{Theorem}[section]
  \newtheorem{lemma}{Lemma}[section]
  \newtheorem{proposition}{Proposition}[section]

\theoremstyle{definition}
  \newtheorem{definition}{Definition}[section]
  
  \newtheorem{principle}{Principle}[section]

\theoremstyle{remark}
  \newtheorem{remark}{Remark}[section]

\title[Absolute and relative invariants for low-$d$ topological insulators]{Absolute and relative invariants for low-dimensional topological insulators}
\author[F.~Manzoni, D.~Monaco, G.~Peluso]{Federico Manzoni, Domenico Monaco, Gabriele Peluso}
\date{\today}

\usepackage{hyperref}
\hypersetup{colorlinks=true, linkcolor=black, citecolor=black, filecolor=black, urlcolor=black}

\begin{document}

\begin{abstract}
Motivated by the theory of topological insulators, we examine several classification schemes for maps on the $d$-dimensional torus, \(d \leq 2\), with values in the space of projections on some ambient Hilbert space. The (pairs of) projections are required to be compatible with time-reversal, particle-hole and chiral symmetries, as prescribed by the ``ten-fold way'' leading to the Altland--Zirnbauer--Cartan (AZC) symmetry classes. The main purpose of the paper is to compare different notions of equivalence for such maps, namely Murray--von Neumann equivalence, unitary equivalence, and homotopy equivalence.

The analysis and comparison is based on the explicit construction of symmetric bases for the ranges of the projections which satisfy a pseudo-periodicity constraint, namely some vectors in the basis are allowed to pick up a phase when looped around the torus, while all the others have the same periodicity of the projections themselves. In dimension~$1$, symmetric, fully periodic bases can be constructed in all AZC classes. In dimension~$2$, topological obstructions appear and are encoded in the topological non-triviality of the phase picked up by the pseudo-periodic vector(s), being labeled by an integer (Chern number) or a \(\mathbb{Z}_{2}\)-valued index (Fu--Kane--Mele invariant). These invariants characterize Murray--von Neumann and unitary equivalence of projection-valued maps, and are interpreted as absolute invariants.

A finer classification may be needed when considering homotopy equivalence. If the ambient Hilbert space is non-minimal, i.e., if the pair of projections together do not span the whole ambient space, homotopy classification reduces to unitary equivalence. By contrast, in ambient spaces of minimal dimension, additional homotopy invariants may arise. These are relative invariants, as their precise value may change under a unitary conjugation: within the same unitary representation, only the relative topological phase between different (pairs of) projection-valued maps is well-defined.

The results provide a unified geometric framework for understanding topological phases of gapped quantum systems in low dimension, clarifying how different equivalence relations result in different classification schemes and, correspondingly, topological invariants.
\end{abstract}

\maketitle

\setcounter{tocdepth}{1}
\tableofcontents

\section{Introduction}

The mathematical description of topological phases of matter is one of the central meeting points between geometry, topology, and quantum theory of matter~\cite{bernevig2013,chiu2016classification}. As briefly reviewed in Appendix~\ref{app:A} (see references therein), in lattice-periodic quantum systems the Bloch--Floquet--Zak representation fibers the Hamiltonian into a family of self-adjoint operators depending continuously on quasi-momentum \(k \in \mathbb{T}^{d}\), where \(\mathbb{T}^{d}\) is the Brillouin torus%
\footnote{It will be often convenient to identify $\T^d$ with the hypercube $[-\pi,\pi]^d \subset \mathbb{R}^{d}$ with periodic boundary conditions, i.e.\ with opposite sides identified.}%
, and acting on a fiber Hilbert space $\mathcal{H}$, accounting for the degrees of freedom in the fundamental cell of the crystal. The Hamiltonian could exhibit further symmetries, like time-reversal, particle-hole or chiral symmetry: their possible combinations lead to the division in ten Altland--Zirnbauer--Cartan (AZC) symmetry classes~\cite{tab,heinzner2005symmetry,Kitaev_2009,ryu2010topological} (cf.\ Table~\ref{tabular:AZC_classes} below). The presence of spectral gaps allows to select the relevant eigenprojections of the fiber Hamiltonians, leading to a continuous%
\footnote{Continuity of projection-valued maps is meant in the operator-norm topology of $\mathcal{B}(\mathcal{H})$, and will always be assumed, even if not stated explicitly.} %
projection-valued map $P \colon \mathbb{T}^{d} \to \mathcal{B}(\mathcal{H})$, or possibly, in presence of a chiral or particle-hole symmetry, to a pair of continuous projection-valued maps $P^{\pm} \colon \mathbb{T}^{d} \to \mathcal{B}(\mathcal{H})$.

The topology of these maps encodes stable global information which cannot be detected by local perturbative arguments~\cite{chiu2016classification, Kohmoto1985, thouless1982quantized}. These robust properties are mathematically described by homotopy invariants: indeed, as long as the spectral gap remains open, continuous deformations of the Hamiltonian induce homotopies of the corresponding projection-valued maps, and topological invariants remain unchanged. A change of such topological labels is therefore interpreted physically as the signature of a transition between distinct topological phases necessarily accompanied by the closing of the gap, as in truncated systems~\cite{prodan2016bulk} or junctions~\cite{gontier2022symmetric, gontier2025topological}.

The purpose of this work is to develop a unified and explicit framework for the classification of projection-valued maps on tori of dimension \(d \leq 2\) in all AZC classes. It turns out that classification under homotopy equivalence is deeply intertwined with unitary equivalence (which, in the low-dimensional cases we consider, is in turn the same as the {\it a priori} weaker notion of Murray--von Neumann equivalence; see below for a discussion on this point). The relation between these notions of equivalence, which in our viewpoint is often disregarded or at least left implicit in the literature, rests on a subtle interplay between the rank of the projection-valued maps and the dimension of the ambient Hilbert space. In a nutshell, the presence of ``extra'' dimensions in the Hilbert space can be exploited to continuously deform pairs of projections which would be homotopically distinct if their ranges spanned the whole ambient. As will be explained below, this helps to clarify the distinction between \emph{absolute} and \emph{relative} topological phases: absolute invariants do not change under homotopies \emph{and} unitary equivalences, while relative invariants are constant under homotopic deformations but may change under a unitary equivalence. It is worth noting that the addition of extra topologically trivial bands is exactly what leads to the (reduced) $K$-theoretic classification under so-called stable equivalence of projections, proposed in the pioneering work by Kitaev~\cite{Kitaev_2009} and leading to the celebrated ``periodic table'' of topological phases of matter. As is well known, homotopy leads to a finer, and thus more complete, classification scheme, in line with the physical lore that ``topological properties of the underlying system are stable under small (continuous) deformations''.

\subsection{Symmetric projection-valued maps} \label{sec:Symm_proj}

Before discussing in detail our results, we formulate the mathematical setting more precisely. The objects under investigation will always be projection-valued maps $P:\T^d \to \Proj_n(\Hi)$ or (ordered) pairs of projection-valued maps $P^\pm \equiv (P^{-}, P^{+}) :\T^d\to \Proj_n(\Hi)$, where $d \le 2$ and where 
\[ \Proj_n(\Hi) := \left\{ P \in \mathcal{B}(\mathcal{H}) : P^2=P^*=P, \; \operatorname{rank}(P) := \tr_{\mathcal{H}}(P)=n \right\} \]
denotes the set of rank-$n$ projections on the Hilbert space $\mathcal{H}$, to which we will refer as the \emph{ambient Hilbert space}. The rank could be finite or infinite, provided of course $\mathcal{H}$ has itself infinite dimension. In the case of pairs of projection-valued maps, it is always assumed that they span orthogonal subspaces:
\begin{equation} \label{standing_orthogonality}
P^{-}(k) \, P^{+}(k) = 0 \quad \forall \: k \in \mathbb{T}^{d}.
\end{equation}

For future reference, we set the following
\begin{definition}[(Non-)minimal ambient Hilbert space] \label{def:minimal}
Let $P^{\pm} \colon \mathbb{T}^{d} \to \Proj_n(\mathcal{H})$ be a pair of projection-valued maps. If together they span the whole ambient Hilbert space, that is, if $P^{-}(k) + P^{+}(k) = \Id_{\mathcal{H}}$ for all $k \in \mathbb{T}^{d}$, we say that the ambient Hilbert space $\mathcal{H}$ is \emph{minimal}; otherwise, $\mathcal{H}$ is said to be \emph{non-minimal}, and $P^{-}(k) + P^{+}(k)$ is an orthogonal projection, due to~\eqref{standing_orthogonality}, onto a proper ($k$-dependent) subspace of $\mathcal{H}$. 

For pairs of finite-rank projection-valued maps, i.e.\ if $n < \infty$, we will also say that the ambient Hilbert space has \emph{minimal dimension} if $\dim(\mathcal{H}) = 2n$, and that it has \emph{non-minimal dimension} if $\dim(\mathcal{H}) > 2n$.
\end{definition}

The projection-valued maps under consideration will possibly satisfy further symmetries, as detailed in the following Definition; their physical origin and relevance is discussed in Appendix~\ref{app:A}.
\begin{definition}[Symmetries] \label{def:Symmetries}
A \emph{antiunitary symmetry operator} (or simply \emph{antiunitary symmetry}) is an antiunitary operator on $\mathcal{H}$ which squares to $\pm \Id_{\mathcal{H}}$; if it squares to $+\Id_{\mathcal{H}}$, the symmetry is called \emph{even}, while, if it squares to $-\Id_{\mathcal{H}}$, it is called \emph{odd}. A \emph{unitary symmetry operator} (or simply \emph{unitary symmetry}) is a unitary operator on $\mathcal{H}$ which squares to $+\Id_{\mathcal{H}}$.

\begin{itemize}[leftmargin=*]
    \item A projection-valued map $P:\T^d \to \Proj_n(\Hi)$ is said to be \emph{time-reversal symmetric} if the ambient Hilbert space is endowed with an antiunitary symmetry $T$ and
    \[ T P(k) = P(-k) T \quad \forall \: k \in \mathbb{T}^{d}. \]
    In presence of the same symmetry operator, a pair of projection-valued maps $P^\pm:\T^d\to \Proj_n(\Hi)$ is said to be \emph{time-reversal symmetric} if both $P^{+}$ and $P^{-}$ are separately time-reversal symmetric.
    \item A projection-valued map $P:\T^d\to \Proj_n(\Hi)$ is said to be \emph{particle-hole covariant} if the ambient Hilbert space is endowed with an antiunitary symmetry $C$ and
    \[ C P(k) = P(-k) C \quad \forall \: k \in \mathbb{T}^{d}. \]
    In presence of the same symmetry operator, a pair of projection-valued maps $P^\pm:\T^d\to \Proj_n(\Hi)$ is said to be \emph{particle-hole symmetric} if
    \[ C P^{+}(k) = P^{-}(-k) C \quad \forall \: k \in \mathbb{T}^{d}. \]
    \item A projection-valued map $P:\T^d\to \Proj_n(\Hi)$ is said to be \emph{chiral covariant} if the ambient Hilbert space is endowed with a unitary symmetry $S$ and
    \[ S P(k) = P(k) S \quad \forall \: k \in \mathbb{T}^{d}. \]
    In presence of the same symmetry operator, a pair of projection-valued maps $P^\pm:\T^d\to \Proj_n(\Hi)$ is said to be \emph{chiral symmetric} if
    \[ S P^{+}(k) = P^{-}(k) S \quad \forall \: k \in \mathbb{T}^{d}. \]
\end{itemize}
\end{definition}

\begin{remark}[The involution on the torus and its fixed points] \label{rmk:Fixed_points}
The presence of antiunitary symmetries leads to consider the involution $k \mapsto -k$ on the torus $\mathbb{T}^{d}$. For later reference, we mention here that we will adopt the notation $k_{\star}$ to denote any of the fixed points under this involution: these are the $2^{d}$ points which, under the identification $\mathbb{T}^{d} \simeq [-\pi, \pi]^{d} / (\pi \sim -\pi)$, have coordinates equal to either $0$ or $\pm \pi$. 
\end{remark}

\begin{remark}[Covariance vs symmetry] \label{rmk:Covariance}
It is worth noting that the notions of particle-hole covariance and time-reversal symmetry for a projection-valued map are formally the same. Moreover, if $P \colon \mathbb{T}^{d} \to \Proj_n(\mathcal{H})$ is (particle-hole or chiral) covariant, then so is $P^\perp := \Id_{\mathcal{H}} - P \colon \mathbb{T}^{d} \to \Proj_{\dim(\mathcal{H})-n}(\mathcal{H})$. Notice how the presence of particle-hole or chiral symmetry forces the two elements of a symmetric pair of projection-valued maps to have the same rank.
\end{remark}

If more than one symmetry is present, a certain compatibility condition needs to be imposed, in the form of a commutation or anticommutation relation among the corresponding symmetry operators. First of all, as it is commonly done, we exclude the case in which the (pair of) projection-valued map under scrutiny is symmetric under more than one symmetry of the same type (time-reversal, particle-hole, or chiral). Then, as soon as two different (commuting or anticommuting) symmetries are present, their product is close to a symmetry of the third type. Let us focus on the case in which a time-reversal symmetry operator $T$ and a particle-hole symmetry operator $C$ are both present. It is easily seen that $S:= T C$ has the correct commutation relations with the projection-valued map(s) so that it makes for a good candidate as a chiral symmetry operator. We need only check whether it squares to the identity operator.

Denoting the sign of the squares of the antiunitary symmetries as $T^2=:\epsilon_T \Id$ and $ C^2=:\epsilon_C\Id$, we have the following identity:
\[
    (TC)^2=TCTC=
    \begin{cases}
        T^2C^2=\epsilon_T\epsilon_C\Id & \text{if } \quad [T,C]=0, \\ 
        -T^2C^2=-\epsilon_T\epsilon_C \Id &\text{if} \quad \{T,C\}=0.
    \end{cases}
\]
Let us set $\epsilon := \epsilon_T\epsilon_C \in \{\pm1\}$. 
If $(TC)^2=\Id$, that is, if $[T,C]=0$ and $\epsilon=1$ or if $\{T,C\}=0$ and $\epsilon = -1$, then we conclude that we can consider $S=TC$ as a {\it bona fide} chiral symmetry. 
Instead, if $(TC)^2=-\Id$, we set $\Tilde{T}:=iT$. Notice that $\Tilde{T}$ is still a time-reversal symmetry for the (pair of) projection-valued maps, has $\epsilon_{\Tilde{T}} = \epsilon_{T}$, but is such that $\Tilde{T}C \pm C\Tilde{T} = TC\mp CT$. Therefore, upon replacing $T$ with $\Tilde{T}$, we can consider $S = TC$ as a chiral symmetry under the same constraint as before, namely
\[
    \epsilon = \epsilon_T \epsilon_C = \begin{cases}
        1 & \mbox{if }\quad [T,C]=0; \\
        -1 & \mbox{if } \quad \{T,C\}=0.
    \end{cases}
\]
The sign $\epsilon$ also controls if $S$ commutes or anticommutes with $T$ and $C$, since one easily checks the following relations:
\begin{align*}
\epsilon=1 & \iff  [T,C] = [T,S]=[S,C]=0; \\ 
\epsilon=-1 & \iff \{T,C\}=\{T,S\}=\{S,C\}=0. 
\end{align*}

It is also possible to prove that one can reach the same conclusions starting with independent symmetries $S$ and $T$ as well as with $S$ and $C$. These arguments imply that there are ten AZC symmetry classes with at most two independent symmetries, as summarized in the following Table~\ref{tabular:AZC_classes}.

\begin{table}[!ht]
\centering
\caption{AZC classification into symmetry classes. Each entry in the columns labeled $T$, $C$ or $S$ corresponds to the sign of the square of the symmetry, and is otherwise $0$ if the symmetry is broken; the last column explictly states, in multi-symmetric classes, whether the three symmetry operators $T$, $C$ and $S = TC$ commute or anticommute with each other.} 
\label{tabular:AZC_classes}
\begin{tabular}{||c||c|c|c||c||}
     \hline
     Symmetry class & $T$ & $C$ & $S$ & Commutation/Anticommutation Relation  \\
     \hline
     A    &  $0$ &  $0$ & $0$ & \textit{irrelevant} \\
     AIII &  $0$ &  $0$ & $1$ & \textit{irrelevant} \\
     AI   &  $1$ &  $0$ & $0$ & \textit{irrelevant} \\
     BDI  &  $1$ &  $1$ & $1$ & Commutation \\
     D    &  $0$ &  $1$ & $0$ & \textit{irrelevant} \\
     DIII & $-1$ &  $1$ & $1$ & Anticommutation \\
     AII  & $-1$ &  $0$ & $0$ & \textit{irrelevant} \\
     CII  & $-1$ & $-1$ & $1$ & Commutation \\
     C    &  $0$ & $-1$ & $0$ & \textit{irrelevant} \\
     CI   &  $1$ & $-1$ & $1$ & Anticommutation\\
     \hline
\end{tabular}
\end{table}

Notice that other (equivalent) conventions are available and adopted in the literature \cite{Bernard_2012,Shiozaki_2014, haake2018quantum}
for example, one can always impose that all symmetries commute with one another provided one relaxes the constraint on the sign of the square of the unitary symmetry, or chooses a different combination of the antiunitary symmetries as the unitary one rather than forcing $S=TC$. With our choice of conventions, it is also possible to choose specific normal forms for the symmetry operators, which will be discussed separately for each symmetry class when needed.

\subsection{Main results} \label{sec:Questions}

We are now able to state the questions addressed in this paper, and present our results at least at an informal level.

\subsubsection{Question 1: Bases and Murray--von Neumann equivalence} \label{pro:frames}

We first address the existence of orthonormal bases spanning the range of the projection-valued map(s) under consideration made out of continuous $\mathcal{H}$-valued functions: for pairs of projection-valued maps $P^{\pm} \colon \mathbb{T}^{d} \to \Proj_n(\mathcal{H})$, we consider a collection of $2n$ such $\mathcal{H}$-valued functions, where we require the first $n$ vectors to span the range of $P^-$ and the second $n$ vectors to span instead the range of $P^+$. Such a basis is called \emph{symmetric} if the collection of vectors further satisfies natural compatibility conditions with the symmetries in the AZC class to which the projection-valued map(s) belong, specified as follows (cf.~\cite[Definition 3.2]{manzonimonacopeluso2026zak}).

\begin{definition}[Symmetric basis] \label{def:Symmetric_basis}
Let $P \colon \mathbb{T}^{d} \to \Proj_n(\mathcal{H})$ be a time-reversal symmetric projection-valued map, under a time-reversal symmetry $T$. A basis $\{v_1(k), \ldots, v_n(k)\}$ for the range of $P(k)$ is called (\emph{time-reversal}) \emph{symmetric} if one of the following conditions hold:
\begin{itemize}
    \item if $T^2 = +\Id_{\mathcal{H}}$, then
    \[ T v_j(k) = v_j(-k) \quad \forall\: j \in \{1, \ldots, n\}; \]
    \item if instead $T^2 = - \Id_{\mathcal{H}}$, then there exists a unitary and anti-symmetric matrix $\varepsilon = [\varepsilon_{i,j}]_{1 \le i,j \le n}$ such that
    \[ T v_j(k) = \sum_{i=1}^{n} \varepsilon_{j,i} v_i(-k) \quad \forall\: j \in \{1, \ldots, n\}. \]
\end{itemize}

Let $P^{\pm} \colon \mathbb{T}^{d} \to \Proj_n(\mathcal{H})$ be a particle-hole symmetric pair of projection-valued maps, under a particle-hole symmetry $C$. A basis $\{v_1(k), \ldots, v_n(k), v_{n+1}(k), \ldots, v_{2n}(k)\}$ for the range of $P^{-}(k) + P^{+}(k)$ is called (\emph{particle-hole}) \emph{symmetric} if
\[ C v_j(k) = v_{2n-j+1}(-k) \quad \forall\: j \in \{1, \ldots, n\}. \]

Let $P^{\pm} \colon \mathbb{T}^{d} \to \Proj_n(\mathcal{H})$ be a chiral symmetric pair of projection-valued maps, under a chiral symmetry $S$. A basis $\{v_1(k), \ldots, v_n(k), v_{n+1}(k), \ldots, v_{2n}(k)\}$ for the range of $P^{-}(k) + P^{+}(k)$ is called (\emph{chiral}) \emph{symmetric} if
\[ S v_j(k) = v_{2n-j+1}(k) \quad \forall\: j \in \{1, \ldots, n\}. \]
\end{definition}

The question of whether one can find symmetric bases in the sense above has been extensively studied in the literature~\cite{kohn1959analytic, des1964analytical,  panati2007triviality, marzari2012maximally, cornean2016construction, Cornean_2017, fiorenza2016construction, fiorenza2016z, graf2013bulk, monaco2023topology, marrazzo2024wannier, monaco2023z2, manzonimonacopeluso2026zak}, in view both of their geometric interpretation and of the practical physical consequences that the use of such a basis yields (compare also Appendix~\ref{app:A}). For the time being, it suffices to stress that, if one also wants to enforce the same periodicity of the projection-valued maps on a symmetric basis (i.e., if one wants the $\mathcal{H}$-valued functions constituting the basis to be defined on the torus $\mathbb{T}^{d}$), then one runs in general in a \emph{topological obstruction}. At most, one can rather impose a \emph{pseudo-periodicity constraint}, where the restriction of the basis on one face of the cube and the corresponding restriction on the opposite face must be matched by an appropriate unitary-matrix-valued function, the so-called \emph{matching matrix}. More precisely, while in dimension~$1$ periodic symmetric bases can be constructed in all AZC symmetry classes~\cite[Theorem~3.2]{manzonimonacopeluso2026zak}, genuine topological obstructions appear in dimension~$2$ and are encoded in the homotopy class of the matching matrices: for example, in class A (no simmetries), labeling this homotopy obstruction via a topological index gives rise to the Chern number~\cite{thouless1982quantized, avron1983homotopy, simon1983holonomy, Kohmoto1985,  vonKlitzing1986quantized, BrouderPanatiCalandraMourouganeMarzari2007, panati2007triviality}, while in class AII (odd time-reversal symmetry) it yields the Fu--Kane--Mele $\mathbb{Z}_2$-invariant~\cite{kane2005z,FuKane2006}. In particular, by choosing an appropriate representative in the homotopy class of the matching matrices, it is possible to ``squeeze'' all the topological information (and thus, all the lack of periodicity) in just a few vectors of the basis. Henceforth, by a \emph{pseudo-periodic basis} for a $2$-dimensional projection-valued map we mean one where all vectors are periodic in $k$ except for the some, which are periodic only in the variable $k_2$ and pick up a ($k_2$-dependent) phase%
\footnote{This fails to hold in higher dimensions and more complex, non-Abelian pseudo-periodic conditions could be needed to encode the topological constraints, cf.~\cite{monaco2023topology}.} %
along the loop in the torus in the $k_1$-direction:
\begin{equation} \label{pseudoper} 
v(k_1,\pi) = v(k_1,-\pi), \quad v(\pi,k_2) = e^{i \theta(k_2)} v(-\pi,k_2) \quad \forall \: k = (k_1,k_2) \in \mathbb{T}^{2}.
\end{equation}
We'll see that in all symmetry classes the number of such non-periodic vectors is at most one or two, with the latter case corresponding to one symmetry-related pair. For example, for a pair of $2$-dimensional projection-valued maps, the pseudo-periodic constraint is required to hold on the first vector in the basis for $P^{-}$ and on the last vector in the basis for $P^{+}$, while all other vectors are fully periodic over the torus.

The first main outcome of the paper is the construction of symmetric pseudo-periodic bases in all symmetry classes in dimension \(d = 2\). In particular we show that, in chiral and/or particle-hole symmetric classes, topological obstructions to enforce full periodicity can be quantified by analogous invariants, obtained by focusing on one of the two members of the pair, say \(P^{-}\), or by some further reduction by symmetry. All the resulting $\mathbb{Z}$-valued invariants are thus in the form of Chern numbers, and all $\mathbb{Z}_2$-valued invariants are in the form of Fu--Kane--Mele invariants, see Table~\ref{tabular:AZC_classes_MvN}. All these invariants are meaningful only when the rank of the involved projections is finite, as infinite-rank projections do not retain any topological information~\cite{kuiper1965homotopy}.

\begin{table}[!ht]
\centering
\caption{Invariants characterizing Murray--von Neumann equivalence classes of symmetric (pairs of) projection-valued maps (as specified by the second-to-last column) $P^{(\pm)} \colon \mathbb{T}^{2} \to \Proj_n(\mathcal{H})$, where $n = \operatorname{rank}(P^{(-)}) < \infty$. In each class, the invariants are complete, i.e., two (pairs of) projection-valued maps are Murray--von Neumann equivalent if and only if their invariants, as listed in the last column, agree. Here, $\Ch(\cdot) \in \Z$ is the Chern number (Definition~\ref{def:Chern_number}) and $\FKM(\cdot) \in \Z_2$ the Fu--Kane--Mele invariant (Definition~\ref{def:FMK}) of a projection-valued map; these invariants lose meaning if the projections have infinite rank. Corresponding results for covariant (rather than symmetric) pairs of projection-valued maps can be inferred by this Table and the content of Section~\ref{sec:Covariant_cases}.} 
\label{tabular:AZC_classes_MvN}
\renewcommand{\arraystretch}{1.15}
\begin{tabularx}{\textwidth}{||>{\centering}m{4.5em}||>{\centering}m{1.5em}|>{\centering}m{1.5em}|>{\centering}m{1.5em}||>{\centering\arraybackslash}m{9em}|>{\centering\arraybackslash}X||}
     \hline
     AZC class & $T$ & $C$ & $S$ & Symmetric (pairs of) projection(s) & Murray--von Neumann-equivalence invariants, $d=2$  \\
     \hline
     A    &  $0$ &  $0$ & $0$ & $P$       & \(\operatorname{rank}(P)\in\mathbb N,\quad \operatorname{Ch}(P)\in\mathbb Z\) \\
     AIII &  $0$ &  $0$ & $1$ & $P^{\pm}$ & \(\operatorname{rank}(P^-)\in\mathbb N,\quad \operatorname{Ch}(P^-)\in\mathbb Z\) \\
     AI   &  $1$ &  $0$ & $0$ & $P$       & \(\operatorname{rank}(P)\in\mathbb N\) \\
     BDI  &  $1$ &  $1$ & $1$ & $P^{\pm}$ & \(\operatorname{rank}(P^-)\in\mathbb N\) \\
     D    &  $0$ &  $1$ & $0$ & $P^{\pm}$ & \(\operatorname{rank}(P^-)\in\mathbb N,\quad \operatorname{Ch}(P^-)\in\mathbb Z\) \\
     DIII & $-1$ &  $1$ &  $1$ & $P^{\pm}$ &  \(\operatorname{rank}(P^-)\in 2\mathbb N,\quad \FKM(P^-)\in\mathbb Z_2\) \\
     AII  & $-1$ &  $0$ & $0$ & $P$       & \(\operatorname{rank}(P)\in 2\mathbb N,\quad \FKM(P)\in\mathbb Z_2\) \\
     CII  & $-1$ & $-1$ & $1$ & $P^{\pm}$ & \(\operatorname{rank}(P^-)\in 2\mathbb N,\quad \FKM(P^-)\in\mathbb Z_2\) \\
     C    &  $0$ & $-1$ & $0$ & $P^{\pm}$ & \(\operatorname{rank}(P^-)\in\mathbb N,\quad \operatorname{Ch}(P^-)\in\mathbb Z\) \\
     CI   &  $1$ & $-1$ & $1$ & $P^{\pm}$ &  \(\operatorname{rank}(P^-)\in\mathbb N\) \\
     \hline
\end{tabularx}
\end{table}

The analysis of symmetric pseudo-periodic bases can be rephrased in terms of Murray--von Neumann equivalence of projections. Recall~\cite[Definition~2.2.1]{rordam2000introduction} that two projection-valued maps $P_{0}, P_{1} \colon \mathbb{T}^{d} \to \Proj_n(\mathcal{H})$ are \emph{Murray--von Neumann equivalent} if there exists a partial-isometry-valued map $V \colon \mathbb{T}^{d} \to \mathcal{B}(\mathcal{H})$ such that
\[ P_{0}(k) = V(k)^* \, V(k) \quad \text{and} \quad P_{1}(k) = V(k) \, V(k)^* \quad \forall \: k \in \mathbb{T}^{d}. \]
Observe that, by cyclicity of the trace, two Murray--von Neumann equivalent projection-valued maps necessarily have the same rank. Notice also how the above relation implies
\[ V(k) \, P_{0}(k) = V(k) \, V(k)^* \, V(k) = P_{1}(k) \, V(k) \]
so that the partial isometry $V$ intertwines the ranges of $P_{0}$ and $P_{1}$. An analogous definition can be given for pairs of projection-valued maps $P_0^{\pm}, P_1^{\pm} \colon \mathbb{T}^{d} \to \Proj_n(\mathcal{H})$: the two elements of the two pairs are required to be separately intertwined by partial-isometry-valued maps $V^{\pm} \colon \mathbb{T}^{d} \to \mathcal{B}(\mathcal{H})$, so that in particular $V^{-} + V^{+}$ is a partial isometry between the ranges of $P_0^{-} + P_0^{+}$ and $P_1^{-} + P_1^{+}$. Moreover, one can easily formulate natural symmetry conditions on these partial isometries that make the notion of Murray--von Neumann equivalence compatible with the AZC symmetry classification: explicitly, these read
\begin{equation} \label{MvNsymmetry}
    T\,V(k)=V(-k)\,T, \quad C\,V(k)=V(-k)\,C, \quad S\,V(k)=V(k)\,S
\end{equation}
depending on which of the symmetry operators (if any) are present. Notice how the first two constraints on $V$, those of commutation with an antiunitary operator up to the involution $k \mapsto -k$, are mathematically equivalent (compare Remark~\ref{rmk:Covariance}).

It turns out that periodic partial-isometry-valued intertwiners between two (pairs of) pro\-jec\-tion-valued maps can be constructed by mapping a symmetric pseudo-periodic basis for one projection to a corresponding basis for the other, provided the two bases satisfy the same pseudo-periodic constraint: these two conditions will be shown to be equivalent in Principle~\ref{pri:Periodicity_from_pseudo-periodicity} below. The two classification schemes (admitting symmetric bases with prescribed pseudo-periodicity constraints, and Murray--von Neumann equivalence) thus lead to the same topological phases, which are labeled by what we will call \emph{absolute invariants}. These invariants measure the intrinsic topological content carried by the projection-valued maps themselves. 

The construction of symmetric pseudo-periodic bases, and therefore of Murray--von Neumann equivalences, will be carried out in Section~\ref{sec:Bases}.

\subsubsection{Question 2: Unitary equivalence}\label{pro:unitary_equivalences}

Two (pairs of) projection valued maps $P_{0}, P_{1} \colon \mathbb{T}^{d} \to \Proj_n(\mathcal{H})$ (respectively $P_0^{\pm}, P_{1}^{\pm} \colon \mathbb{T}^{d} \to \Proj_n(\mathcal{H})$) are said to be \emph{unitarily equivalent} if there exists a continuous unitary-valued map $V:\T^d \to \U(\Hi)$ such that
\[
   V(k)P_0(k)=P_1(k)V(k) \quad \text{(respectively } V(k)P_0^\pm(k)=P_1^\pm(k)V(k) \text{)}.
\]
Again, the unitary-valued map is assumed to satisfy the appropriate symmetry constraints dictated by the specific AZC symmetry class, as in~\eqref{MvNsymmetry}. The difference with Murray--von Neumann equivalence is that the intertwiner is assumed to act isometrically on the whole ambient space, and not just on the ranges of the relevant projections. As detailed in Appendix~\ref{app:A}, unitary equivalence is related to a change of representation for the physical system under scrutiny. 

On the one hand, it is clear that two (pairs of) projection-valued maps are unitarily equivalent if and only if they are Murray--von Neumann equivalent and the projections on their orthogonal complements are Murray--von Neumann equivalent as well (compare~\cite[Proposition~2.2.2]{rordam2000introduction}). The question of unitary equivalence is thus solved by looking at the absolute invariants from Table~\ref{tabular:AZC_classes_MvN}, provided we identify the relevant topological information contained in these orthogonal complements: this analysis will be carried out in Section~\ref{sec:unitary_equivalences}. It turns out that the possible additional invariants coming from the orthogonal complements either vanish trivially in case the latter are infinite-dimensional, or they are actually completely determined from the ones for the projection-valued maps themselves. Therefore, the same topological data which classify Murray--von Neumann equivalence presented in Table~\ref{tabular:AZC_classes_MvN} also control unitary equivalence: see Theorem~\ref{thm:unitary_iff_MvN}. This is one of the conceptual simplifications of the paper: although the unitary equivalence condition is {\it a priori} stronger, no genuinely new independent invariants appear at this stage.

On the other hand, the unitary-valued map used to conjugate two equivalent (pairs of) projection-valued maps could carry itself some topological information: this will have some consequences, discussed in Section~\ref{section:Dimerization_ambiguity}, on the labeling of homotopy classes of projection-valued maps (see Question 3 below), and leading to the fact that some homotopy invariants only label \emph{relative topological phases}.

\subsubsection{Question 3: Homotopy equivalence} \label{pro:homotopies}

As discussed above, our main aim is a complete homotopy classification of (pairs of) projection-valued maps in all AZC symmetry classes. Recall that $P_0, P_1 \colon \mathbb{T}^{d} \to \Proj_n(\mathcal{H})$ (respectively $P_{0}^{\pm}, P_{1}^{\pm} \colon \mathbb{T}^{d} \to \Proj_n(\mathcal{H})$) are said to be \emph{homotopically equivalent} if there exists continuous maps $P \colon [0,1] \times \mathbb{T}^{d} \to \Proj_n(\mathcal{H})$ (respectively $P^{\pm} \colon [0,1] \times \mathbb{T}^{d} \to \Proj_n(\mathcal{H})$) such that $P^{(\pm)}(0, \cdot) \equiv P^{(\pm)}_0$, $P^{(\pm)}(1, \cdot) \equiv P^{(\pm)}_1$, and $P^{(\pm)}_t \colon \mathbb{T}^{d} \to \Proj_n(\mathcal{H})$ lies in the same AZC class for all $t \in [0,1]$.

At a very general level, Principle~\ref{pri:unitary_eq_to_homotopy} states that being homotopically equivalent is the same as being unitarily equivalent through a unitary-valued map which can be contracted to one commuting with the first of the two (pairs) of projection-valued maps (compare~\cite[Proposition~2.2.6]{rordam2000introduction}). On the other hand, as was remarked above, the unitary equivalence itself could be topologically non-trivial and thus not connected to a map of this type, leading to a distinction between the classifications by unitary equivalence and homotopy equivalence.

As our main contribution with respect to this question, we highlight a subtle point regarding the dimensionality of the ambient Hilbert space. As anticipated above, when one considers a pair of projection-valued maps whose ranges do not exhaust the whole ambient space, then unitary equivalence and homotopy equivalence coincide in \(d \leq 2\). The reason is that the part of the Hilbert space not spanned by the projections provides enough ``room'' to deform one range into the other. In this case, the classification is completely governed by the same absolute invariants already detected by symmetric pseudo-periodic bases and Murray--von Neumann equivalence.

The situation changes in ambient spaces of minimal dimension (see Definition~\ref{def:minimal}). There, one can define finer homotopy invariants for the pair, depending on the symmetry class: a complete analysis is provided in Section~\ref{sec:Homotopy}, and the results are summarized in Tables~\ref{tabular:AZC_classes_Homotopy_minimal} and~\ref{tabular:AZC_classes_Homotopy_minimal2}. Once again, however, the precise value of a topological label attached to a certain homotopy class can change under unitary equivalence, as prescribed by the statements in Section~\ref{section:Dimerization_ambiguity}. One can only then \emph{compare} different topological phases, and identify univocally the \emph{relative} content of topological information.

\begin{table}[!ht]
\centering
\caption{Invariants characterizing homotopy equivalence classes of symmetric (pairs of) projection-valued maps $P^{(\pm)} \colon \mathbb{T}^{d} \to \Proj_n(\mathcal{H})$, where $d \le 2$, $n = \operatorname{rank}(P^{(-)}) < \infty$, and the ambient Hilbert space $\mathcal{H}$ is minimal, i.e.\ $\dim(\mathcal{H})=2n$. In each class, the invariants are complete, i.e., two pairs of projection-valued maps are homotopically equivalent if and only if their invariants, as listed in the last column, agree; the listed set of invariants may be overcomplete, with all possible redundancies discussed in Section~\ref{sec:Properties_of_the_invariants} (compare Table~\ref{tabular:AZC_classes_Homotopy_minimal2}). Here, $\W(\cdot)$ is the chiral invariant (Definition~\ref{def:W_invariant}), $\PI(\cdot)$ is the even-particle-hole invariant (Definition~\ref{def:P_invariant}), and $\G(\cdot)$ is the DIII invariant (Definition~\ref{def:DIII_invariant}).} 
\label{tabular:AZC_classes_Homotopy_minimal}
\renewcommand{\arraystretch}{1.15}
\begin{tabularx}{\textwidth}{||>{\centering}m{4.5em}||>{\centering\arraybackslash}m{6em}||>{\centering\arraybackslash}m{10em}||>{\centering\arraybackslash}X||}
     \hline
     AZC class & $d=0$ & $d=1$ & $d=2$  \\
     \hline
     A    & $n\in\mathbb N$ & $n\in\mathbb N$ &  $n\in\mathbb N$, $\operatorname{Ch}(P)\in\mathbb Z$ \\
     \vspace{.5em}AIII & $n \in \mathbb{N}$ & $n\in\mathbb N$, $\W(P^{\pm})\in \mathbb Z$ & $n \in \mathbb{N}$, $\W(P^{\pm})\in \mathbb{Z}^{2}$ \\
     AI   & $n\in\mathbb N$ & $n\in\mathbb N$ & $n\in\mathbb N$ \\
     BDI  & $n\in\mathbb N$, $\PI(P^{\pm}) \in \Z_2$ & $n\in\mathbb N$, $\W(P^{\pm}) \in 2\Z$, $\PI(P^{\pm}) \in (\Z_2)^2$ & $n\in\mathbb N$, $\W(P^{\pm}) \in (2\Z)^2$, $\PI(P^{\pm}) \in (\Z_2)^4$ \\
     D    & $n\in\mathbb N$, $\PI(P^{\pm}) \in \Z_2$ & $n\in\mathbb N$, $\PI(P^{\pm}) \in (\Z_2)^{2}$ & $n\in\mathbb N$, $\PI(P^{\pm}) \in (\Z_2)^{4}$, $\operatorname{Ch}(P^-)\in\mathbb Z$ \\
     DIII &  $n\in 2\mathbb N$ & $n\in 2\mathbb N$, $\G(P^{\pm}) \in \Z_2 $ & $n\in 2\mathbb N$, $\G(P^{\pm}) \in (\Z_2)^{3} $ \\
     AII  & $n \in 2\mathbb N$ & $n \in 2\mathbb N$ & $n \in 2\mathbb N$, $\FKM(P)\in\mathbb Z_2$ \\
     CII  & $n\in 2\mathbb N$ & $n\in 2\mathbb N$, $\W(P^{\pm}) \in 2\Z$ & $n\in 2\mathbb N$, $\W(P^{\pm}) \in (2\Z)^2$ \\
     C    & $n\in\mathbb N$ & $n\in\mathbb N$ & $n\in\mathbb N$, $\operatorname{Ch}(P^-)\in\mathbb Z$ \\
     CI   &  $n\in\mathbb N$ & $n\in\mathbb N$ & $n\in\mathbb N$ \\
     \hline
\end{tabularx}
\end{table}

In summary, this work provides a systematic review of several classification schemes for symmetric (pairs of) projection-valued maps over \(\mathbb{T}^{d}\), \(d \leq 2\), across all AZC symmetry classes. The comparison relies on the explicit construction of symmetric pseudo-periodic bases, and achieves a complete characterization of symmetry-preserving Murray--von Neumann equivalences, unitary equivalences and homotopy equivalences. The main conceptual distinction arising from this comparison is between absolute invariants, which identify the obstructions to constructing symmetric periodic bases and intertwiners, and relative invariants, which appear in minimal ambient spaces in chiral and/or particle-hole classes and obstruct homotopies even among unitarily equivalent pairs of projection-valued maps. The analysis gives a geometric interpretation to all these topological indices, makes them accessible via computable data, and clarifies their physical meaning in the study of gapped quantum systems and topological phases of matter. While we study here projection-valued maps on the $d$-dimensional torus with $d\le2$, some of the arguments can be extended to higher dimensions in a clear way, at the expense of further significant effort, which we postpone to future work.

\section{Preliminaries and general Principles}

As highlighted in the Introduction, the goal of the paper is to develop a formalism that allows us to tackle each of the Questions raised in Section~\ref{sec:Questions}, by relating the different equivalence notions formulated above. These relations will be made explicit by relying on the following ``Principles'', which are in fact provable mathematical tools that will be applied as ``black boxes'' throughout the paper. These hold in any dimension, so we formulate them in wide generality without relying on the assumption $d \le 2$.
 
Let us start from Question 1, as stated in Section~\ref{pro:frames}. Assume that $P_0^{(\pm)}, P_1^{(\pm)} \colon \mathbb{T}^{d} \to \Proj_n(\mathcal{H})$ are two (pairs of) projection-valued maps which admit symmetric pseudo-periodic bases $\{v^{0}_1(k), \ldots, v^{0}_{(2)n}(k)\}$ and $\{v^{1}_{1}(k), \ldots, v^{1}_{(2)n}(k)\}$, respectively. (If $d \le 2$, their construction will be carried out in Section~\ref{sec:Bases}.) Assume that the two bases have the \emph{same} pseudo-periodicity constraints (cf.~\eqref{pseudoper}), which we formulate as follows: For each $k_{+}, k_{-} \in [-\pi,\pi]^{d}$ which are identified in the torus $\mathbb{T}^{d}$ (i.e., which lie on opposite faces of the hypercube), there exists a matching unitary matrix $\alpha(k_-)$ such that
\begin{equation} \label{matching}
    v_j^0(k_+)= \sum_{i=1}^{(2)n} \alpha(k_-)_{j,i} v_i^0(k_-) \quad \text{and} \quad v_j^1(k_+)= \sum_{i=1}^{(2)n} \alpha(k_-)_{j,i} v_i^1(k_-) \quad \forall \: j\in\{1,\ldots ,(2)n\}.
\end{equation}
(When $d=2$, the matching matrix $\alpha(-\pi,k_2)$ in~\eqref{pseudoper} is a diagonal matrix with entries $\alpha(-\pi,k_2)_{i,i} = e^{i \theta_i(k_2)}$ corresponding to the non-periodic vectors, and $\alpha(-\pi,k_2)_{j,j}=1$ for all the other $j$'s.) We show that the above condition is equivalent to the two (pairs of) projection-valued maps being Murray--von Neumann equivalent. This is established by the following
\begin{principle}[Periodicity from pseudo-periodicity]\label{pri:Periodicity_from_pseudo-periodicity}
Two symmetric bases satisfying the pseudo-periodic constraint~\eqref{matching} exist if and only if there exists a \emph{periodic} partial-isometry-valued map $V \colon \mathbb{T}^{d} \to \mathcal{B}(\mathcal{H})$ intertwining the ranges of the two (pairs of) projections, which moreover satisfies the symmetry constraints~\eqref{MvNsymmetry}.
\end{principle}
\begin{proof}
Given the two symmetric pseudo-periodic bases, define $V(k)$ as the linear operator such that 
\[ V(k)v_j^0(k) := v_j^1(k), \quad j \in \{1,\ldots,(2)n\}, \]
and $V(k)w:=0$ for any $w$ which is orthogonal to the range of $P_0^{(\pm)}(k)$. By definition $V \colon [-\pi,\pi]^{d} \to \mathcal{B}(\mathcal{H})$ is a partial isometry intertwining the ranges of $P_0^{(\pm)}$ and $P_1^{(\pm)}$; since the bases used to define it are only pseudo-periodic, it remains to be shown that this construction yields a \emph{periodic} partial-isometry valued map. To this end, let $k_-$ and $k_+$ be as in~\eqref{matching}; then we have, for $j \in \{1, \ldots, (2)n\}$,
\begin{align*}
V(k_-)v_j^0(k_+) &= V(k_-) \sum_{i=1}^{(2)n} \alpha_{j,i}(k_-) v_i^0(k_-) = \sum_{i=1}^{(2)n} \alpha_{j,i}(k_-) V(k_-) v_i^0(k_-)= \\ 
& = \sum_{i=1}^{(2)n} \alpha_{j,i}(k_-) v_i^1(k_-) = v_j^1(k_+)= V(k_+) v_j^0(k_+)
\end{align*}
by linearity. Since the two partial isometries $V(k_-)$ and $V(k_+)$ coincide both on the vectors spanning the range of $P_0^{(\pm)}(k_+)$ and on the orthogonal complement, where they are both zero, they are the same linear operator.

It remains to show that the constructed $V$ also satisfies the symmetry constraints~\eqref{MvNsymmetry}. The symmetry conditions on the two bases from Definition~\ref{def:Symmetric_basis} read, in compact form,
\begin{equation}\label{eq:reshuffling_matrix}
    O v_j^a(k) = \sum_{i=1}^{n} \varepsilon_{j,i} v_i^a(\pm k) \quad \forall\: a \in \{0,1\}, \; j\in\{1,\cdots ,n\},
\end{equation}
where $O\in\{T,C,S\}$ is a symmetry operator, $\varepsilon$ is an appropriate unitary matrix, and $\pm$ is an appropriate sign (depending on the unitarity or antiunitarity of $O$). Consequently
\begin{align*} 
    O V(k) v_j^0(k) &= O v_j^1(k) = \sum_{i=1}^{n} \varepsilon_{j,i} v_i^1(\pm k) = \sum_{i=1}^{n} \varepsilon_{j,i} V(\pm k) v_i^0(\pm k) = \\ 
    &= V(\pm k) \sum_{i=1}^{n} \varepsilon_{j,i} v_i^0(\pm k) = V(\pm k) O v_j^0(k).  
\end{align*}
Once again, this identity allows to conclude that $O V(k)$ and $V(\pm k) O$ are the same operator, which is exactly~\eqref{MvNsymmetry}.

Conversely, similar computations yield that the image of a symmetric pseudo-periodic basis for one (pair of) projection-valued map(s) via a symmetric periodic partial-isometry-valued intertwiner gives a symmetric pseudo-periodic basis for the second (pair of) projection-valued map(s) with the same pseudo-periodicity constraint.
\end{proof}

As we move to Questions 2 and 3 from Sections~\ref{pro:unitary_equivalences} and~\ref{pro:homotopies}, namely unitary and homotopy equivalences, we first relate these two notions by means of the following result, which holds both for covariant projection-valued maps and for symmetric (pairs of) projection-valued maps.

\begin{principle}[From unitary equivalence to homotopy, and back] \label{pri:unitary_eq_to_homotopy}
Two (pairs of) projection-valued maps $P_0^{(\pm)}, P_1^{(\pm)} \colon \mathbb{T}^{d} \to \Proj_n(\mathcal{H})$ are homotopically equivalent if and only if they are unitarily equivalent via a symmetric unitary-valued map $V \colon \mathbb{T}^{d} \to \U(\mathcal{H})$, and there exists a continuous unitary-valued map $V \colon [0,1] \times \mathbb{T}^{d} \to \U(\mathcal{H})$ such that, for all $k \in \mathbb{T}^{d}$, $\left[V_0(k),P_0^{(\pm)}(k)\right] \equiv 0$, $V(1,k) \equiv V(k)$, and $V(t,\cdot)$ satisfies the appropriate symmetry constraints~\eqref{MvNsymmetry} for all $t \in [0,1]$.
\end{principle}
\begin{proof}
Assuming that a unitary equivalence as in the statement exists among the two (pairs of) projection-valued maps, then a homotopy between them is easily constructed by setting
\[ P^{(\pm)}(t,k) := V(t,k) P_0^{(\pm)}(k) V(t,k)^{-1}. \]

Conversely, let $P^{(\pm)} \colon [0,1] \times \mathbb{T}^{d} \to \Proj_n(\mathcal{H})$ be a homotopy between the two (pairs of) projection-valued maps. By continuity and compactness of $\mathbb{T}^{d}$, the interval $[0,1]$ can be partitioned in sub-intervals with endpoints $0=:t_0 < t_1 < \ldots < t_{N} < t_{N+1} := 1$ so that
\[ \sup_{t \in [t_{i}, t_{i+1}]} \sup_{k \in \mathbb{T}^{d}} \left\| P^{(\pm)}(t, k) - P^{(\pm)}(t_{i}, k) \right\| < 1 \quad \forall\: i \in \{0, \ldots, N\}. \]
Then, there exists a continuous Kato--Nagy unitary~\cite{kato2013perturbation} intertwining $P^{(\pm)}(t_{i}, k)$ and $P^{(\pm)}(t, k)$ over the interval $t \in [t_i, t_{i+1}]$:
\[ P^{(\pm)}(t, k) = U_i(t, k) \, P^{(\pm)}(t_{i}, k) \, U_i(t, k)^{-1}\,, \quad U_i(t, k) \in \U(\mathcal{H}). \]
The explicit construction of $U_i(t,k)$ from the reference has $U_i(t_i,k) \equiv \Id_{\mathcal{H}}$, but it can be easily modified to accommodate that $\left[U_0(0,k), P_0^{(\pm)}(k)\right] \equiv 0$ for all $k \in \mathbb{T}^{d}$. The Kato--Nagy unitary also inherits the appropriate symmetry constraints from those of the projection-valued maps. Define now, for $k \in \mathbb{T}^{d}$,
\[ V(t, k) := U_i(t, k) \, U_{i-1}(t_{i},k) \, \cdots \, U_0(t_1,k) \quad \text{if } t \in [t_{i}, t_{i+1}]. \]
Then $V \colon [0,1] \times \mathbb{T}^{d} \to \U(\mathcal{H})$ is a continuous unitary-valued map satisfying all the properties claimed in the statement.
\end{proof}

Finally, we characterize a useful decomposition for chiral symmetric pairs of projection-valued maps spanning the entire ambient Hilbert space, that is, such that $P^{-}(k) + P^{+}(k) \equiv \Id_{\mathcal{H}}$: in this situation, the ambient space $\mathcal{H}$ is called minimal, according to Definition~\ref{def:minimal}. The following statement shows how the study of chiral symmetric pairs of projection-valued maps is equivalent to the study of unitary-valued maps.

\begin{principle}[Chiral projections on minimal ambient spaces] \label{pri:chiral_classes_min_dim}
A pair of projection-valued maps $P^{\pm} \colon \mathbb{T}^{d} \to \Proj_n(\mathcal{H})$ over a minimal ambient Hilbert space $\mathcal{H}$ is chiral symmetric if and only if there exists an appropriate ($k$-independent) orthogonal decomposition $\mathcal{H} = \mathcal{H}^{\uparrow} \oplus \mathcal{H}^{\downarrow}$ of $\mathcal{H}$ in which both subspaces $\mathcal{H}^{\uparrow}$ and $\mathcal{H}^{\downarrow}$ have the same dimension, and in which the operator $P^+(k)-P^-(k)$ admits the off-diagonal block decomposition
\begin{equation} \label{chiral_symm_P+-P-} 
P^+(k)-P^-(k)=\begin{pmatrix}
        0 & W(k) \\
        W(k)^* & 0
    \end{pmatrix}  , 
\end{equation}
where $W:\T^d \to \U(\mathcal{H}^{\downarrow} \to \mathcal{H}^{\uparrow})$ is a continuous unitary-valued map.
\end{principle}
\begin{proof}
Let $P^\pm:\T^d\to \Proj_n(\Hi)$ be a pair of projection-valued maps such that $P^-(k)+P^+(k) \equiv \Id_{\mathcal{H}}$, and recall the standing orthogonality assumption~\eqref{standing_orthogonality}. Call $R(k) := P^+(k)-P^-(k)$. Then $R(k)^* = R(k)$ is self-adjoint and
\[ R(k)^2 = P^+(k)^2-P^+(k)\,P^-(k)-P^-(k)\,P^+(k) +P^-(k)^2 = P^+(k)+P^-(k) = \Id_{\mathcal{H}}, \]
so that $R(k)^* R(k) = \Id_{\mathcal{H}}$ and $R(k)$ is also unitary.

Assume now that $P^\pm:\T^d\to \Proj_n(\Hi)$ is a chiral symmetric pair, under a chiral symmetry~$S$. Then
\[ S R(k) = S \left( P^+(k)-P^-(k) \right) = \left( P^-(k)-P^+(k) \right) S = - R(k) S.  
\]
Moreover, since $S^2 = \Id_{\mathcal{H}}$, the chiral symmetry operator can only have eigenvalues $\{-1,+1\}$, and by the above identity $R(k)$ establishes a unitary intertwiner among the corresponding eigenspaces. It follows that
\[ \mathcal{H}^{\uparrow} := \ker(S - \Id_{\mathcal{H}}) \quad \text{and} \quad \mathcal{H}^{\downarrow} := \ker(S + \Id_{\mathcal{H}}) \]
are orthogonal subspaces of $\mathcal{H}$ of the same dimension, and that $R(k)$ is off-diagonal in the decomposition $\mathcal{H} = \mathcal{H}^{\uparrow} \oplus \mathcal{H}^{\downarrow}$, admitting the form~\eqref{chiral_symm_P+-P-}. A simple computation now shows that the unitarity of $R(k)$ reflects in the unitarity of the block $W(k)$, as wanted.

Conversely, given the orthogonal decomposition $\mathcal{H} = \mathcal{H}^{\uparrow} \oplus \mathcal{H}^{\downarrow}$ and a unitary-valued map $W:\T^d \to \U(\mathcal{H}^{\downarrow} \to \mathcal{H}^{\uparrow})$ as in the statement, set
\[ S:=
\begin{pmatrix}
    \Id_{\mathcal{H}^{\uparrow}} & 0 \\ 
    0 & -\Id_{\mathcal{H}^{\downarrow}}
\end{pmatrix} 
\quad \text{and} \quad 
P^\pm(k):=\frac{1}{2}
\begin{pmatrix}
    \Id_{\mathcal{H}^{\uparrow}} & \pm W(k) \\ 
    \pm W^*(k) & \Id_{\mathcal{H}^{\downarrow}}
\end{pmatrix}, \quad k \in \mathbb{T}^{d}. \]
It is easy to check then that $S$ is a chiral symmetry and $P^{\pm} \colon \mathbb{T}^{d} \to \Proj_{n}(\mathcal{H})$ is chiral symmetric pair of projection-valued maps.
\end{proof}

In $K$-theoretic approaches to topological insulators, it is essentially the above Principle that is invoked in chiral symmetric classes to justify the passage from the reduced $K_0$-group of projection-valued maps up to stable equivalence to the reduced $K_1$-group of unitary-valued maps up to homotopy equivalence~\cite{prodan2016bulk}; the fact is used also in functional-analytic approaches to topological quantum matter, in the analysis of the so-called ``spectrally-flattened Hamiltonian''~\cite{graf2018bulk, chung2025topological, chung2024topological, chung2026topological}. For us, the Principle will be useful to address Question 3 from Section~\ref{pro:homotopies} in chiral symmetric classes. Notice that additional symmetries on the pair of projection-valued maps will reflect on additional symmetry constraints on the unitary-valued map $W$ in the above statement: these will be discussed class by class.

\section{Symmetric bases and Murray--von Neumann equivalence} \label{sec:Bases}

We tackle now Question 1 raised in Section~\ref{pro:frames}, concerning the construction of symmetric pseudo-periodic bases for (pairs of) projection-valued maps. We address this question class by class.

\subsection{Class A} \label{sec:Bases_class_A}

The analysis of projection-valued maps in absence of any further symmetry is obviously the first that has been treated in the literature, and led to the realization that they host interesting topological phases: see e.g.~\cite{thouless1982quantized, Kohmoto1985, avron1983homotopy, avron1985quantization, simon1983holonomy}. We sketch the relevant constructions here for the sake of completeness, referring the reader to~\cite{monaco2023topology, Peluso2026} and references therein for more details.

Let $P \colon \mathbb{T}^{d} \to \Proj_n(\mathcal{H})$ be a projection-valued map. If $d=0$, any orthonormal basis $\{v_1, \ldots, v_n\}$ of the range of $P\equiv P(0)$ will do. If $d=1$, this basis can be extended to $\mathbb{T}^{1}$ as follows: the Kato--Nagy construction sketched in the proof of Principle~\ref{pri:unitary_eq_to_homotopy} yields a unitary-valued map $U \colon [-\pi,\pi] \to \U(\mathcal{H})$ intertwining the ranges of $P(0)$ and $P(k)$ for any $k \in [-\pi,\pi]$. In turn, this unitary allows to define 
\[ \tilde{v}_j(k) := U(k) \, v_j \quad \forall\: j \in \{1,\ldots,n\} \]
as an orthonormal basis for the range of $P(k)$. However, since this unitary-valued map fails in general to be periodic, this definition leads to a \emph{holonomy unitary} $U(-\pi)^{-1} \, U(\pi)$ comparing the bases at $k=-\pi$ and at $k=\pi$, and in turn to the matching matrix $\alpha = [\alpha_{i,j}]_{1 \le i,j \le n} \in U(n)$ given by
\begin{equation}\label{eq:standard_alpha_d=1} \alpha_{i,j} := \inn{\tilde{v}_i(-\pi)}{\tilde{v}_j(\pi)} = \inn{v_i}{U(-\pi)^{-1} \, U(\pi) \, v_j}, \quad i,j \in \{1,\ldots, n\}. \end{equation}
Since the holonomy unitary commutes with $P(0)$, it can be written as $e^{i L/2\pi}$ with $L^*=L$ also commuting with $P(0)$, e.g.\ using an appropriate spectral decomposition. The basis
\begin{equation} \label{eq:modified_parallel}
v_j(k) := U(k) \, e^{-ikL} v_j \quad \forall\: j \in \{1,\ldots,n\}
\end{equation}
then yields the desired periodic basis.

If $d=2$, one can start from a periodic basis for the restriction $P \big|_{\{k_1=0\}} \colon \mathbb{T}^{1} \subset \mathbb{T}^{2} \to \Proj_n(\mathcal{H})$ and perform a similar construction. This time, a $k_2$-dependent matching matrix $\alpha \colon \mathbb{T}^{1} \to U(n)$ emerges:
\begin{equation}\label{eq:standard_alpha}
    \left[\alpha(k_2)\right]_{i,j}:=\inn{v_i}{U(-\pi,k_2)^{-1} U(\pi,k_2)v_j}, \quad i,j \in \{1,\ldots, n\}, \; k_2 \in \mathbb{T}^1,
\end{equation}
measuring the lack of periodicity of an extended basis in the $k_1$-direction. The possibility to deform these matrices to the identity, that is, ultimately, the study of homotopy classes of such unitary-valued maps, is what hinders the existence of periodic bases, and leads to introduce the following notion.

\begin{definition}[Chern number]\label{def:Chern_number}
Given a \emph{finite-rank} projection-valued map $P:\T^2 \to \Proj_n(\Hi)$, $n < \infty$, its \emph{Chern number} is defined as 
\[ \Ch(P) := [\det (\alpha(\cdot))] \in \Z \]
where $[f]$ denotes the winding number of the function $f \colon \mathbb{T}^{1} \simeq S^1 \to U(1) \simeq S^1$.
\end{definition}

\begin{remark}[On the definition and properties of the Chern number] \label{rmk:Chern}
That this definition of the Chern number reproduces the well-known one given in terms of the Berry curvature associated with the projection-valued map $P$ is the content of~\cite[Proposition 5.3]{cornean2019parseval}, which shows in particular that $\Ch(P)$ is independent of the choices of the basis for the range of $P \big|_{\{k_1=0\}}$ and of the particular unitary used to extend the basis to the whole torus. Moreover, this integer remains constant throughout Murray--von Neumann equivalences, unitary equivalences and homotopy equivalences of the underlying projection-valued map, and is additive, meaning that if $P(k)Q(k)\equiv0$ then $\Ch(P+Q)=\Ch(P)+\Ch(Q)$: we refer the reader e.g.\ to~\cite{Peluso2026} for proofs of these statements.
\end{remark}

\begin{proposition}[Class-A symmetric pseudo-periodic bases, $d=2$] \label{prop:Class_A_frames}
Given a finite-rank projection-valued map $P:\T^2 \to \Proj_n(\Hi)$, $n<\infty$, it is always possible to construct a pseudo-periodic basis in which 
\[ v_1(\pi,k_2) = e^{i \, \Ch(P) \, k_2} \, v_1(-\pi,k_2) \quad \forall\: k_2 \in \mathbb{T}^{1} \]
(cf.~\eqref{pseudoper}).

If instead the projection-valued map has infinite rank, then a periodic basis always exists.
\end{proposition}
For proofs and further considerations, we refer again to~\cite{Peluso2026} and references therein. The case of infinite-rank projections is ultimately a consequence of Kuiper's theorem~\cite{kuiper1965homotopy} on the contractibility of the group of unitaries on an infinite-dimensional Hilbert space.

\bigskip

Invoking Principle~\ref{pri:Periodicity_from_pseudo-periodicity}, the previous results and considerations imply that two projection-valued maps with equal finite rank in dimension $d=1$ are always Murray--von Neumann equivalent, while in dimension $d=2$ they are equivalent if and only if their Chern numbers coincide. If the rank is infinite, then periodic bases always exist, and any two such maps are Murray--von Neumann equivalent, both in $d=1$ and in $d=2$.

\subsection{Classes AI and AII} \label{sec:Bases_class_AIAII}

Time-reversal symmetric classes AI and AII were the first AZC classes to be thoroughly investigated in relation to the construction of symmetric bases and the topological obstructions to enforcing periodicity~\cite{panati2007triviality, fiorenza2016construction, fiorenza2016z, monaco2015symmetry, graf2013bulk, de2014classification, de2015classification, Cornean_2017}. We refer the reader to these references for more details.

Given an even time-reversal symmetry operator $T$ on $\mathcal{H}$, it makes sense to consider the invariant subspace $\mathcal{H}_{\mathbb{R}}$ of vectors fixed by the action of $T$. This is clearly a real Hilbert space, and induces a splitting of $\mathcal{H}$ as $\mathcal{H} = \mathcal{H}_{\mathbb{R}} \otimes_{\mathbb{R}} \mathbb{C} = \mathcal{H}_{\mathbb{R}} \oplus i \mathcal{H}_{\mathbb{R}}$. In particular, given a real orthonormal basis $\{w_j\}_{j \in \{1, \ldots, \dim(\mathcal{H})\}}$ for $\mathcal{H}_{\mathbb{R}}$, it can be regarded as a complex orthonormal basis of $\mathcal{H}$, on which  the time-reversal symmetry operator acts as the standard complex conjugation operator:
\begin{equation} \label{eq:TRS_normal_form_even}
T = \mathcal{K}, \quad T^2 = \Id_{\mathcal{H}}.
\end{equation}
A similar scheme holds for any time-reversal invariant subspace of $\mathcal{H}$, which can be thus regarded as the complexification of the real Hilbert subspace of fixed vectors under the action of $T$.

In the case of class-AI projection-valued maps, with the previous considerations one can always enforce time-reversal symmetry on a periodic basis. This is because, in dimensions $d\in\{0,1,2\}$, one can reproduce the construction from class A but restricting to vectors within the range of the projection $P(0)$ which are left invariant by the action of $T$. No topological obstruction arises, not even in dimension $2$: if the projection-valued map has finite rank, then time-reversal symmetry (of even or odd type) annihilates the Chern number~\cite{panati2007triviality, monaco2015symmetry}, while, if its rank is infinite, one can invoke again Kuiper's theorem~\cite{kuiper1965homotopy}, which holds on real Hilbert spaces as well, to contract the time-reversal symmetric loop of holonomy unitaries to the identity. This well-known construction is presented in some detail in Appendix~\ref{app:XYZ} for the infinite-rank situation. Ultimately, one has the following result (cf.~\cite{cornean2016construction, fiorenza2016construction}).

\begin{proposition}[Class-AI symmetric pseudo-periodic bases, $d=2$] \label{prop:Class_AI_frames}
Given an even-time-reversal symmetric projection-valued map $P:\T^2 \to \Proj_n(\Hi)$, it is always possible to construct a time-reversal symmetric periodic basis for $P$.
\end{proposition}

Furthermore, as a consequence of Principle~\ref{pri:Periodicity_from_pseudo-periodicity}, any two projection-valued maps in class~AI are Murray--von Neumann equivalent.

\bigskip

The situation is considerably more interesting in class AII, where an odd time-reversal symmetry is present. The basic observation is that an antiunitary symmetry $T$ such that $T^2 = - \Id_{\mathcal{H}}$ maps a vector $v \in \mathcal{H}$ into an \emph{orthogonal} vector $Tv \perp v$: the pair $\{v, Tv\}$ is then often referred to as a \emph{Kramers pair}. As a consequence, when endowed with $T$, the Hilbert space $\mathcal{H}$ can be regarded as a quaternionic Hilbert space, and the operator $J := \mathcal{K} T$, which is unitary and squares to $- \Id_{\mathcal{H}}$, can be represented in an appropriate basis as the standard symplectic ``matrix'': in other words, in the orthogonal splitting $\mathcal{H} = \ker(J + i \Id_{\mathcal{H}}) \oplus \ker(J - i \Id_{\mathcal{H}})$, the time-reversal symmetry operator acts as
\begin{equation} \label{eq:TRS_normal_form_odd}
T = \mathcal{K} J \quad \text{with} \quad J = \begin{pmatrix}
    0 & \Id \\
    -\Id & 0
\end{pmatrix}, \quad T^{2} = - \Id_{\mathcal{H}}. 
\end{equation}
Similar considerations lead to Definition~\ref{def:Symmetric_basis} of a symmetric basis for a time-reversal symmetric projection-valued map $P \colon \mathbb{T}^{d} \to \Proj_n(\mathcal{H})$, where one can take $\varepsilon := \begin{pmatrix} 0 & 1 \\ -1 & 0 \end{pmatrix}^{\oplus n}$ to be also a symplectic $(n\times n)$-matrix; notice in particular that, due to Kramers pairing, the rank $n$ of the projections is necessarily even (or infinite) in this case.

Concerning symmetric bases, the modified Kato--Nagy construction from class A can be again replicated for class AII in $d=1$, as the unitary used to extend a symmetric basis of the range of $P(0)$ can be required to be compatible with time-reversal symmetry in the sense of~\eqref{MvNsymmetry}: see e.g.~\cite{Peluso2026} and references therein. In $d=2$, again one can invoke the quaternionic Hilbert space version of Kuiper's theorem~\cite{kuiper1965homotopy} to get rid of the holonomy unitary of infinite-rank projection valued maps, and produce a time-reversal symmetric periodic basis of the form $\{v_{2j+1}(k), v_{2j+2}(k) := - T v_{2j+1}(-k)\}_{j \in \mathbb{N}}$ (compare again Appendix~\ref{app:XYZ}). Instead, for finite-rank projection-valued maps, while there is no Chern-number obstruction as was previously mentioned, a finer topological constraint appears. To identify it, we first note that the matching matrix $\alpha \colon \mathbb{T}^{1} \to U(n)$, $n < \infty$, now satisfies a further compatibility condition with time-reversal symmetry given by 
\begin{equation}\label{eq:matching_matrices_symmetric_constraint}
\alpha(k_2)^{\mathrm{t}} J = J\alpha(-k_2) \iff [J \alpha(k_2)]^{\mathrm{t}} = - J\alpha(-k_2) \quad \forall\: k_2 \in \mathbb{T}^{1}.
\end{equation}
This is indeed the condition which forces $\Ch(P) = 0$, as $k_2 \mapsto \det(\alpha(k_2)) = \det(\alpha(-k_2))$ is an even function and thus has no winding. In order to extract the finer topological information, we make the following general considerations, also for future reference.

\begin{remark}[Teo--Kane invariant]\label{def:gamma_invariant}
Consider a unitary-valued map $W:\T^1 \to U(n)$, with $n \in 2\mathbb{N}$ even, such that $W(k)^{\mathrm{t}} = -W(-k)$ for all $k \in \mathbb{T}^{1}$. In particular, the constraint forces $W(0)$ and $W(\pi) = W(-\pi)$ to be skew-symmetric, so they admit a Pfaffian $\Pf(W(k_{\star}))$ such that $\Pf(W(k_{\star}))^2 = \det(W(k_{\star}))$, $k_{\star} \in \{0,\pi\}$. Choose two real numbers $\lambda(0), \lambda(\pi) \in \R$ such that $\Pf (W(k_{\star})) = e^{i\lambda(k_{\star})}$, $k_{\star} \in \{0,\pi\}$. Let $\mu \colon [-\pi,\pi] \to \mathbb{R}$ be a continuous function%
\footnote{The condition $W(k)^{\mathrm{t}} = -W(-k)$ implies that $\det(W(k)) = (-1)^n \det(W(-k)) = \det(W(-k))$, in view of the parity of $n \in 2 \mathbb{N}$. Consequently, the winding number of the map $\mathbb{T}^{1} \simeq S^1 \to U(1) \simeq S^1$, $k \mapsto \det(W(k))$, vanishes. Thus, the argument of the phase $\det(W(k)) = e^{i \mu(k)}$ can actually be chosen to be a continuous \emph{and periodic} function of $k \in \mathbb{T}^{1}$, i.e., $\mu(-\pi) = \mu(\pi)$.} %
such that $\det(W(k))=e^{i\mu(k)}$ for all $k \in [-\pi,\pi]$; in view of the defining property of the Pfaffian, we must have $e^{i \mu(k_{\star})} = e^{i 2 \lambda(k_{\star})}$, or $\mu(k_{\star}) \equiv 2 \lambda(k_{\star}) \bmod 2\pi\mathbb{Z}$, $k_{\star} \in \{0, \pi\}$.

We define the \emph{Teo--Kane invariant} of such $W$ (see~\cite[Equation~(4.27)]{teo2010topological},~\cite[Equation~(3.70)] {chiu2016classification} and~\cite[Theorem~3.5]{de2022cohomology}) as 
\begin{equation}\label{eq:delta_invariant}
    \TK(W(\cdot)):= e^{i [\mu(0)-2\lambda(0)]/2} e^{i [\mu(\pi)-2\lambda(\pi)]/2} = \prod_{k_{\star} \in \{0,\pi\}} \frac{\sqrt{\det(W(k_{\star}))}}{\Pf(W(k_{\star}))} \in \{1,-1\} \simeq \Z_2,
\end{equation}
where the choice of the square-root on the right-hand side is performed continuously by setting $\sqrt{\det(W(k))}:=e^{i \mu(k)/2}$, $k \in [-\pi,\pi]$.
\end{remark}

\begin{definition}[Fu--Kane--Mele invariant]\label{def:FMK}
Given an odd-time-reversal symmetric \emph{finite-rank} projection-valued map $P:\T^2 \to \Proj_n(\Hi)$, $n < \infty$, its \emph{Fu--Kane--Mele invariant} is defined as 
\[ \FKM(P) := \TK(J \alpha(\cdot)) \in \Z_2. \]
\end{definition}

\begin{remark}[On the definition and properties of the Fu--Kane--Mele invariant] \label{rmk:FKM}
Section 5 in~\cite{Peluso2026} discusses how this definition of the Fu--Kane--Mele invariant reproduces the original one proposed by~\cite{KaneMele2005, FuKane2006}, as well as how it compares with many other equivalent reformulations appearing in the mathematical literature (cf.\ also~\cite{Cornean_2017}). This shows in particular that $\FKM(P)$ is independent of the choices made to define the matching matrix and the associated Teo--Kane invariant, and actually characterizes the homotopy class of the matching matrix within deformations which preserve the constraint~\eqref{eq:matching_matrices_symmetric_constraint}. Moreover, this $\Z_2$-valued quantity remains constant throughout Murray--von Neumann equivalences, unitary equivalences and homotopy equivalences of the underlying projection-valued map, and is multiplicative, meaning that if $P(k)Q(k)\equiv0$ then $\FKM(P+Q)=\FKM(P)\cdot\FKM(Q)$: we refer the reader again to~\cite{Peluso2026} and references therein for proofs of these statements.
\end{remark}

\begin{proposition}[Class-AII symmetric pseudo-periodic bases, $d=2$] \label{prop:Class_AII_frames}
Given an odd-time-reversal symmetric finite-rank projection-valued map $P:\T^2 \to \Proj_n(\Hi)$, it is always possible to construct a time-reversal symmetric pseudo-periodic basis in which 
\[ v_1(\pi,k_2) = e^{i \, \ell \, k_2} \, v_1(-\pi,k_2) \quad \text{and} \quad v_2(\pi,k_2) = e^{-i \, \ell \, k_2} \, v_2(-\pi,k_2) \quad \forall\: k_2 \in \mathbb{T}^{1}, \]
where $\ell \in \mathbb{Z}$ is any integer such that $\ell \equiv \FKM(P) \bmod 2$.

If instead the projection-valued map has infinite rank, then a symmetric periodic basis always exists.
\end{proposition}
For the proof, we refer again to~\cite{Peluso2026} and references therein.

\bigskip

Invoking Principle~\ref{pri:Periodicity_from_pseudo-periodicity}, the previous results and considerations imply that two odd-time-reversal symmetric projection-valued maps with equal finite rank in dimension $d=1$ are always Murray--von Neumann equivalent, while in dimension $d=2$ they are equivalent if and only if their Fu--Kane--Mele invariants coincide. If the rank is infinite, then symmetric periodic bases always exist, and any two such maps are Murray--von Neumann equivalent, both in $d=1$ and in $d=2$.

\subsection{Classes AIII, C, and D} \label{sec:Bases_class_AIIICD}

In classes AIII, C, and D, only a single chiral symmetry, odd particle-hole symmetry, or even particle-hole symmetry, respectively, is present. The relevant object to study is a symmetric pair of projection-valued maps $P^{\pm} \colon \mathbb{T}^{d} \to \Proj_n(\mathcal{H})$. The presence of a particle-hole (respectively chiral) symmetry allows to reconstruct one element of the pair from the other: given $P^-$, the projection $P^+$ is recovered as $P^+(k) = C^{-1} P^{-}(-k) C$ (respectively $P^+(k) = S^{-1} P^{-}(k) S$). Consequently, we can focus on $P^-$ and treat it as an independent projection-valued map, which in these classes carries no other symmetry constraint: as such, $P^- \colon \mathbb{T}^{d} \to \Proj_n(\mathcal{H})$ can be regarded as an element in class A. If a pseudo-periodic basis $\{v_1(k), \ldots, v_n(k)\}$ for the range of $P^{-}(k)$ can be found, then applying the relevant symmetry $C$ or $S$ to each vector allows to define a pseudo-periodic basis for the range of $P^{+}(k)$, given by
\[ v_{j}(k) := C v_{2n-j+1}(-k) \quad \text{(respectively } v_{j}(k) := S v_{2n-j+1}(k) \text{)} \quad \forall\: j \in \{n+1, \ldots, 2n\}. \]
By construction, the full collection $\{v_{j}(k)\}_{j \in \{1,\ldots,2n\}}$ constitutes a symmetric pseudo-periodic basis for the pair $P^{\pm}$, in the sense of Definition~\ref{def:Symmetric_basis}.

In view of the results from Section~\ref{sec:Bases_class_A} above, we can conclude that in $d=1$ all pairs of projection-valued maps in classes AIII, C and D admit a symmetric periodic basis, while in $d=2$ the following statement holds.

\begin{proposition}[Class-AIII, class-C and class-D symmetric pseudo-periodic bases, $d=2$] \label{prop:Class_AIIICD_frames}
Given a particle-hole symmetric pair of finite-rank projection-valued maps $P^{\pm}:\T^2 \to \Proj_n(\Hi)$, it is always possible to construct a particle-hole symmetric pseudo-periodic basis in which 
\[ v_1(\pi,k_2) = e^{i \, \Ch(P^-) \, k_2} \, v_1(-\pi,k_2) \quad \text{and} \quad   v_{2n}(\pi,k_2) = e^{-i \, \Ch(P^-) \, k_2} \, v_{2n}(-\pi,k_2)\quad \forall\: k_2 \in \mathbb{T}^{1}. \]

Given a chiral symmetric pair of finite-rank projection-valued maps $P^{\pm}:\T^2 \to \Proj_n(\Hi)$, it is always possible to construct a chiral symmetric pseudo-periodic basis in which 
\[ v_1(\pi,k_2) = e^{i \, \Ch(P^-) \, k_2} \, v_1(-\pi,k_2) \quad \text{and} \quad   v_{2n}(\pi,k_2) = e^{i \, \Ch(P^-) \, k_2} \, v_{2n}(-\pi,k_2)\quad \forall\: k_2 \in \mathbb{T}^{1}. \]

If instead the projection-valued maps have infinite rank, then a symmetric periodic basis always exists.
\end{proposition}

\bigskip

Invoking Principle~\ref{pri:Periodicity_from_pseudo-periodicity}, the previous results and considerations imply that, in the classes under consideration, two symmetric pairs of finite-rank projection-valued maps $P_0^{\pm}, P_1^{\pm} \colon \mathbb{T}^{d} \to \Proj_n(\mathcal{H})$ in dimension $d=1$ are always Murray--von Neumann equivalent, while in dimension $d=2$ they are equivalent if and only if $\Ch(P_0^{-}) = \Ch(P_1^{-})$. If the rank is infinite, then symmetric periodic bases always exist, and any two such pairs of maps are Murray--von Neumann equivalent, both in $d=1$ and in $d=2$.

\subsection{Classes BDI, CI, CII, and DIII} \label{sec:Bases_class_BDIetal}

In these classes, all three symmetries $T$, $C$ and $S$ are present, and commute or anticommute as prescribed by Table~\ref{tabular:AZC_classes}. In particular, the presence of particle-hole and chiral symmetries imply that the considerations at the start of Section~\ref{sec:Bases_class_AIIICD} apply: given a symmetric pair of projection-valued maps $P^{\pm} \colon \mathbb{T}^{d} \to \Proj_n(\mathcal{H})$, one can recover $P^+$ (and its symmetric pseudo-periodic bases) from $P^-$ (and its symmetric pseudo-periodic bases). As before, we therefore focus on $P^-$: this time, further constraints come from time-reversal symmetry, as the latter projection-valued map is an element of class AI if the original pair is in classes BDI or CI (with an even time-reversal symmetry), respectively an element of class AII if the original pair is in classes CII or DIII (with an odd time-reversal symmetry). In particular, in all classes under consideration, the Chern number of $P^{-}$ vanishes if the latter has finite rank.

Combining the results from the previous Sections~\ref{sec:Bases_class_AIAII} and~\ref{sec:Bases_class_AIIICD}, we conclude that all pairs of projection-valued maps $P^{\pm} \colon \mathbb{T}^{d} \to \Proj_n(\mathcal{H})$ admit a symmetric pseudo-periodic basis if $d=1$, while if $d=2$ then the following result holds.

\begin{proposition}[Class-BDI, class-CI, class-CII and class-DIII symmetric pseudo-periodic bases, $d=2$] \label{prop:Class_BDIetal_frames}
Given a symmetric pair of projection-valued maps $P^{\pm}:\T^2 \to \Proj_n(\Hi)$ in classes BDI or CI, it is always possible to construct a symmetric periodic basis for $P^{\pm}$.

Given a symmetric pair of finite-rank projection-valued maps $P^{\pm}:\T^2 \to \Proj_n(\Hi)$ in classes CII or DIII, it is always possible to construct a symmetric pseudo-periodic basis in which 
\[\begin{matrix} v_1(\pi,k_2) = e^{i \, \ell \, k_2} \, v_1(-\pi,k_2), \quad & v_2(\pi,k_2) = e^{-i \, \ell \, k_2}v_2(-\pi,k_2),  \\ v_{2n}(\pi,k_2) = e^{i \, \ell \, k_2} \, v_{2n}(-\pi,k_2), & \mbox{and} \quad v_{2n-1}(\pi,k_2) = e^{-i \, \ell \, k_2} \, v_{2n-1}(-\pi,k_2) \end{matrix}\quad \forall\: k_2 \in \mathbb{T}^{1},\]
where $\ell \in \mathbb{Z}$ is any integer such that $\ell \equiv \FKM(P^{-}) \bmod 2$.

If instead the projection-valued maps have infinite rank, then a symmetric periodic basis always exists.
\end{proposition}

\bigskip

Invoking Principle~\ref{pri:Periodicity_from_pseudo-periodicity}, the previous results and considerations imply that any two symmetric pairs of projection-valued maps $P_0^{\pm}, P_1^{\pm} \colon \mathbb{T}^{d} \to \Proj_n(\mathcal{H})$ are Murray--von Neumann equivalent, both in $d=1$ and in $d=2$, if they belong to classes BDI or CI. Instead, in classes CII or DIII, two such pairs are always Murray--von Neumann equivalent in $d=1$, while in dimension $d=2$ they are equivalent:
\begin{itemize}
    \item always, if their rank is infinite;
    \item if and only if $\FKM(P_0^{-}) = \FKM(P_1^{-})$, if their rank is finite.
\end{itemize}

\subsection{Covariant maps} \label{sec:Covariant_cases}

The previous arguments allow to discuss the construction of symmetric pseudo-periodic bases and Murray--von Neumann equivalence also in the context of particle-hole and/or chiral covariant projection-valued maps, introduced in Definition~\ref{def:Symmetries}. This discussion will be instrumental to the study of unitary equivalences in the next Section~\ref{sec:unitary_equivalences}.

We first look at particle-hole covariant projection-valued maps. In view of Remark~\ref{rmk:Covariance}, the conclusions reached in Section~\ref{sec:Bases_class_AIAII} apply also to this situation, as the particle-hole covariance constraint is formally equivalent to that of time-reversal symmetry. We obtain the following statement.

\begin{proposition}[Particle-hole covariant projection-valued maps] \label{prop:PHCovariant}
Let $P \colon \mathbb{T}^{1} \to \Proj_n(\mathcal{H})$ be a $1$-dimensional particle-hole covariant projection-valued map. Then, it admits a (time-reversal) symmetric periodic basis, in the sense of Definition~\ref{def:Symmetric_basis}. Moreover, any two such maps are Murray--von Neumann equivalent.

Let now $P \colon \mathbb{T}^{2} \to \Proj_n(\mathcal{H})$ be a $2$-dimensional particle-hole covariant projection-valued map. 
\begin{itemize}
    \item If the particle-hole symmetry is even, then $P$ admits a (time-reversal) symmetric periodic basis. Moreover, any two such maps are Murray--von Neumann equivalent.
    \item If the particle-hole symmetry is odd, then $P$ admits a (time-reversal) symmetric (pseudo-)periodic basis as in Proposition~\ref{prop:Class_AII_frames}. Moreover, if they have finite rank, two such maps are Murray--von Neumann equivalent if and only if their Fu--Kane--Mele indices coincide; if they have infinite rank, any two such maps are Murray--von Neumann equivalent.
\end{itemize}    
\end{proposition}

Coming to chiral covariant maps $P \colon \mathbb{T}^{d} \to \Proj_n(\mathcal{H})$, the covariant constraint $S P(k) = P(k) S$, $k \in \mathbb{T}^{d}$, implies that $P$ can be decomposed onto the eigenspaces of the symmetry operator $S$ as $P(k) = P^{\uparrow}(k) + P^{\downarrow}(k)$, where
\begin{equation}\label{eq:decomposition_chiral_covariant}
\begin{aligned}
P^\uparrow(k) & := \frac{\Id_{\mathcal{H}} + S}{2} P(k) = \frac{\Id_{\mathcal{H}} + S}{2} P(k) \frac{\Id_{\mathcal{H}} + S}{2}, \\
P^\downarrow(k) & :=\frac{\Id_{\mathcal{H}} - S }{2} P(k) = \frac{\Id_{\mathcal{H}} - S}{2} P(k) \frac{\Id_{\mathcal{H}} - S}{2}.
\end{aligned}
\end{equation}
Indeed, $(\Id_{\mathcal{H}} + S)/2$ (respectively $(\Id_{\mathcal{H}} - S)/2$)  is the spectral eigenprojection of $S$ onto $\mathcal{H}^{\uparrow} := \ker(S - \Id_{\mathcal{H}})$ (respectively $\mathcal{H}^{\downarrow} := \ker(S + \Id_{\mathcal{H}})$); note that these eigenspaces need not have the same dimension. Since the two eigenprojections $(\Id_{\mathcal{H}} \pm S)/2$ correspond to distinct eigenvalues of the unitary operator $S$, the projections $P^{\uparrow}$ and $P^{\downarrow}$ are orthogonal to one another. This means that the datum of $P$ is equivalent to the datum of two unrelated projection-valued maps $P^\uparrow \colon \T^d \to \Proj_{n^\uparrow}(\mathcal{H}^{\uparrow})$ and $P^\downarrow \colon \T^d \to \Proj_{n^\downarrow}(\mathcal{H}^{\downarrow})$, with $n^\uparrow+n^\downarrow=n$: in absence of other symmetries, they can be treated as two elements of class A. In this case, a notion of ``symmetric basis'' for the chiral covariant basis could be formulated as a basis which respects the decomposition~\eqref{eq:decomposition_chiral_covariant}, i.e., as a juxtaposition of a basis for $P^{\uparrow}$ (on which $S$ acts as $\Id_{\mathcal{H}}$) and one for $P^{\downarrow}$ (on which $S$ acts as $-\Id_{\mathcal{H}}$). This leads to the following

\begin{proposition}[Chiral covariant projection-valued maps] \label{prop:SCovariant}
Let $P \colon \mathbb{T}^{d} \to \Proj_n(\mathcal{H})$, $d \le 2$, be a chiral covariant projection-valued map. Then, it admits a symmetric (pseudo-)periodic basis if and only if $P^{\uparrow} \colon \mathbb{T}^{d} \to \Proj_{n^{\uparrow}}(\mathcal{H}^{\uparrow})$ and $P^{\downarrow} \colon \mathbb{T}^{d} \to \Proj_{n^{\downarrow}}(\mathcal{H}^{\downarrow})$, defined as in~\eqref{eq:decomposition_chiral_covariant}, admit (pseudo-)periodic bases, as described in Section~\ref{sec:Bases_class_A}. Moreover, two such maps $P_0, P_1$ are Murray--von Neumann equivalent if and only if $P_0^{\uparrow}$ is Murray--von Neumann equivalent to $P_1^{\uparrow}$ and $P_0^{\downarrow}$ is Murray--von Neumann equivalent to $P_1^{\downarrow}$.
\end{proposition}

In the statement, we decided not to spell out the conditions for $P^{\uparrow}$ and $P^{\downarrow}$ to admit (pseudo-)periodic bases, as they depend on the rank of the two projections and on the dimensionality of the eigenspaces of $S$. The reader may easily fill in these details and work out in particular the possible topological obstructions arising in $d=2$.

\bigskip

Finally, the situation becomes more involved if the projection-valued map under scrutiny is both particle-hole covariant and chiral covariant. As discussed in Section~\ref{sec:Symm_proj}, the presence of symmetry operators of the two kinds implies also the presence of a symmetry operator of the third kind, namely time-reversal symmetry. The three symmetry operators can commute or anticommute, depending on the sign of their squares, as prescribed in Table~\ref{tabular:AZC_classes}. In what follows, we focus on the chiral covariance and time-reversal symmetry constraint imposed on a projection-valued map $P \colon \mathbb{T}^{d} \to \Proj_n(\mathcal{H})$.

If the time-reversal symmetry operator $T$ commutes with the chiral symmetry operator $S$, then $T$ acts diagonally with respect to the Hilbert decomposition $\Hi = \mathcal{H}^{\uparrow} \oplus \mathcal{H}^{\downarrow}$ into eigenspace of $S$ presented above. This easily implies that the time-reversal symmetry constraints descends to the two components $P^{\uparrow}$ and $P^{\downarrow}$ in~\eqref{eq:decomposition_chiral_covariant}, which are then two time-reversal symmetric projection-valued maps, i.e., elements of class AI or AII according to the sign of $T^2$. This leads to the following result: once again, the careful reader can spell out the possible topological obstructions arising in the construction of symmetric periodic bases or Murray--von Neumann equivalences.

\begin{proposition}[Particle-hole and chiral covariant projection-valued maps, commuting symmetries] \label{prop:PHSCovariant_commute}
Let $P \colon \mathbb{T}^{d} \to \Proj_n(\mathcal{H})$, $d \le 2$, be a particle-hole and chiral covariant projection-valued map. Assume the particle-hole and chiral symmetry operators commute. Then, $P$ admits a symmetric (pseudo-)periodic basis if and only if $P^{\uparrow} \colon \mathbb{T}^{d} \to \Proj_{n^{\uparrow}}(\mathcal{H}^{\uparrow})$ and $P^{\downarrow} \colon \mathbb{T}^{d} \to \Proj_{n^{\downarrow}}(\mathcal{H}^{\downarrow})$, defined as in~\eqref{eq:decomposition_chiral_covariant}, admit time-reversal symmetric (pseudo-)periodic bases, as described in Section~\ref{sec:Bases_class_AIAII}. Moreover, two such maps $P_0, P_1$ are Murray--von Neumann equivalent if and only if $P_0^{\uparrow}$ is Murray--von Neumann equivalent to $P_1^{\uparrow}$ and $P_0^{\downarrow}$ is Murray--von Neumann equivalent to $P_1^{\downarrow}$.
\end{proposition}

If instead the time-reversal symmetry operator anticommutes with the chiral symmetry operator, we observe first of all that the presence of (anti)unitary operators anticommuting with $S$ imply that its eigenspaces are (anti)unitarily related, and therefore they have the same dimension and can be identified as Hilbert spaces (compare the proof of Principle~\ref{pri:chiral_classes_min_dim}). This time $T$ is off-diagonal in the Hilbert space decomposition $\Hi = \mathcal{H}^{\uparrow} \oplus \mathcal{H}^{\downarrow}$: this implies that 
\[ T P^\uparrow(k) = P^\downarrow(-k)T \quad \forall\: k \in \mathbb{T}^{d}, \]
that is, that $P^{\uparrow}$ and $P^{\downarrow}$ are also related by a antiunitary symmetry, and in particular they must have the same rank. More importantly, the above identity exhibits $\{P^\uparrow, P^\downarrow\}$ as a \emph{particle-hole symmetric} pair of projection-valued maps belonging to class C or D, depending on the sign of $T^2$: the time-reversal symmetry operator plays the role, from a mathematical viewpoint, of a particle-hole symmetry operator. The conclusion is the following.

\begin{proposition}[Particle-hole and chiral covariant projection-valued maps, anticommuting symmetries] \label{prop:PHSCovariant_anticommute}
Let $P \colon \mathbb{T}^{d} \to \Proj_n(\mathcal{H})$, $d \le 2$, be a particle-hole and chiral covariant projection-valued map. Assume the particle-hole and chiral symmetry operators anticommute. Then, $P$ admits a symmetric (pseudo-)periodic basis if and only if the pair of projection-valued maps $\{P^{\uparrow}, P^{\downarrow} \}$, defined as in~\eqref{eq:decomposition_chiral_covariant}, admit (pseudo-)periodic bases which are particle-hole symmetric with respect to $T$, as described in Section~\ref{sec:Bases_class_AIIICD}. Moreover, two such maps $P_0, P_1$ are Murray--von Neumann equivalent if and only if the pair $\{P_0^{\uparrow}, P_0^{\downarrow}\}$ is Murray--von Neumann equivalent to the pair $\{P_1^{\uparrow},P_1^{\downarrow}\}$.
\end{proposition}

\subsection{Parseval frames} \label{sec:Parseval}

The previous analysis has identified symmetric pseudo-periodic bases for (pairs of) projection-valued maps on the torus $\mathbb{T}^{d}$, $d \le 2$, in all AZC symmetry classes. We have shown how topological obstructions to the existence of periodic bases may arise in $d=2$, are quantified by Chern numbers of Fu--Kane--Mele invariants, and appear explicitly in the pseudo-periodicity constraints that one may enforce on one, one pair, or two pairs of vectors in the basis. We prove now that the topological obstructions can be avoided provided one relaxes the orthonormality constraint of the vectors of the basis into the condition of forming a \emph{Parseval frame} for the range of the corresponding projections~\cite{christensen2003introduction, freeman2014parseval}.

To this end, consider a pseudo-periodic vector $v(k)$ in the basis: its pseudo-periodicity condition reads
\[ v(\pi, k_2) = e^{i m k_2} v(-\pi,k_2) \quad \forall\: k_2 \in \mathbb{T}^{1} \]
with $m \in \mathbb{Z}$, cf.~\eqref{pseudoper} and the statements in the previous Sections. Define the functions $f, g \colon [-\pi,\pi]^2 \to \mathbb{C}$ as follows:
\[ f(k_1,k_2) := \frac{1-\sin(k_1/2)}{2} + \frac{1+\sin(k_1/2)}{2} \, e^{-i m k_2}, \quad g(k_1,k_2) := f(k_1, k_2) + \cos(k_1/2). \]
It is immediate to check the following properties:
\begin{itemize}
    \item both $f$ and $g$ are continuous with respect to $(k_1,k_2) \in [-\pi,\pi]^2$ and $2\pi$-periodic in $k_2$;
    \item $f(-\pi,k_2) = 1 = g(-\pi,k_2)$ and $f(\pi,k_2) = e^{-i m k_2} = g(\pi,k_2)$ for all $k_2 \in [-\pi,\pi]$, since $\cos(\pm \pi/2) = 0$;
    \item for the same reason, $f$ and $g$ cannot vanish simultaneously on $[-\pi,\pi]^2$: if they were both zero for some $(k_1,k_2) \in [-\pi,\pi]^2$, then $\cos(k_1/2) = 0$ or $k_1 = \pm \pi$, but there $|f(\pm \pi, k_2)|=1$, a contradiction.
\end{itemize}
Define now the continuous $\mathcal{H}$-valued functions
\[ u(k):= f(k) \, v(k) \quad \text{and} \quad w(k) := g(k) \, v(k) \quad \forall\: k \in [-\pi,\pi]^2. \]
The collection $\{u(k), w(k)\}$ spans the same ray as $v(k)$, as $f$ and $g$ are never simultaneously zero; moreover, they are both periodic in both directions, since
\[ u(\pi,k_2) = e^{-im k_2} \, v(\pi,k_2) = v(-\pi,k_2) = u(-\pi,k_2) \] 
and a similar consideration applies to $w(k)$. This means that every pseudo-periodic vector in a symmetric basis can be replaced by a pair of \emph{periodic} vectors, which together with the rest of the vectors in the basis form a Parseval frame. This procedure realizes a symmetric periodic Parseval frame for (pairs of) projection-valued maps over the $d$-dimensional torus with $d \le 2$, which moreover has the minimal possible number of redundant vectors, generalizing the results of~\cite{auckley2018parseval,cornean2019parseval, monaco2023topology, Peluso2026} to all AZC classes in this low-dimensional setting.

\section{Unitary equivalence}\label{sec:unitary_equivalences}

We now come to Question 2 raised in Section~\ref{pro:unitary_equivalences}, namely that of unitary equivalence between (pairs of) projection-valued maps. As discussed there, in order to construct unitary equivalences, it suffices to combine two Murray--von Neumann equivalences, one for the projection-valued maps and a second one for the projection on the orthogonal complement of the ranges of the projections. More precisely, we set the following

\begin{definition}[Complementary projection-valued map] \label{def:Complementary}
Let $P \colon \mathbb{T}^{d} \to \Proj_n(\mathcal{H})$ be a single projection-valued map. The \emph{complementary projection-valued map} is defined as $Q := P^\perp \colon \mathbb{T}^{d} \to \Proj_{\dim(\mathcal{H})-n}(\mathcal{H})$, i.e.
\[ Q(k) := \Id_{\mathcal{H}} - P(k) \quad  \forall\:k \in \mathbb{T}^{d}. \]

Let $P^{\pm} \colon \mathbb{T}^{d} \to \Proj_n(\mathcal{H})$ be a pair of projection-valued maps (recall in particular the orthogonality assumption~\eqref{standing_orthogonality}). The \emph{complementary pro\-jec\-tion-valued map} is defined as $Q := (P^{-} + P^{+})^\perp \colon \mathbb{T}^{d} \to \Proj_{\dim(\mathcal{H})-2n}(\mathcal{H})$, i.e.
\[ Q(k) := \Id_{\mathcal{H}} - P^{-}(k) - P^{+}(k)\quad  \forall\: k \in \mathbb{T}^{d}. \]
\end{definition}

Two (pairs of) projection valued-maps $P_0^{(\pm)}, P_1^{(\pm)}$ are thus unitarily equivalent if and only if they are Murray--von Neumann equivalent and their complementary projection-valued maps $Q_0, Q_1$ are Murray--von Neumann equivalent as well. For each AZC class, we need to study the symmetry (or rather covariance, compare Definition~\ref{def:Symmetries}) constraints satisfied by these complementary projection-valued maps, and check if they produce additional topological obstructions. 

Let us start with some general considerations.
\begin{itemize}
    \item As we have shown in Section~\ref{sec:Bases} (cf.\ also~\cite[Theorem~3.2]{manzonimonacopeluso2026zak}), in all AZC classes any two $0$-dimensional or $1$-dimensional (pairs of) projection-valued maps, be it symmetric or covariant, are Murray--von Neumann equivalent. In this case, there is no topological information to record, neither in the (pair of) projection-valued map(s) nor in the corresponding complementary projection-valued map, and two such (pairs of) projection-valued maps are unitarily equivalent if and only they have the same rank.
    \item Similarly, we have seen how Kuiper's theorem shows that any two (pairs of) infinite-rank projection-valued maps on the $d$-dimensional torus, $d \le 2$, are Murray--von Neumann equivalent, with equivalences which respect any symmetry or covariance constraint coming from the AZC classification. In particular, if the ambient Hilbert space $\mathcal{H}$ has infinite dimension, then the complementary projection-valued map always has infinite rank, and thus carries no extra topological information to the classification under unitary equivalence.
\end{itemize}

In view of the previous considerations, we obtain the following result dealing with the infinite-rank situation.

\begin{theorem} \label{thm:unitary_eq_infinite-rank}
Any two particle-hole and/or chiral symmetric pairs of infinite-rank projection-valued maps over the $d$-dimensional torus, $d \le 2$, are unitarily equivalent.
\end{theorem}

In what follows we will focus our attention on the potentially more interesting case in which the projection-valued maps under scrutiny are defined on the $2$-dimensional torus and act on an ambient Hilbert space $\mathcal{H}$ which is finite-dimensional.

\subsection{Class A} \label{sec:Unitary_A}

For a projection-valued map $P:\T^2 \to \Proj_{n}(\Hi)$ in class A, Proposition~\ref{prop:Class_A_frames} states that there is a topological invariant $\Ch(P) \in \mathbb{Z}$, its Chern number, characterizing its Murray--von Neumann equivalence class. The complementary projection-valued map $Q$ is still in class A, and therefore has its own Chern number $\Ch(Q) \in \mathbb{Z}$. However, since $\Id_{\mathcal{H}} = P(k)+ Q (k)$ for all $k \in \mathbb{T}^{2}$ and since the left-hand side of this equality corresponds to a constant projection-valued map with vanishing Chern number, the additive property discussed in Remark~\ref{rmk:Chern} ensures that $\Ch(Q) = -\Ch(P)$. So, the topological information carried by the complementary projection-valued map is completely determined by the one carried by the original projection-valued map, and thus the rank $n \in \mathbb{N}$ and the Chern number $\Ch(\cdot) \in \mathbb{Z}$ are the only topological invariants that characterize unitary equivalences (as well as Murray--von Neumann equivalences) in class A.

\subsection{Classes AI and AII} \label{sec:Unitary_AIAII}

For a time-reversal symmetric projection-valued map $P:\T^2 \to \Proj_n(\Hi)$, the complementary map $Q$ is again time-reversal symmetric (with respect to the same time-reversal symmetry operator $T$). In class AI, Proposition~\ref{prop:Class_AI_frames} states that any two such projection-valued maps are Murray--von Neumann equivalent, and thus unitarily equivalent. In class AII, Proposition~\ref{prop:Class_AII_frames} states that the Murray--von Neumann equivalence class of $P$ (respectively $Q$) is characterized by its Fu--Kane--Mele invariant $\FKM(P) \in \mathbb{Z}_2$ (respectively $\FKM(Q) \in \mathbb{Z}_2)$. Arguing as above, the multiplicative property discussed in Remark~\ref{rmk:FKM} ensures that $\FKM(Q)=\FKM(P) \in \mathbb{Z}_2$, so that the topological content of the complementary map is completely determined by that of the original projection-valued map. We conclude that the rank $n \in 2\mathbb{N}$ and the Fu--Kane--Mele invariant $\FKM(\cdot) \in \mathbb{Z}_2$ are the only topological invariants that characterizes unitary equivalences (as well as Murray--von Neumann equivalences) in class AII.

\subsection{Classes C and D} \label{sec:Unitary_CD}

For a pair of particle-hole symmetric projection-valued maps $P^\pm :\T^2 \to \Proj_n(\Hi)$, its Murray--von Neumann equivalence class is characterized by the Chern number $\Ch(P^{-}) \in \mathbb{Z}$. It is easily realized that the complementary projection-valued map $Q$ is particle-hole covariant, i.e., time-reversal symmetric (compare Remark~\ref{rmk:Covariance}). In view of Propositions~\ref{prop:Class_AII_frames} and~\ref{prop:PHCovariant}, we conclude that the complementary projection-valued maps of any two projection-valued maps in class D are always Murray--von Neumann equivalent (as maps in class AI), while in class C the complementary maps (which lie in class AII) are Murray--von Neumann equivalent if and only if their Fu--Kane--Mele invariants agree. We can still invoke the multiplicativity of this invariant to compute $\FKM(Q)$ as $\FKM(P^{-} + P^{+})$. It turns out that the latter is still related to the Chern number of $P^{-}$: indeed,~\cite[Theorem~6.1]{Peluso2026} shows that 
\[ \FKM(P^-+P^+) = (-1)^{\Ch(P^-)} \in \mathbb{Z}_2. \]
We conclude that, in both classes C and D, unitary equivalence classes (as well as Murray--von Neumann equivalence classes) of pairs of particle-hole symmetric rank-$n$ projection-valued maps $P^{\pm}$ are characterized by the Chern number $\Ch(P^-) \in \mathbb{Z}$.

\subsection{Class AIII} \label{sec:Unitary_AIII}

For a pair of chiral symmetric projection-valued maps $P^\pm:\T^2 \to \Proj_n(\Hi)$, the two elements of the pair are unitarily related by the chiral symmetry $S$, leading to the fact that they have equal ranks and Chern numbers (cf.\ Remark~\ref{rmk:Chern}). Moreover, we find in this case that the projection-valued maps $Q$ complementary to the pair is chiral covariant. Let us split $Q(k) = Q^\uparrow(k)+Q^\downarrow(k)$, $k \in \mathbb{T}^{2}$, into the eigenspaces $\mathcal{H}^{\uparrow} := \ker(S - \Id_{\mathcal{H}})$ and $\mathcal{H}^{\downarrow} := \ker(S + \Id_{\mathcal{H}})$ of $S$, as in~\eqref{eq:decomposition_chiral_covariant}. Proposition~\ref{prop:SCovariant} yields that there are in principle six topological invariants involved in the characterization of unitary equivalence classes, namely the ranks and Chern numbers of $P^{-}$, $Q^{\uparrow}$ and $Q^{\downarrow}$. Let us show, however, that the data from the complementary projection can be reconstructed from the knowledge of $P^{-}$ and of the chiral symmetry operator $S$.

\begin{lemma} \label{lemma:AIII_unitary_equivalence}
With the notation above, the following identities hold:
\begin{align*}
\operatorname{rank}(Q^\uparrow(k)) & = \dim(\mathcal{H}^{\uparrow}) - \operatorname{rank}(P^{-}(k)) \equiv \dim(\mathcal{H}^{\uparrow}) - n, & \Ch(Q^\uparrow) & = -\Ch(P^-), \\
\operatorname{rank}(Q^\downarrow(k)) & = \dim(\mathcal{H}^{\downarrow}) - \operatorname{rank}(P^{-}(k)) \equiv \dim(\mathcal{H}^{\downarrow}) - n, &
\Ch(Q^\downarrow) & = -\Ch(P^-). 
\end{align*}
\end{lemma}
\begin{proof}
Let us consider the decomposition~\eqref{eq:decomposition_chiral_covariant} into eigenspaces of $S$ for the chiral covariant projection-valued map $P := P^{-} + P^{+} \colon \mathbb{T}^{2} \to \Proj_{2n}(\mathcal{H})$:
\[ P^\uparrow(k):=\frac{\Id_{\mathcal{H}} + S}{2}(P^+(k)+P^-(k)), \quad P^\downarrow(k):=\frac{\Id_{\mathcal{H}} - S}{2}(P^+(k)+P^-(k)), \quad k \in \mathbb{T}^{2}. \]
We observe that these two projection-valued maps are Murray--von Neumann equivalent via the partial-isometry-valued map
\[ V \colon \mathbb{T}^{2} \to \mathcal{B}(\mathcal{H}), \quad V(k) := P^+(k)-P^-(k). \]
Indeed, one can argue as in the proof of Principle~\ref{pri:chiral_classes_min_dim} that $V(k)$ is the partial isometry onto the range of $P^{+}(k) + P^{-}(k)$ (recall our standing assumption~\eqref{standing_orthogonality}), that it anticommutes with $S$, and therefore that
\[ V(k) \, P^\uparrow(k) = P^\downarrow(k) \, V(k), \quad k \in \mathbb{T}^{2}, \]
yielding the desired Murray--von Neumann equivalence. The results of Section~\ref{sec:Bases_class_AIIICD} imply that the ranks of $P^\uparrow$ and $P^\downarrow$ are equal, as well as their Chern numbers.

Notice that 
\[ P^\uparrow(k) + P^\downarrow(k) = P^+(k) + P^-(k) \quad \forall \: k \in \mathbb{T}^{2}. \]
The identity $\Ch(P^{\uparrow}) = \Ch(P^{\downarrow})$, together with the additivity of the Chern number and the already-noted identity $\Ch(P^{+}) = \Ch(P^{-})$, allow then to conclude that
\[ \Ch(P^\uparrow) = \Ch(P^\downarrow) = \Ch(P^-) = \Ch(P^+). \]
Let us further notice the identities
\[ \frac{\Id_{\mathcal{H}} + S}{2} = P^\uparrow(k) + Q^\uparrow(k) \quad \text{and} \quad \frac{\Id_{\mathcal{H}} - S}{2} = P^\downarrow(k) + Q^\downarrow(k) \quad \forall\: k \in \mathbb{T}^{2}. \]
Since in both cases the left-hand side is a constant projection-valued map, and since $P^{\uparrow} Q^{\uparrow} \equiv 0 \equiv P^{\downarrow} Q^{\downarrow}$, the above imply that 
\[ \operatorname{rank}(Q^\uparrow(k)) \equiv \dim(\mathcal{H}^{\uparrow})-n, \quad \operatorname{rank}(Q^\downarrow(k)) \equiv \dim(\mathcal{H}^{\downarrow})-n, \]
as well as
\[ \Ch(Q^\uparrow) = - \Ch(P^\uparrow) = -\Ch(P^-), \quad \Ch(Q^\downarrow) = - \Ch(P^\downarrow) = -\Ch(P^-), \]
as desired.
\end{proof}

We conclude that, in class AIII, the unitary equivalence class (as well as the Murray--von Neumann equivalence class) of a chiral symmetric pair of projection-valued maps $P^{\pm} \colon \mathbb{T}^{2} \to \Proj_n(\mathcal{H})$ is characterized by the rank $n \in \mathbb{N}$ and the Chern number $\Ch(P^-) \in \mathbb{Z}$ of one of its elements.

\subsection{Classes BDI and CII} \label{sec:Unitary_BDICII}

In these classes, all three types of symmetries are present and commute with each other. The analysis from class AIII on the projection-valued map $Q$ complementary to a symmetric pair $P^{\pm} \colon \mathbb{T}^{2} \to \Proj_n(\mathcal{H})$ carries over, with the difference that now the components of the splitting $Q = Q^{\uparrow} + Q^{\downarrow}$ are also separately time-reversal symmetric. In particular, Proposition~\ref{prop:PHSCovariant_commute} can be invoked to identify the Murray--von Neumann equivalence classes of such complementary projection-valued maps.

In class BDI, where the time-reversal symmetry operator is even, all the involved projection-valued maps have trivial Chern numbers, resulting in any two pairs of projection-valued maps with the same rank being unitarily equivalent (as well as Murray--von Neumann equivalent). 

In class CII, where instead there is an odd time-reversal symmetry, the relevant topological quantities are Fu--Kane--Mele invariants rather than Chern numbers. Nonetheless, the arguments in the proof of Lemma~\ref{lemma:AIII_unitary_equivalence} can be easily adapted to the present setting, using multiplicativity of the Fu--Kane--Mele invariant rather than additivity of the Chern number. The end result is that 
\begin{align*}
\operatorname{rank}(Q^\uparrow(k)) & = \dim(\mathcal{H}^{\uparrow}) - \operatorname{rank}(P^{-}(k)) \equiv \dim(\mathcal{H}^{\uparrow}) - n, & \FKM(Q^\uparrow) & = -\FKM(P^-), \\
\operatorname{rank}(Q^\downarrow(k)) & = \dim(\mathcal{H}^{\downarrow}) - \operatorname{rank}(P^{-}(k)) \equiv \dim(\mathcal{H}^{\downarrow}) - n, &
\FKM(Q^\downarrow) & = -\FKM(P^-). 
\end{align*}
so that as before $P^-$ contains all the topological information, through its rank and Fu--Kane--Mele invariant, to characterize unitary equivalence classes (as well as Murray--von Neumann equivalence classes) of symmetric pairs of projection-valued maps in class CII.

\subsection{Classes CI and DIII} \label{sec:Unitary_CIDIII}

Lastly, in classes CI and DIII, all three symmetries are present, but this time they anticommute with each other. The Murray--von Neumann equivalence class of the projection-valued map $Q$ complementary to a pair $P^{\pm} \colon \mathbb{T}^{2} \to \Proj_n(\mathcal{H})$ can be analyzed via Proposition~\ref{prop:PHSCovariant_anticommute}. The pair $\{Q^{\uparrow}, Q^{\downarrow}\}$ from~\eqref{eq:decomposition_chiral_covariant} behaves as a particle-hole symmetric pair, where the role of the particle-hole symmetry is played by the time-reversal symmetry operator $T$: in particular, the pair $\{Q^{\uparrow}, Q^{\downarrow}\}$ is in class D if the original pair $P^{\pm}$ is in class CI (even~$T$), while it is in class C if the original pair is in class DIII (odd~$T$). With this point of view, we can deduce from Proposition~\ref{prop:Class_AIIICD_frames} that the Murray--von Neumann equivalence class of the pair $\{Q^{\uparrow}, Q^{\downarrow}\}$ is characterized by the rank of $Q^{\downarrow}$ together with $\Ch(Q^{\downarrow}) = - \Ch(Q^{\uparrow})$. On the other hand, the argument from Lemma~\ref{lemma:AIII_unitary_equivalence} still holds, yielding that 
\[ \Ch(Q^{\downarrow}) = \Ch(Q^{\uparrow}) = - \Ch(P^{-}) = 0, \]
with the latter equality following from the time-reversal symmetry of $P^{-}$. The discussion before Proposition~\ref{prop:PHSCovariant_anticommute} also tells that the eigenspaces of $S$, namely $\mathcal{H}^{\uparrow}$ and $\mathcal{H}^{\downarrow}$, have the same dimension, leading to $Q^{\uparrow}$ and $Q^{\downarrow}$ also having the same rank equal to $\dim(\mathcal{H}^{\uparrow}) - n = \dim(\mathcal{H}^{\downarrow})-n$. In conclusion, the knowledge of $n = \operatorname{rank}(P^{-})$ completely characterizes the Murray--von Neumann equivalence class of $Q$; notice that, while in class CI the rank $n \in \mathbb{N}$ also characterizes the Murray--von Neumann equivalence class of $P^{\pm}$, in class DIII this positive integer (which is necessarily even) must be complemented with $\FKM(P^{-}) \in \mathbb{Z}_2$ to recover a complete set of invariants.

\bigskip

To sum up the previous discussion, we have proved the following

\begin{theorem} \label{thm:unitary_iff_MvN}
In all AZC classes, two symmetric (pairs of) projection-valued maps over the $d$-dimensional torus, $d \le 2$, are unitarily equivalent if and only if they are Murray--von Neumann equivalent.
\end{theorem}

\subsection{Covariant maps} \label{sec:Unitary_covariant}

The fact that unitary equivalence of projection-valued maps is the same as Murray--von Neumann equivalence of the maps themselves as well as of their complementary maps holds for covariant maps, too. Remark~\ref{rmk:Covariance} observes that the complementary map to a covariant map is covariant as well (with respect to the same symmetries); besides, the results of Section~\ref{sec:Bases}, and in particular of Section~\ref{sec:Covariant_cases}, have reduced the condition for covariant projection-valued maps to be Murray--von Neumann equivalent to the equality of certain topological quantities, in the form of Chern numbers or Fu--Kane--Mele invariants (compare Table~\ref{tabular:AZC_classes_MvN}). By the additivity (respectively multiplicativity) of these invariants, a moment's thought allows to realize that the conditions to be checked for Murray--von Neumann equivalence of the complementary projection-valued maps coincide with the ones for the equivalence of the maps themselves. We can therefore complement Theorem~\ref{thm:unitary_iff_MvN} on the classification of symmetric (pairs of) projection-valued maps up to unitary equivalence with the following corresponding statement for covariant maps.

\begin{theorem} \label{thm:unitary_iff_MvN_covariant}
In all symmetry classes, two covariant projection-valued maps over the $d$-dimensional torus, $d \le 2$, are unitarily equivalent if and only if they are Murray--von Neumann equivalent.
\end{theorem}

\section{Homotopy equivalence}\label{sec:Homotopy}

We come now to Question 3 raised in Section~\ref{pro:homotopies}, namely that of homotopy equivalence among symmetric (pairs of) projection-valued maps on the $d$-dimensional torus, $d \le 2$. Principle~\ref{pri:unitary_eq_to_homotopy} reduced this Question to that of unitary equivalence, provided moreover the symmetric unitary-valued map intertwining two such maps can be deformed to one commuting with one of the projections: for most practical purposes, one may think that the unitary-valued map should be continuously deformable to the constant map equal to the identity on the ambient Hilbert space. We will thus start this Section by analyzing homotopy classes of unitary-valued maps which respect the symmetry constraints~\eqref{MvNsymmetry}, which could be of independent interest. We stress once again that, as for projection-valued maps, for us the notion of `homotopy' $V \colon [0,1] \times \mathbb{T}^{d} \to \U(\mathcal{H})$ requires that $V(t,\cdot)$ satisfies the same symmetry constraints as its endpoints for all $t \in [0,1]$.

\subsection{Homotopies of symmetric unitary-valued maps} \label{sec:Homotopies_unitary}

The symmetry constraints~\eqref{MvNsymmetry} on a map $V \colon \mathbb{T}^{d} \to \U(\mathcal{H})$ can be regarded as covariance conditions (or rather, equivariance conditions, in more mathematical terms) expressing the interplay between an action of the group $\Z_2$ on the torus, either via the involution $k \mapsto -k$ (in the case of antiunitary symmetries) or via the identity map (in the case of unitary symmetries), and the one on the space of unitaries via conjugation by the appropriate symmetry operator, the latter being a $\Z_2$-action since symmetry operators square to $\pm \Id_{\mathcal{H}}$. As was already remarked after~\eqref{MvNsymmetry}, time-reversal symmetry $T V(k) = V(-k) T$ and particle-hole symmetry $C V(k) = V(-k) C$ express the same type of condition, as $T$ and $C$ are both instances of antiunitary operators on $\mathcal{H}$ squaring to $\pm \Id_{\mathcal{H}}$. Moreover, many structural results presented in Section~\ref{sec:Covariant_cases} in the context of projection-valued maps apply {\it verbatim} to unitary-valued maps as well. For example, unitary-valued maps such that $S V(k) = V(k) S$, where $S$ is a chiral symmetry operator, can be block-decomposed into the eigenspaces of $S$: writing $\mathcal{H} = \mathcal{H}^{\uparrow} \oplus \mathcal{H}^{\downarrow}$ with $\mathcal{H}^{\uparrow} = \ker(S - \Id_{\mathcal{H}})$ and $\mathcal{H}^{\downarrow} = \ker(S + \Id_{\mathcal{H}})$, as before, any chiral covariant $V$ splits as $V = V^{\uparrow} \oplus V^{\downarrow}$, where $V^{\uparrow} \colon \mathbb{T}^{d} \to \U(\mathcal{H}^{\uparrow})$ and $V^{\downarrow} \colon \mathbb{T}^{d} \to \U(\mathcal{H}^{\downarrow})$ are defined by
\begin{equation} \label{eq:decomposition_chiral_unitary}
V^{\uparrow}(k) := \frac{\Id_{\mathcal{H}} + S}{2} \, V(k), \quad V^{\downarrow}(k) := \frac{\Id_{\mathcal{H}} - S}{2} \, V(k), \quad k \in \mathbb{T}^{d},
\end{equation}
compare~\eqref{eq:decomposition_chiral_covariant}. 

These considerations produce a significant simplification in the study of homotopy classes of symmetric unitary-valued maps. Indeed, in presence of chiral symmetry, a homotopy preserving the symmetry must preserve also the splitting above; therefore, chiral symmetric unitary-valued maps are homotopic if and only if the two blocks in~\eqref{eq:decomposition_chiral_unitary} are separately homotopic as maps with values in unitary operators on $\mathcal{H}^{\uparrow}$ and $\mathcal{H}^{\downarrow}$, respectively. In presence of a further time-reversal symmetry constraint, the time-reversal symmetry operator $T$ and the chiral symmetry operator $S$ may commute or anticommute. If they commute, then $T$ acts diagonally in the decomposition $\mathcal{H} = \mathcal{H}^{\uparrow} \oplus \mathcal{H}^{\downarrow}$, meaning that the blocks $V^{\uparrow}$ and $V^{\downarrow}$ will be separately time-reversal symmetric. Thus two symmetric unitary-valued maps will be homotopic if and only if the corresponding blocks are homotopic as time-reversal symmetric maps. If instead $T$ and $S$ anticommute, then $T$ exchanges the two blocks:
\[ T V^{\uparrow}(k) = V^{\downarrow}(-k) T, \quad k \in \mathbb{T}^{d}. \]
A homotopy of the block $V^{\uparrow}$ thus induces a corresponding homotopy of the block $V^{\downarrow}$ by conjugation with $T$ and inversion of $k$; these two homotopies may be combined in a way which produces a chiral symmetric homotopy of the overall $\U(\mathcal{H})$-valued maps.

We have thus reduced the homotopy classification of symmetric unitary-valued maps to the cases with no symmetries or with only time-reversal symmetry (of even or odd type). This question has been already investigated in the literature: see e.g.~\cite[Appendix~A]{Monaco_2017} and references therein. We state the results and highlight some relevant aspects of the proofs in the low-dimensional setting $d \le 2$ under consideration; we refer the reader to the previous references for more details. In the following statement, we adopt the same notation introduced in Definition~\ref{def:Chern_number} for which $[\cdot] \in \mathbb{Z}$ denotes the winding number of a map $\mathbb{T}^1 \simeq S^1 \to U(1) \simeq S^1$.

\begin{theorem}[Homotopy classes of (time-reversal symmetric) unitary-valued maps] \label{thm:TRS_unitary_homotopy}
Assume that $\mathcal{H}$ is a finite-dimensional Hilbert space. Let $V_0, V_1 \colon \mathbb{T}^{d} \to \U(\mathcal{H})$ be unitary-valued maps.

\begin{itemize}[leftmargin=*]
    \item Let $d=0$. Then $V_0$ and $V_1$ are homotopic. The same holds if they are time-reversal symmetric and the time-reversal symmetry operator is odd. If instead the time-reversal symmetry operator is even, then they are homotopic if and only if
    \[ \det(V_0) = \det(V_1) \in \{\pm 1\}. \]
    \item Let $d=1$. Then $V_0$ and $V_1$ are homotopic if and only if
    \[ [\det(V_0(\cdot))] = [\det(V_1(\cdot))] \in \Z. \]
    The same holds if the maps are time-reversal symmetric. If the time-reversal symmetry operator is odd, then the winding numbers above are even.
    \item Let $d=2$. Denote
    \begin{equation} \label{eq:T^1_1,2}
    \mathbb{T}^{1}_{1} := \{(k_1,k_2) \in \mathbb{T}^{2} : k_2 = 0\} \subset \mathbb{T}^{2} \quad \text{and} \quad \mathbb{T}^{1}_{2} := \{(k_1,k_2) \in \mathbb{T}^{2} : k_1 = 0\} \subset \mathbb{T}^{2}.
    \end{equation}
    Then $V_0$ and $V_1$ are homotopic as (time-reversal symmetric) unitary-valued maps on the $2$-dimensional torus if and only if $V_0 \big|_{\mathbb{T}^{1}_{1}}$ is homotopic to $V_1 \big|_{\mathbb{T}^{1}_{1}}$ and $V_0 \big|_{\mathbb{T}^{1}_{2}}$ is homotopic to $V_1 \big|_{\mathbb{T}^{1}_{2}}$ as (time-reversal symmetric) unitary-valued maps on the $1$-dimensional torus.
\end{itemize}

If instead $\mathcal{H}$ is infinite-dimensional, then any two (time-reversal symmetric) unitary-valued maps $V_0, V_1 \colon \mathbb{T}^{d} \to \U(\mathcal{H})$, $d \le 2$, are homotopic (with the homotopy preserving time-reversal symmetry).
\end{theorem}
\begin{proof}[Sketch of the proof.]
In $d=0$, the choice of two ``unitary-valued maps'' is just the choice of two unitaries $V_0, V_1 \in \U(\mathcal{H})$, and the question of homotopy reduces to the question of path-connectedness. The group of unitaries on any Hilbert space (be it finite- or infinite-dimensional) is path-connected. If one wants to enforce time-reversal symmetry, then, as discussed in Section~\ref{sec:Bases_class_AIAII}, the time-reversal symmetry operator endows the Hilbert space $\mathcal{H}$ with a real or quaternionic structure, depending on whether the symmetry is even or odd, respectively. If the Hilbert space is infinite-dimensional, then Kuiper's theorem guarantees path-connectedness of any (real, complex, or quaternionic) Hilbert space. If the Hilbert space is finite-dimensional, we can use the normal form~\eqref{eq:TRS_normal_form_even} for an even time-reversal symmetry operator, respectively~\eqref{eq:TRS_normal_form_odd} for an odd one, to realize that the time-reversal symmetry constraint $T V = V T$ is equivalent to the request that the matrix $V$ be orthogonal, respectively symplectic. While the group of symplectic matrices is path-connected (and all such matrices have determinant equal to 1), the group of orthogonal matrices is split into two path-connected components, corresponding to the value of their determinant, which can be either $+1$ or $-1$.

Moving to higher $d \ge 1$, still for finite-dimensional Hilbert spaces, we note that the above considerations also apply at the fixed points $k_{\star}$ for the involution $k \mapsto -k$ (compare Remark~\ref{rmk:Fixed_points}). 
In $d=1$, the fact that the winding number of the determinant characterizes the homotopy class of a unitary-valued map on $\mathbb{T}^{1} \simeq S^1$ is well known~\cite[Chapter~8, Section~12]{Husemoller1994}. 
Time-reversal symmetry imposes the constraint $\det(V(k)) = \overline{\det(V(-k))}$, $k \in \mathbb{T}^{1}$, at the level of the determinant, so that its values on the half-torus $k \in [0,\pi]$ fully determine the values over the whole torus as well as the winding. 
In the case of an odd time-reversal symmetry, the symplectic constraint on $V(k_{\star})$ forces $\det(V(k_{\star}))=1$, $k_{\star} \in \{0, \pm \pi\}$, so that the restriction $\det V(\cdot) \big|_{[0,\pi]}$ is already periodic and has an integer winding number. This establishes the evenness of the winding number of the determinant for an odd-time-reversal symmetric unitary-valued map on the $1$-dimensional torus.

Finally, constructing homotopies for unitary-valued maps on the $2$-dimensional torus reduces to constructing homotopies for their restrictions on the two $1$-dimensional sub-tori $\mathbb{T}^{1}_{1}$ and $\mathbb{T}^{1}_{2}$, as their union constitutes the $1$-skeleton of the CW-complex decomposition of the $2$-dimensional torus. Attaching the $2$-cell produces no further topological obstruction as the second homotopy group $\pi_2(\U(\mathcal{H}))$ is trivial, both for finite- and infinite-dimensional Hilbert spaces (see~\cite[Chapter 8, Section 12]{Husemoller1994} and~\cite{kuiper1965homotopy}, respectively).
\end{proof}

\subsection{Homotopies of projection-valued maps}

We return to the discussion of homotopies for projection-valued maps $P \colon \mathbb{T}^{d} \to \Proj_n(\mathcal{H})$, starting from the case of (possibly time-reversal symmetric) maps in classes A, AI, and AII. We claim that all unitary equivalences among projection-valued maps in these classes can be chosen to be of the type specified in the statement of Principle~\ref{pri:unitary_eq_to_homotopy}. We will thus prove the following

\begin{theorem}[Homotopy equivalence in class A, AI, and AII] \label{thm:Homotopies_AAIAII}
In classes A, AI, and AII, two (symmetric) projection-valued maps over the $d$-dimensional torus, $d \le 2$, are homotopically equivalent if and only if they are Murray--von Neumann equivalent.
\end{theorem}
\begin{proof}
Principle~\ref{pri:unitary_eq_to_homotopy} states in general that homotopic maps are also unitarily equivalent, which in turn implies Murray--von Neumann equivalence. Conversely, by Theorem~\ref{thm:unitary_iff_MvN}, two Murray--von Neumann equivalent projection-valued maps $P_0, P_1 \colon \mathbb{T}^{d} \to \Proj_n(\mathcal{H})$ are also unitarily equivalent, say via a unitary-valued map $V \colon \mathbb{T}^{d} \to \U(\mathcal{H})$; in view of the Principle mentioned above, we need to show that this unitary equivalence can be continuously deformed to $V_0 \colon \mathbb{T}^{d} \to \U(\mathcal{H})$ such that $\left[ V_0(k), P_0(k) \right] = 0$ for all $k \in \mathbb{T}^{d}$.

According to Theorem~\ref{thm:TRS_unitary_homotopy}, there are instances in which any two (time-reversal symmetric) unitary-valued maps are homotopic to each other, and in particular $V$ can be continuously deformed to the map constantly equal to the identity, which certainly commutes with $P_0$: this occurs if the ambient Hilbert space $\mathcal{H}$ is infinite-dimensional, or if $\mathcal{H}$ is finite-dimensional, $d=0$, and either no symmetries or an odd time-reversal symmetry are present. In all other cases, the homotopy class of a unitary-valued map is specified by looking at its determinant. 

First of all, let us consider classes A and AI. Let $\{v_j(k)\}_{j \in \{1,\ldots,n\}}$ be a (time-reversal symmetric) orthonormal basis for $P_0(k)$, which exists by the results of Sections~\ref{sec:Bases_class_A} and~\ref{sec:Bases_class_AIAII}. Define $V_0 \colon \mathbb{T}^{d} \to \U(\mathcal{H})$ as the linear extension of the operator given by
\begin{equation}\label{eq:auxiliary_equation_30} V_0(k) \, v_1(k) := \det(V(k)) \, v_1(k), \quad V_0(k) \, w := w \text{ if } w \perp v_1(k), \quad k \in \mathbb{T}^{d}. \end{equation}
If the basis is pseudo-periodic, the operator $V_0$ is still periodic thanks to  Principle \ref{pri:Periodicity_from_pseudo-periodicity}.
Then by construction $V_0(k)$ commutes with $P_0(k)$; if time-reversal symmetry is present and the basis is chosen such that $T v_j(k) = v_j(-k)$, $j \in \{1, \ldots, n\}$, $k \in \mathbb{T}^{d}$, then
\begin{align*}
T \, V_0(k) \, v_1(k) & = T \left[ \det(V(k)) \, v_1(k) \right] = \overline{\det(V(k))} \, T \, v_1(k) = \det(V(-k)) \, v_1(-k) \\
& = V_0(-k) \, v_1(-k) = V_0(-k) \, T \, v_1(k)
\end{align*}
since, as already noted, the time-reversal symmetry constraint~\eqref{MvNsymmetry} on $V$ implies $\overline{\det(V(k))} = \det(V(-k))$. Together with the fact that $V_0(k)$ acts as the identity on the orthogonal complement of the ray spanned by $v_1(k)$, the above identity implies that $T V_0(k) = V_0(-k) T$, that is, that $V_0$ is time-reversal symmetric as well. Since clearly $\det(V_0(k)) = \det(V(k))$, the maps $V$ and $V_0$ are homotopic, in view of Theorem~\ref{thm:TRS_unitary_homotopy}. Invoking Principle~\ref{pri:unitary_eq_to_homotopy} allows to construct the corresponding homotopy between the projection-valued maps and conclude the proof.

In class AII, the previous construction needs to be modified as follows. Recall that, in presence of an odd time-reversal symmetry, the projection-valued map $P_0$ has even rank $n$, and vectors spanning its range come in Kramers pairs $\{v_{2j-1}(k), v_{2j}(k) = - T v_{2j-1}(k)\}_{j \in \{1, \ldots, n/2\}}$. Recall also that, in view of Theorem~\ref{thm:TRS_unitary_homotopy}, the winding number of the determinant of the unitary equivalence $V$ along the $1$-dimensional (sub-)torus $\mathbb{T}^{1}$ (respectively $\mathbb{T}^1_s \subset \mathbb{T}^{2}$) is necessarily even, say equal to $2\ell \in \mathbb{Z}$ (respectively $2\ell_s \in \mathbb{Z}$, $s \in \{1,2\}$). Define this time, for $k \in \mathbb{T}^{1}$,
\begin{equation}\label{eq:auxiliary_equation_31} V_0(k) \, v_j(k) := e^{i \ell k} \, v_j(k), \; j \in \{1,2\}, \quad  V_0(k) \, w := w \mbox{ if } w \perp \operatorname{Span} \{ v_1(k), v_2(k)\}, \end{equation}
respectively, for $k \in \mathbb{T}^{2}$,
\begin{equation} \label{eq:auxiliary_equation_32} V_0(k) \, v_j(k) := e^{i (\ell_1 k_1 + \ell_2 k_2)} \, v_j(k), \; j \in \{1,2\}, \quad V_0(k) \, w := w \mbox{ if } w \perp \operatorname{Span} \{ v_1(k), v_2(k)\}. \end{equation}
Arguing as above, one shows that $V_0$ is periodic even if the basis is pseudo-periodic, that $V_0(k)$ commutes with $P_0(k)$ and that $T V_0(k) = V_0(-k) T$. Moreover, $\det(V(k))$ and $\det(V_0(k))$ have the same winding number(s) by construction, and hence $V$ and $V_0$ are homotopic by virtue of Theorem~\ref{thm:TRS_unitary_homotopy}. Once again, the statement then follows from Principle~\ref{pri:unitary_eq_to_homotopy}. 
\end{proof}

The above Theorem shows that homotopy classes of projection-valued maps in classes A, AI and AII are fully characterized by the same invariants appearing in the classification up to Murray--von Neumann equivalences (cf.\ Table~\ref{tabular:AZC_classes_MvN}), in particular, for $2$-dimensional projection-valued maps, by the Chern number in class A and by the Fu--Kane--Mele invariant in class AII.

\bigskip

As a further byproduct, Theorem~\ref{thm:Homotopies_AAIAII} also allows to identify homotopy classes of covariant projection-valued maps, in the sense of Definition~\ref{def:Symmetries}. Indeed, the considerations preceding the statement of Theorem~\ref{thm:TRS_unitary_homotopy} do not rely on the maps being unitary-valued. One can therefore apply the same conclusions to covariant projection-valued maps, using the decomposition from~\eqref{eq:decomposition_chiral_covariant} in presence of chiral symmetry, and one is again reduced to study homotopies among projection-valued maps without symmetry constraints or which are time-reversal symmetric. Together with Theorem~\ref{thm:Homotopies_AAIAII}, these considerations then yield the following

\begin{theorem} \label{thm:homotopy_iff_MvN_covariant}
In all symmetry classes, two covariant projection-valued maps over the $d$-dimensional torus, $d \le 2$, are homotopically equivalent if and only if they are Murray--von Neumann equivalent.
\end{theorem}

\subsection{Homotopies of pairs of projection-valued maps}

It remains to discuss homotopies of particle-hole and/or chiral symmetric pairs of rank-$n$ projection-valued maps $P^{\pm} \colon \mathbb{T}^{d} \to \Proj_n(\mathcal{H})$. We first analyze a simpler situation, namely that of projections of infinite rank.

\begin{theorem} \label{thm:homotopy_iff_MvN_infinite-rank}
Any two particle-hole and/or chiral symmetric pairs of infinite-rank projection-valued maps over the $d$-dimensional torus, $d \le 2$, are homotopically equivalent.
\end{theorem}
\begin{proof}
Consider two such pairs $P_{0}^{\pm}, P_{1}^{\pm} \colon \mathbb{T}^{d} \to \Proj_{\infty}(\mathcal{H})$: they are unitarily equivalent in view of Theorem~\ref{thm:unitary_eq_infinite-rank}. This unitary equivalence can be continuously deformed to the map which is constantly equal to~$\Id_{\mathcal{H}}$, as discussed in Section~\ref{sec:Homotopies_unitary}. The conclusion follows from Principle~\ref{pri:unitary_eq_to_homotopy}.
\end{proof}

Henceforth, we thus consider pairs of finite-rank projection-valued maps. As anticipated in the Introduction, the minimality of the ambient Hilbert space, that is, the condition for which the ranges of the orthogonal projections $P^{-}(k)$ and $P^{+}(k)$ together exhaust the whole~$\mathcal{H}$ (compare Definition~\ref{def:minimal}), influences the final results. We will therefore distinguish between the minimal and non-minimal settings.

\subsubsection{Non-minimal ambient Hilbert space}

In a non-minimal ambient Hilbert space, homotopy equivalence is again reduced to Murray--von Neumann equivalence. Notice that this setting covers in particular  finite-rank projections acting on an infinite-dimensional ambient space; the same proof presented below would also cover the case of infinite-rank projections onto proper infinite-dimensional subspaces of the ambient Hilbert space, for which we have however already proved Theorem~\ref{thm:homotopy_iff_MvN_infinite-rank}.

\begin{theorem} \label{thm:homotopy_iff_MvN_non-minimal}
Two particle-hole and/or chiral symmetric pairs of projection-valued maps over the $d$-dimensional torus, $d \le 2$, acting on a non-minimal ambient Hilbert space are homotopically equivalent if and only if they are Murray--von Neumann equivalent.
\end{theorem}
\begin{proof}
As before, the non-trivial implication is proving that Murray--von Neumann equivalent maps are also homotopic. Let $P_0^{\pm}, P_1^{\pm} \colon \mathbb{T}^{d} \to \Proj_n(\mathcal{H})$ be two Murray--von Neumann equivalent symmetric pairs of projection-valued maps, and assume that the span of the ranges of $P_0^{-}(k)$ and $P_0^{+}(k)$, as well as that of the ranges of $P_1^{-}(k)$ and $P_1^{+}(k)$, are proper subspaces of $\mathcal{H}$. Consider then the covariant complementary projection-valued maps $Q_0, Q_1 \colon \mathbb{T}^{d} \to \Proj_{\dim(\mathcal{H})-n}(\mathcal{H})$, as in Definition~\ref{def:Complementary}. Since $P_0^{\pm}$ and $P_1^{\pm}$ are Murray--von Neumann equivalent, they are unitarily equivalent by Theorem~\ref{thm:unitary_iff_MvN}: the unitary equivalence $V$ also implements Murray--von Neumann equivalence among the complementary maps $Q_0$ and $Q_1$. In turn, Theorem~\ref{thm:homotopy_iff_MvN_covariant} yields that they are homotopic via a homotopy $Q \colon [0,1] \times \mathbb{T}^{d} \to \Proj_{\dim(\mathcal{H})-n}(\mathcal{H})$, which is implemented by unitary conjugation via unitaries $V(t,k)$ such that $V(1,k) \equiv V(k)$ and $\left[V(0,k), Q_0(k)\right] \equiv 0$ for all $k \in \mathbb{T}^{d}$, in accordance with Principle~\ref{pri:unitary_eq_to_homotopy}. Upon close inspection of the proof of Theorem~\ref{thm:Homotopies_AAIAII}, on which Theorem~\ref{thm:homotopy_iff_MvN_covariant} relies, we see that $V(0,k)$ actually acts as the identity on the orthogonal complement to the range of $Q_0(k)$ (cf.~\eqref{eq:auxiliary_equation_30},~\eqref{eq:auxiliary_equation_31}, and~\eqref{eq:auxiliary_equation_32}): in particular, it commutes separately with $P_0^{-}(k)$ and $P_0^{+}(k)$ for $k \in \mathbb{T}^{d}$. This shows that $V \colon [0,1] \times \mathbb{T}^{d} \to \U(\mathcal{H})$ implements a homotopy of unitary-valued maps of the type prescribed by Principle~\ref{pri:unitary_eq_to_homotopy}, which then yields that $P_0^{\pm}$ and $P_1^{\pm}$ are homotopic, as desired.
\end{proof}

\subsubsection{Minimal ambient Hilbert space} \label{sec:Homotopy_minimal}

Finally, we discuss the homotopy classification for pairs of finite-rank projection-valued maps $P^\pm(k):\T^d \to \Proj_n(\Hi)$ in a minimal ambient Hilbert space of dimension $\dim(\Hi)=2n$, which belong to the seven AZC classes other than A, AI and AII. The minimality assumption justifies the identification $\mathcal{H} \simeq \C^{2n}$ upon choosing an appropriate basis of the ambient Hilbert space, which in turn allows to interpret the projection-valued maps $P^\pm$ as families of matrices in $\Proj_n(\C^{2n})$. With this minimality assumption, the homotopy classification of such families of matrices has been conducted in~\cite{gontier2022symmetric} in $d \le 1$, and topological indices characterizing the homotopy classes have been proposed. In the following, we will recover these results albeit at times in a slightly reformulated (but equivalent) version, which is suited for the extension to $d=2$. A more complete classification, valid in any dimension and which highlights the difference between \emph{strong} and \emph{weak} invariants (i.e.\ top-dimensional and lower-dimensional with respect to the CW complex decomposition of the $d$-dimensional torus), has been carried out in~\cite{Kennedy_2015}: it assumes minimality in chiral classes by a version of our Principle~\ref{pri:chiral_classes_min_dim}, but uses abstract algebraic-topological arguments to construct sets of equivariant homotopy classes of functions on $\mathbb{T}^{d}$ (with values in an appropriate classifying space depending on the AZC symmetry class) without discussing specific indices to be computed to identify the relevant homotopy invariants.

Once again, we'll carry out the analysis class by class, and highlight the new homotopy invariants that arise in each symmetry class. These invariants are not necessarily independent of each other or from other invariants appearing, for example, in the classification up to Murray--von Neumann equivalences; these possible further relations will be discussed in the next Section~\ref{sec:Properties_of_the_invariants}.

\medskip

\paragraph{\textsl{Class AIII}}

We choose an eigenbasis of the chiral symmetry operator $S$ for the ambient Hilbert space $\mathcal{H}$, so that the first vectors belong to $\mathcal{H}^{\uparrow} = \ker(S - \Id_{\mathcal{H}})$ and the last vectors belong to $\mathcal{H}^{\downarrow} = \ker(S + \Id_{\mathcal{H}})$. As dictated by Principle~\ref{pri:chiral_classes_min_dim}, the presence of a chiral pair of rank-$n$ projection-valued maps forces $\mathcal{H}^{\uparrow}$ and $\mathcal{H}^{\downarrow}$ to have the same dimension $n$, and in the decomposition $\C^{2n} \simeq \mathcal{H} = \mathcal{H}^{\uparrow} \oplus \mathcal{H}^{\downarrow} \simeq \C^n \oplus \C^n$ induced by the choice of the basis the chiral symmetry operator acts as
\begin{equation} \label{eq:S_normal_form}
S = \begin{pmatrix}
    \Id_n & 0 \\
    0 & - \Id_n
\end{pmatrix}.
\end{equation}
The same Principle also reduces the study of chiral pairs of projection-valued maps to that of unitary-valued maps $W \colon \mathbb{T}^{d} \to U(n)$. Indeed, looking at the block decomposition~\eqref{chiral_symm_P+-P-} for the reflection $P^{+} - P^{-} = \Id - 2 P^{-}$, it is clear that any homotopy $P^{-}(t, \cdot)$ induces a corresponding homotopy $P^{+}(t, \cdot) := S P^{-}(t, \cdot) S$ as well as a homotopy $W(t,\cdot)$ of the corresponding unitary block, $t \in [0,1]$; viceversa, any homotopy of the unitary block $W$ can be used to reconstruct one for the pair $P^{\pm}$ which is compatible with chiral symmetry.

The characterization of homotopy classes of unitary-valued maps has been carried out in Section~\ref{sec:Homotopies_unitary}. In particular, Theorem~\ref{thm:TRS_unitary_homotopy} motivates the following
\begin{definition}[Chiral invariant] \label{def:W_invariant}
Let $P^{\pm} \colon \mathbb{T}^{d} \to \Proj_n(\C^{2n})$ be a chiral symmetric pair of projection-valued maps. Let also $W \colon \mathbb{T}^{d} \to U(n)$ be as in~\eqref{chiral_symm_P+-P-}. We define the \emph{chiral invariant} of $P^{\pm}$ as follows:
\begin{itemize}
    \item if $d=0$, we set%
    \footnote{In this case, there are no homotopy classes to be distinguished. For the sake of a cleaner statement below, we still introduce a ``trivial'' invariant also for $0$-dimensional chiral pairs of rank-$n$ projections in $\C^{2n}$.}
    \[ \W(P^{\pm}) := 0; \]
    \item if $d=1$, we set
    \[ \W(P^{\pm}) := [\det W(\cdot)] \in \mathbb{Z}; \]
    \item if $d=2$, we set
    \[ \W(P^{\pm}) := \left( \left[\det W(\cdot) \big|_{\mathbb{T}^{1}_{1}}\right], \left[\det W(\cdot) \big|_{\mathbb{T}^{1}_{2}}\right] \right) \in \mathbb{Z} \times \mathbb{Z}, \]
    where $\mathbb{T}^{1}_{1}, \mathbb{T}^{1}_{2} \subset \mathbb{T}^{2}$ are as in~\eqref{eq:T^1_1,2}.
\end{itemize}
\end{definition}

Together with the previous considerations, an immediate application of Theorem~\ref{thm:TRS_unitary_homotopy} gives

\begin{theorem}[Homotopies of class-AIII pairs of projections on minimal spaces] \label{thm:Homotopy_AIII_minimal}
Two pairs of finite-rank projection-valued maps in class AIII on the $d$-dimensional torus, $d \le 2$, acting on a minimal ambient Hilbert space are homotopically equivalent if and only if their chiral invariants $\W(\cdot)$ coincide.
\end{theorem}

\begin{remark}[Homotopy vs Murray--von Neumann equivalence in chiral classes] \label{rmk:Hom_vs_MvN_in_chiral}
The previous Theorem, when compared with the results of Sections~\ref{sec:Bases_class_AIIICD} and~\ref{sec:Unitary_AIII} (in particular with Theorem~\ref{thm:unitary_iff_MvN}) and with Theorem~\ref{thm:homotopy_iff_MvN_non-minimal}, highlights a remarkable property of class-AIII pairs of projection-valued maps acting on minimal and non-minimal ambient spaces: their classification under Murray--von Neumann (or unitary) equivalence and under homotopy equivalence appear uncorrelated. Indeed, if the ambient Hilbert space is non-minimal, i.e.\ if $\dim(\mathcal{H}) > 2n$, then any two chiral symmetric pairs $P_0^{\pm}, P_1^{\pm} \colon \mathbb{T}^{d} \to \Proj_n(\mathcal{H})$ are homotopic if and only if they are Murray--von Neumann equivalent, that is, always when $d \le 1$, and if and only if $\Ch(P_0^{-}) = \Ch(P_1^{-}) \in \mathbb{Z}$ when $d=2$ (provided the rank $n$ is finite). In contrast, when $d=2$, $\dim(\mathcal{H}) = 2n$, and $P_0^{+}(k) + P_0^{-}(k) \equiv \Id_{\mathcal{H}}$, then on the one hand by additivity of the Chern number $\Ch(P_0^{+}) = - \Ch(P_0^{-})$; on the other hand, $P_0^{+}$ and $P_0^{-}$ are unitarily intertwined by the chiral symmetry operator by Definition~\ref{def:Symmetries}, and so $\Ch(P_0^{+}) = \Ch(P_0^{-})$. These identities together imply that
\[ \Ch(P_0^{+}) = \Ch(P_0^{-}) = 0, \]
and the same conclusion holds for any other pair $P_1^{\pm}$. So, on minimal ambient Hilbert spaces, all class-AIII pairs of projection-valued maps on the $d$-dimensional torus, $d \le 2$, are Murray--von Neumann equivalent to each other.

Theorem~\ref{thm:Homotopy_AIII_minimal} shows instead that two such pairs of projection-valued maps may very well be non-homotopic, at least if $d \in \{1,2\}$, as the chiral invariant characterizes their classification under homotopy equivalence: the class-AIII pairs therefore fall in countably-many distinct homotopy classes.

This decoupling between the classification under Murray--von Neumann equivalence and under homotopy equivalences in minimal vs non-minimal ambient spaces will be a recurring feature in all AZC classes comprising a chiral symmetry, as we'll see below.
\end{remark}

\medskip

\paragraph{\textsl{Classes C and D}}

We now treat symmetry classes with a single particle-hole symmetry, of even (class D) or odd (class C) type. Recall that an antiunitary symmetry operator $C$ admits the normal form $C = \mathcal{K}$, the complex conjugation operator, if it squares to $+\Id_{\mathcal{H}}$, and $C = \mathcal{K} J$, with $J$ the standard symplectic matrix, if it squares to $-\Id_{\mathcal{H}}$: compare~\eqref{eq:TRS_normal_form_even} and~\eqref{eq:TRS_normal_form_odd}. The identification $\mathcal{H} \simeq \C^{2n}$ will therefore be done using the basis in which the particle-hole symmetry operator assumes this normal form.

Let us start from the even case, class D. For this type of pairs of projections, we need to introduce a new homotopy invariant.

\begin{definition}[Even-particle-hole invariant] \label{def:P_invariant}
Let $P^{\pm} \colon \mathbb{T}^{d} \to \Proj_n(\C^{2n})$ be an even-par\-ti\-cle-hole symmetric pair of projection-valued maps. Set
\begin{equation} \label{eq:L(k)}
L(k) := i \left( P^{+}(k) - P^{-}(k) \right) = i \left( \Id_{\mathcal{H}}  - 2 P^{-}(k) \right) , \quad k \in \mathbb{T}^{d}. 
\end{equation}
We define the \emph{even-particle-hole invariant} of $P^{\pm}$ as follows:
\begin{itemize}
    \item if $d=0$, we set
    \[ \PI(P^{\pm}) := \Pf(L) \in \{\pm 1\} \simeq \mathbb{Z}_2; \]
    \item if $d=1$, we set
    \[ \PI(P^{\pm}) := \big( \Pf(L(0)), \Pf(L(\pi)) \big) \in \mathbb{Z}_2 \times \mathbb{Z}_2; \]
    \item if $d=2$, we set
    \[ \PI(P^{\pm}) := \big( \Pf(L(0,0)), \Pf(L(0,\pi)), \Pf(L(\pi,0)), \Pf(L(\pi,\pi)) \big) \in \mathbb{Z}_2 \times \mathbb{Z}_2 \times \mathbb{Z}_2 \times \mathbb{Z}_2. \]
\end{itemize}
\end{definition}

\begin{remark}[On the definition of the even-particle-hole invariant] \label{rmk:P_invariant}
By definition, the even-particle-hole invariant of a pair of projection-valued maps on the $d$-dimensional torus in class D is a collection of Pfaffians, computed from certain matrices defined at the $2^{d}$ fixed points $k_{\star}$ for the involution $k \mapsto -k$ on $\mathbb{T}^{d}$, $d \le 2$: compare Remark~\ref{rmk:Fixed_points}. The fact that, at any fixed point $k_{\star} \in \mathbb{T}^{d}$, the matrix $L_{\star} := L(k_{\star})$ is orthogonal and skew-symmetric is an easy consequence of the particle-hole symmetry constraint $\mathcal{K} P^{+}(k_{\star}) = P^{-}(k_{\star}) \mathcal{K}$ (i.e.\ $\overline{ P^{+}(k_{\star}) }= P^{-}(k_{\star})$) on the two orthogonal projections $P^{\pm}(k_{\star})$: compare~\cite[Lemma III.6]{gontier2022symmetric}. Since $L_{\star}$ is skew-symmetric and real-valued, it admits a Pfaffian $\Pf(L_{\star}) \in \mathbb{R}$ such that $\Pf(L_{\star})^2 = \det(L_{\star})$; and since $L_{\star}$ is orthogonal, its determinant $\det(L_{\star}) \in \{\pm 1\}$ is necessarily equal to $+1$, being a square. It follows that $\Pf(L_{\star}) \in \mathbb{Z}_2$ as claimed in Definition~\ref{def:P_invariant}.

Notice that, changing the basis in which the particle-hole symmetry operator acts as the complex conjugation or changing the order of the pair $P^{\pm} \equiv (P^{-}, P^{+})$ may result in an overall change of sign for the Pfaffians entering the definition of the even-particle-hole invariant. 
\end{remark}

We are ready to identify homotopy classes of finite-rank projections in class D.

\begin{theorem}[Homotopies of class-D pairs of projections on minimal spaces] \label{thm:Homotopy_D_minimal}
Two pairs of projection-valued maps $P_0^{\pm}, P_1^{\pm} \colon \mathbb{T}^{d} \to \Proj_n(\C^{2n})$, $d \le 2$, $n < \infty$, in class D are homotopically equivalent if and only if they are Murray--von Neumann equivalent and $\PI(P_0^{\pm}) = \PI(P_1^{\pm})$.
\end{theorem}
\begin{proof}
The fact that the even-particle-hole invariant is indeed a homotopy invariant follows from the fact that, by definition, it is a continuous function of the pair of projection-valued maps, and takes values in the discrete space $\mathbb{Z}_2$: along continuous deformations (i.e.\ homotopies) of its argument, it must therefore stay constant. Recall now that, in view of the results of Sections~\ref{sec:Bases_class_AIIICD} and~\ref{sec:Unitary_CD}, the Chern number $\Ch(P^{-}) \in \mathbb{Z}$ characterizes the Murray--von Neumann and unitary equivalence classes of a pair of projection-valued maps in class D and in $d=2$: that the Chern number is also a homotopy invariant was stated in Remark~\ref{rmk:Chern}. 

We need to show that, conversely, two pairs of projection-valued maps $P_0^{\pm}, P_1^{\pm} \colon \mathbb{T}^{d} \to \Proj_n(\C^{2n})$ with the same Chern numbers (if $d=2$) and even-particle-hole invariants are homotopic. As usual, we resort to Principle~\ref{pri:unitary_eq_to_homotopy}. Notice first of all that, thanks to Theorem~\ref{thm:unitary_iff_MvN} and the hypothesis that $P_0^{\pm}$ and $P_1^{-}$ are Murray--von Neumann equivalent, there exists a unitary equivalence $V \colon \mathbb{T}^{d} \to U(2n)$ intertwining $P_0^{\pm}$ and $P_1^{\pm}$. We will now show that the equality of the even-particle-hole invariants is the condition that guarantees that $V$ can be deformed continuously to a unitary-valued map commuting with $P_0^{\pm}$.

Indeed, for $k \in \mathbb{T}^{d}$, the unitary operator $V(k)$ is such that $P_1^\pm(k)=V(k)P_0^\pm(k)V(k)^{-1}$ and satisfies the particle-hole symmetry constraint $\K V(k)=V(-k)\K$. At the fixed points $k_{\star} = -k_{\star} \in \mathbb{T}^{d}$, these two identities force
\[ L_1(k_{\star}) = V(k_{\star}) L_0(k_{\star}) V(k_{\star})^{\mathrm{t}} \quad \text{and} \quad V(k_\star) \in O(2n), \]
where $L_s(k)$ is as in~\eqref{eq:L(k)} with $P^{\pm} = P_s^{\pm}$, $s \in \{0,1\}$. Since the Pfaffian changes under congruences as $\Pf(BAB^{\mathrm{t}})=\Pf(A) \det(B)$, we conclude that
\[ \Pf\big( L_1(k_{\star}) \big) = \Pf\big( L_0(k_{\star}) \big) \, \det(V(k_{\star})). \]
If the two pairs $P_0^{\pm}$ and $P_1^{\pm}$ have the same even-particle-hole invariant, by its Definition~\ref{def:P_invariant} the Pfaffians on both sides of the above identity are equal, yielding that $\det(V(k_\star))=1$ on all fixed points $k_{\star} \in \mathbb{T}^{d}$. 

After these general considerations, we now specialize the dimension. If $d=0$, then $V$ with $\det(V)=1$ is in the path-connected component of the identity in $O(2n)$, and can therefore be continuously deformed to $\Id_{2n}$, which certainly commutes with $P_0^{\pm}$. If instead $d=1$, to study the homotopy class of $V \colon \mathbb{T}^{1} \to U(2n)$ we invoke Theorem~\ref{thm:TRS_unitary_homotopy}: since $\det(V(0)) = 1 = \det(V(\pi))$ and $\det(V(-k)) = \overline{\det(V(k))}$ by the particle-hole symmetry constraint~\eqref{MvNsymmetry}, the winding number of $\det(V(\cdot)) \colon \mathbb{T}^{1} \to U(1)$ is even, say equal to $2 \ell \in 2\mathbb{Z}$. If $d=2$, the same conclusion holds for the winding numbers of $\det(V(\cdot)) \big|_{\mathbb{T}^{1}_{1}}$ and $\det(V(\cdot)) \big|_{\mathbb{T}^{1}_{2}}$, which are then both even, say equal to $2 \ell_1 \in 2\mathbb{Z}$ and $2 \ell_2 \in 2\mathbb{Z}$, respectively. Let now $\{v_j(k)\}_{j \in \{1,\ldots,2n\}}$ be a particle-hole symmetric pseudo-periodic basis for $P_0^{\pm}$, whose existence is guaranteed by Proposition~\ref{prop:Class_AIIICD_frames}. Define $V_0(k) \in U(2n)$ as the linear operator such that
\[ V_0(k)\, v_j(k) := 
\begin{cases}
    e^{i \ell \cdot k} v_j(k) & \text{if } j \in \{1,2n\},\\
    v_j(k) & \text{otherwise},
\end{cases} \]
where $\ell \cdot k := \ell k$ in $d=1$ and $\ell \cdot k := \ell_1 k_1 + \ell_2 k_2$ in $d=2$. The map $V_0 \colon \mathbb{T}^{d} \to U(2n)$ is periodic and particle-hole symmetric in view of Principle~\ref{pri:Periodicity_from_pseudo-periodicity}, commutes with $P_0^{\pm}$ by construction, and is homotopic to $V$ thanks to Theorem~\ref{thm:TRS_unitary_homotopy}. Principle~\ref{pri:unitary_eq_to_homotopy} concludes the proof.
\end{proof}
    
For class C, the proof above can be applied without resorting to additional topological invariants. Indeed, on a fixed point $k_{\star} = - k_{\star} \in \mathbb{T}^{d}$, the odd-particle-symmetry constraint on a unitary equivalence between a pair of projection-valued maps reads $\K J V(k_\star) = V(k_\star) \mathcal{K} J$ or $V(k_{\star})^{\mathrm{t}} J V(k_{\star}) = J$, meaning that $V(k_\star)$ is symplectic and therefore has unit determinant. We conclude that

\begin{theorem}[Homotopies of class-C pairs of projections on minimal spaces] \label{thm:Homotopy_C_minimal}
Two pairs of finite-rank projection-valued maps in class C on the $d$-dimensional torus, $d \le 2$, acting on a minimal ambient Hilbert space are homotopically equivalent if and only if they are Murray--von Neumann equivalent.
\end{theorem}

\medskip

\paragraph{\textsl{Classes BDI and CII}}

In these classes, the symmetries commute, and we can impose normal forms for all of them in the same basis: \eqref{eq:TRS_normal_form_even} and~\eqref{eq:TRS_normal_form_odd} for antiunitary symmetries, and~\eqref{eq:S_normal_form} for the unitary chiral symmetry. In this basis, the minimal ambient Hilbert space is then identified as $\mathcal{H} \simeq \C^{2n}$. As already noticed, in presence of an odd antiunitary symmetry, $n$ is itself necessarily even.

\begin{theorem}[Homotopies of class-BDI pairs of projections on minimal spaces] \label{thm:Homotopy_BDI_minimal}
Two pairs of finite-rank projection-valued maps in class BDI on the $d$-dimensional torus, $d \le 2$, acting on a minimal ambient Hilbert space are homotopically equivalent if and only if both their chiral invariants $\W(\cdot)$, which take values in even integers, and their even-particle-hole invariants $\PI(\cdot)$ coincide.
\end{theorem}
\begin{proof}
Notice how in class BDI both chiral symmetry and an even particle-hole symmetry are present, so that one can define the corresponding invariants as in the statement, and both are homotopy invariants. These two symmetry operators will be taken of the form
\[ S=\begin{pmatrix}
        \Id_n & 0 \\ 0 & -\Id_n
    \end{pmatrix}, \quad C=\mathcal{K}, \]
compare~\eqref{eq:S_normal_form} and~\eqref{eq:TRS_normal_form_even}. For a map $P^{\pm}$ in class BDI (in particular, chiral symmetric), Principle~\ref{pri:chiral_classes_min_dim} still applies and yields the following block decomposition:
\begin{equation} \label{eq:L(k)_block}
L(k) = i \begin{pmatrix}
0 & W(k) \\ W(k)^* & 0 
\end{pmatrix},    
\end{equation}
where $L(k)$ is as in~\eqref{eq:L(k)} and $W \colon \mathbb{T}^{d} \to U(n)$ is a unitary-valued map. The particle-hole symmetry satisfied by $P^{\pm}$ further implies
\[ C L(k) = L(-k) C \iff \overline{L(k)} = - L(-k) \iff \overline{i W(k)} = i W(-k) \quad \forall\: k \in \mathbb{T}^{d}. \]
In what follows, we'll denote $V(k) := i W(k)$, $k \in \mathbb{T}^{d}$, the unitary off-diagonal block associated to the matrix $L(k)$ (and thus, in turn, to the pair of projection-valued maps $P^{\pm}$). Then, the condition $\overline{V(k)} = V(-k)$ qualifies $V \colon \mathbb{T}^{d} \to U(n)$ as a time-reversal symmetric unitary-valued map with respect to an even time-reversal symmetry given by the complex conjugation $\mathcal{K}$.

As before, we only need to prove that maps $P_0^{\pm}, P_1^{\pm} \colon \mathbb{T}^{d} \to \Proj_n(\C^{2n})$ with the same invariants are homotopic. In view of the previous considerations and of the discussion in class AIII, in order to do so it suffices to construct a corresponding homotopy between their unitary-valued off-diagonal blocks $W_0$ and $W_1$, or, equivalently, between $V_0 = i W_0$ and $V_1 = i W_1$, which further preserves the constraint coming from particle-hole symmetry, namely the time-reversal symmetry constraint on the $V$'s discussed above. We show that, under the hypothesis that $\W(P_0^{\pm}) = \W(P_1^{\pm})$ and $\PI(P_0^{\pm}) = \PI(P_1^{\pm})$, this can be achieved. We will rely on Theorem~\ref{thm:TRS_unitary_homotopy}, which characterizes homotopy classes of (time-reversal symmetric) unitary-valued maps.

Let first $d=0$. The condition for $V_0$ and $V_1$ to be homotopic (i.e.\ path-connected) as even-time-reversal symmetric unitary matrices is that $\det(V_0) = \det(V_1)$. By Definition~\ref{def:P_invariant}, we have
\[ \det(V_0) = \Pf(L_0)^2 = \PI(P_0^{\pm})^2 = \PI(P_1^{\pm})^2 = \Pf(L_1) = \det(V_1), \]
with an obvious meaning of $L_0$ and $L_1$, and we are done. In particular, both determinants are equal to $1$.

Let now $d=1$. The condition for $V_0$ and $V_1$ to be homotopic as even-time-reversal symmetry unitary-valued maps on $\mathbb{T}^{1}$ is that the winding numbers of their determinants coincide. By Definition~\ref{def:W_invariant},  we have
\[ [\det(V_0(\cdot))] = [\det(W_0(\cdot))] = \W(P_0^{\pm}) = \W(P_1^{\pm}) = [\det(W_1(\cdot))] = [\det(V_1(\cdot))] \]
so that the conclusion follows by the equality of the chiral invariants. Notice that the equality of the even-particle-hole invariants implies this time that $\det(V_0(k_{\star})) = 1 = \det(V_1(k_{\star}))$ at the fixed points $k_{\star} \in \{0, \pm \pi\}$, so that the winding numbers above are also even: the determinant is already periodic on the sub-interval $[0,\pi] \subset \mathbb{T}^{1}$, and the time-reversal symmetry condition implies $\overline{\det(V(k))} = \det(V(-k))$, so that the two ``halves'' of the torus $[-\pi,0]$ and $[0,\pi]$ give the same integer contribution to the winding number.

Finally, if $d=2$, a similar argument holds when considering winding numbers for the restriction of the determinats of $V_0$ and $V_1$ on the two $1$-dimensional sub-tori $\mathbb{T}^{1}_{1}, \mathbb{T}^{1}_{2} \subset \mathbb{T}^{2}$, as those winding numbers compute the chiral invariant of the pairs of projection-valued maps. Since the equality of these winding numbers is the condition which guarantees that the two time-reversal symmetric unitary-valued maps are homotopic, this concludes the proof.
\end{proof}

\begin{remark}
It appears from the proof that, as far as homotopies among class-BDI pairs of maps are concerned, the even-particle-hole invariant plays a role in $d=0$, while the chiral invariant plays a role in $d=1$. Requiring that both chiral invariants and even-particle-hole invariant agree in all dimensions $d \le 2$ is indeed redundant, and the two sets of invariants are related among each other. This redundancy will be explored in Section~\ref{sec:Properties_of_the_invariants}.
\end{remark}

Coming to class CII, the above proof can be replicated with minor modifications. First of all, the relevant symmetries here are commuting chiral and odd time-reversal symmetries, acting on the same basis as
\[ S=\begin{pmatrix}
        \Id_n & 0 \\ 0 & -\Id_n
    \end{pmatrix}, \quad T=\K \begin{pmatrix}
        J&0\\0 &J
    \end{pmatrix}, \] 
where $J$ is the standard $n \times n$ symplectic matrix: compare~\eqref{eq:S_normal_form} and~\eqref{eq:TRS_normal_form_odd}. This time, the time-reversal symmetric constraint $T P^\pm(k) = P^\pm(-k) T$ implies that $T(P^+(k)-P^-(k))=(P^+(-k)-P^-(-k))T$. It follows then from Principle~\ref{pri:chiral_classes_min_dim} that
    \begin{equation}\label{eq:Class_CII_symmetry_structure}
        \begin{pmatrix}
            J&0\\ 0 & J
        \end{pmatrix}\overline{\begin{pmatrix}
            0 & W(k) \\ W(k)^* & 0 
        \end{pmatrix}} = \begin{pmatrix}
            0 & W(-k) \\ W(-k)^* & 0
        \end{pmatrix} \begin{pmatrix}
            J&0\\ 0 & J
        \end{pmatrix}.
    \end{equation}
Hence the unitary-valued map $W:\T^d \to U(n)$ must satisfy $J\overline{W(k)}=W(-k)J$, which is indeed an odd-time-reversal symmetry constraint. Comparing with the classification up to homotopy equivalence for such maps established in Theorem~\ref{thm:TRS_unitary_homotopy}, one can conclude the following

\begin{theorem}[Homotopies of class-CII pairs of projections on minimal spaces] \label{thm:Homotopy_CII_minimal}
Two pairs of finite-rank projection-valued maps in class CII on the $d$-dimensional torus, $d \le 2$, acting on a minimal ambient Hilbert space are homotopically equivalent if and only if their chiral invariants $\W(\cdot)$, which take values in even integers, coincide.
\end{theorem}

\medskip

\paragraph{\textsl{Classes CI and DIII}} 

In these classes, all three types of symmetries are present, but the symmetry operators anticommute among each other. In both AZC classes, we will choose to work in the eigenbasis of the chiral symmetry operator $S$, namely to represent it as in~\eqref{eq:S_normal_form} in the form
\[ S=\begin{pmatrix}
        \Id_n & 0 \\ 0 & -\Id_n
    \end{pmatrix} \]
and use this basis to identify $\mathcal{H} \simeq \C^{2n}$. For class CI, we focus on the anticommuting even time-reversal symmetry operator $T$, which antiunitarily intertwines the two eigenspaces of~$S$: as such, it can be assumed to take the form
\[ T=\mathcal{K} \begin{pmatrix}
        0 & \Id_n \\ \Id_n & 0
    \end{pmatrix}, \]
where $\mathcal{K}$ denotes the complex conjugation operator, compare~\eqref{eq:TRS_normal_form_even}. For class DIII, where the time-reversal symmetry operator is odd instead, we will assume 
\[ T=\K J \]
where $J$ is the standard $2n \times 2n$ symplectic matrix, compare~\eqref{eq:TRS_normal_form_odd}. Recall that $n$ itself must be even in this case, as each element of a pair of symmetric rank-$n$ projection-valued map on $\C^{2n}$ is separately an odd-time-reversal symmetric projection-valued map. Under these assumptions and with the usual convention $S=TC$, it is rather the even particle-hole symmetry operator of class DIII that takes the form
\begin{equation} \label{C_DIII} 
C = \mathcal{K} \begin{pmatrix}
        0 & \Id_n \\ \Id_n & 0
\end{pmatrix}.
\end{equation}

\begin{theorem}[Homotopies of class-CI pairs of projections on minimal spaces] \label{thm:Homotopy_CI_minimal}
Any two pairs of finite-rank projection-valued maps in class CI on the $d$-dimensional torus, $d \le 2$, acting on a minimal ambient Hilbert space are homotopically equivalent.
\end{theorem}
\begin{proof}
Let $P^{\pm} \colon \mathbb{T}^{d} \to \Proj_n(\C^{2n})$ be a pair of projection-valued maps in class CI. Invoking Principle~\ref{pri:chiral_classes_min_dim}, we look at the combination $P^+-P^-$, which satisfies the time-reversal symmetric constraint $T(P^+(k)-P^-(k))=(P^+(-k)-P^-(-k))T$, $k \in \mathbb{T}^{d}$. At the level of the unitary-valued map $W:\T^d\to U(n)$ from~\eqref{chiral_symm_P+-P-}, this constraint reads
\begin{equation} \label{eq:W^t=W} 
\overline{W(k)} = W(-k)^* \iff W(k)^{\mathrm{t}} = W(-k) \quad \forall\: k \in \mathbb{T}^{d}. 
\end{equation} 
Arguing as we did for class AIII, the statement follows if we can prove that any two unitary-valued maps satisfying the above condition are homotopic to each other%
\footnote{Compare~\cite[Theorem IV.5]{gontier2022symmetric} for a different proof in $d \le 1$.}.

Let $d=0$. The constraint~\eqref{eq:W^t=W} then reads $W^{\mathrm{t}} = W$, that is, $W$ is a symmetric unitary matrix. By the spectral theorem, it can be written as $W = e^{i L}$, where $L^{\mathrm{t}} = L$ is a complex symmetric matrix%
\footnote{Write the spectral decomposition of $W$ as $W = \sum_j e^{i \lambda_j} E_j$, where $\lambda_j \in \mathbb{R}$ and the spectral eigenprojections $E_j$ are mutually orthogonal for all $j$. The constraint $W^{\mathrm{t}} = W$ implies
\[ \sum_j e^{i \lambda_j} E_j^{\mathrm{t}} = \sum_j e^{i \lambda_j} E_j \iff E_j^{\mathrm{t}} = E_j \; \forall\: j. \]
Setting 
\[ L := \sum_j \lambda_j E_j = L^{\mathrm{t}} \]
one gets that $W = e^{i L}$ by functional calculus, as desired.}%
. Any two such matrices $W_0 = e^{i L_0}$ and $W_1 = e^{i L_1}$ are then homotopic via
\[ W \colon [0,1] \to U(n), \quad W(t) := \exp \big(i \big[(1-t) L_0 + t L_1\big]\big), \; t \in [0,1], \]
with $W(t)^{\mathrm{t}} = W(t)$ for all $t \in [0,1]$.

Let now $d=1$. Observe that~\eqref{eq:W^t=W} implies that $\det(W(k)) = \det(W(-k))$ for all $k \in \mathbb{T}^{1}$, so that in particular the winding number of $\det(W(\cdot))$ is zero as the two ``halves'' of the torus $[-\pi,0]$ and $[0,\pi]$ give opposite contributions to the winding itself; stated differently, the chiral invariant $\W(\cdot)$ of a pair of projection-valued maps in class CI vanishes. Consequently, $W$ can be continuously deformed to the map constantly equal to $\Id_{n}$ according to Theorem~\ref{thm:TRS_unitary_homotopy}. This is however not enough, as we also want the homotopy to preserve the constraint~\eqref{eq:W^t=W} along the whole deformation. We have therefore to argue differently.

Let $W_0, W_1 \colon \mathbb{T}^{1} \to U(n)$ be two unitary-valued maps which satisfy~\eqref{eq:W^t=W}. At the two fixed points $k_{\star} \in \{0, \pi\}$ for the involution $k \mapsto -k$ on $\mathbb{T}^{1}$, the matrices $W_0(k_{\star}), W_1(k_{\star})$ are unitary and symmetric, and by what was argued above they are connected by a path of symmetric and unitary matrices $\tilde{W}(t,k_{\star}) = e^{i L(t, k_{\star})}$, $L(t,k_{\star}) = L(t,k_{\star})^{\mathrm{t}}$, $t \in [0,1]$. We have thus defined a continuous map $\tilde{W} \colon \partial ([0,1] \times [0,\pi]) \to U(n)$; notice that, topologically, the boundary of the rectangle $[0,1] \times [0,\pi]$ is a circle, so it makes sense to compute $\tilde{w} := [\det(\tilde{W}(\cdot))] \in \mathbb{Z}$. Define
\[ W(t, 0) := \exp \left( i \left[ L(0,k_{\star}) + t \begin{pmatrix} -2 \pi \tilde{w} & 0 \\
0 & \Id_{n-1} \end{pmatrix} \right] \right), \quad t \in [0,1]. \]
The matrix $W(t, 0)$ is symmetric and unitary, depends continuously on $t \in [0,1]$, and is such that $W(0,0) = \tilde{W}(0,0) = W_0(0)$ and $W(1,0) = \tilde{W}(1,0) = W_1(0)$. With this redefinition on the side of the rectangle $\{(t,0): t \in [0,1]\} \subset \partial( [0,1] \times [0,\pi] )$, the map $W \colon \partial ([0,1] \times [0,\pi]) \to U(n)$ has vanishing winding of the determinant, and can therefore be extended continuously to the interior of the rectangle: we denote this extension again as $W \colon [0,1] \times [0,\pi] \to U(n)$, and further extend it to $[0, 1] \times [-\pi,\pi]$ by setting
\[ W(t, -k) := W(t, k)^{\mathrm{t}} \quad \forall\:k \in [0,\pi]. \]
Observe now that this definition yields
\[ W(t, -\pi) = W(t, \pi)^{\mathrm{t}} = W(t,\pi) \]
in view of the symmetry of $W(t,\pi) = \tilde{W}(t,\pi)$, $t \in [0,1]$. Therefore, the map $W(t,\cdot)$ is periodic over $[-\pi,\pi]$, and thus yields a well-defined symmetric homotopy $W \colon [0,1] \times \mathbb{T}^{1} \to U(n)$ between $W_0 \equiv W(0,\cdot)$ and $W_1 \equiv W(1,\cdot)$, as desired.

Finally, let $d=2$. Given two unitary-valued maps $W_0, W_1 \colon \mathbb{T}^{2} \to U(n)$ satisfying~\eqref{eq:W^t=W}, consider their restrictions to 
\begin{equation} \label{eqn:subtori} 
\mathbb{T}^{1}_{j,\,k_{\star}} := \{(k_1,k_2) \in \mathbb{T}^{2} : k_j = k_{\star}\}, \quad j \in \{1, 2\}, \; k_{\star} \in \{0,\pi\}.
\end{equation}
Due to the periodicity in $k_1$, it is easily realized that the restrictions $W_0 \big|_{\mathbb{T}^{1}_{1,\,k_{\star}}}$ and $W_1 \big|_{\mathbb{T}^{1}_{1,\,k_{\star}}}$ both satisfy~\eqref{eq:W^t=W}, as unitary-valued maps on the $1$-dimensional torus $\mathbb{T}^{1}_{1,\,k_{\star}}$, $k_{\star} \in \{0,\pi\}$. From what was discussed before in $d=1$, $W_0 \big|_{\mathbb{T}^{1}_{1,\,0}}$ is then homotopic to $W_1 \big|_{\mathbb{T}^{1}_{1,\,0}}$ and $W_0 \big|_{\mathbb{T}^{1}_{1,\,\pi}}$ is homotopic to $W_1 \big|_{\mathbb{T}^{1}_{1\,,\pi}}$. The same argument holds for the restrictions $W_0 \big|_{\mathbb{T}^{1}_{2\,,\pi}}$ and $W_1 \big|_{\mathbb{T}^{1}_{2\,,\pi}}$, which are then also homotopic as maps that satisfy~\eqref{eq:W^t=W}. The collection of all these homotopies yields a map $W \colon \partial \Pi \to U(n)$, where%
\footnote{The set $[0,\pi] \times \mathbb{T}^{1}$, consisting of ``half'' of the torus, is sometimes called the \emph{effective Brillouin zone}~\cite{moore2007topological}. It gives a set of representatives for the quotient $\mathbb{T}^{2} / (k \sim -k)$.}
\[ \Pi := [0,1] \times [0,\pi] \times \mathbb{T}^{1} \simeq [0,1] \times [0,\pi] \times [-\pi,\pi] / ((t,k_1,-\pi) \sim (t,k_1,\pi)). \]
Topologically, the boundary of the parallelepiped $[0,1] \times [0,\pi] \times [-\pi,\pi]$ is a $2$-sphere: as the second homotopy group of the space of unitary matrices is trivial, $\pi_2(U(n)) = \{0\}$ \cite[Chapter 8, Section 12]{Husemoller1994}, the map $W$ extends continuously to $\Pi$ (and again we denote the extension with the same symbol). Extend further the map to $[0,1] \times [-\pi,\pi] \times \mathbb{T}^{1}$ by setting
\[ W(t, -k_1, k_2) := W(t, k_1, -k_2)^{\mathrm{t}} \quad \forall\:k_1 \in [0,\pi], \; k_2 \in \mathbb{T}^{1}. \]
This extension is continuous at $k_1=0$ since the homotopy $W(t,0,\cdot)$, $t \in [0,1]$, connecting $W_0 \big|_{\mathbb{T}^{1}_{1,\,0}}$ and $W_1 \big|_{\mathbb{T}^{1}_{1,\,0}}$ preserves the constraint~\eqref{eq:W^t=W}, and it is periodic in $k_1$ (i.e.\ $W(t,-\pi,k_2) = W(t,\pi,k_2)$ for all $(t,k_2) \in [0,1] \times \mathbb{T}^{1}$) because so does the homotopy $W(t,\pi,\cdot)$, $t \in [0,1]$, connecting $W_0 \big|_{\mathbb{T}^{1}_{1,\,\pi}}$ and $W_1 \big|_{\mathbb{T}^{1}_{1,\,\pi}}$. The resulting map $W \colon [0,1] \times \mathbb{T}^{2} \to U(n)$ therefore exhibits the desired homotopy between $W_0$ and $W_1$.
\end{proof}

Before stating and proving the corresponding statement for class DIII, we make a few preliminary considerations. Arguing as at the beginning of the proof of Theorem~\ref{thm:Homotopy_CI_minimal}, one reduces the study of pairs of projection-valued maps in class DIII to that of unitary-valued maps $W \colon \mathbb{T}^{d} \to U(n)$ via Principle~\ref{pri:chiral_classes_min_dim}. This time, it is the particle-hole symmetry operator~$C$ that takes the form~\eqref{C_DIII}; in turn, the particle-hole symmetry constraint for the map $P^{+}-P^{-}$ reads this time $C(P^+(k)-P^-(k))=-(P^+(-k)-P^-(-k))C$, $k \in \mathbb{T}^{d}$. At the level of the unitary-valued map $W$ from~\eqref{chiral_symm_P+-P-}, this condition yields
\begin{equation} \label{eq:W^t=-W}
W(k)^{\mathrm{t}}=-W(-k) \quad \forall\: k \in \mathbb{T}^{d}.
\end{equation}
Notice how the above relation, together with the fact that $n$ is necessarily even, imply that $\det(W(k)) = (-1)^n \, \det(W(-k)) = \det(W(-k))$, so that, as in the previous proof, the map $\det(W(\cdot))$ has no winding along any cardinal direction in the torus: in particular, also in class DIII the chiral invariant from Definition~\ref{def:W_invariant} vanishes in any dimension $d \le 2$.

If $d=0$, a unitary-valued matrix satisfying~\eqref{eq:W^t=-W} is also skew-symmetric, and therefore admits a decomposition~\cite[Corollary A.4]{gontier2022symmetric}
\begin{equation} \label{eq:DIIId0}
 W = V^{\mathrm{t}} J V, \quad V \in U(n),
\end{equation}
where $J$ is the $n \times n$ standard symplectic matrix. If $d=1$, then the homotopy classes of unitary-valued maps satisfying~\eqref{eq:W^t=-W} have been studied in~\cite[Theorem 5.6]{Peluso2026}, and are fully characterized by the Teo--Kane invariant introduced in Remark~\ref{def:gamma_invariant}. This, together with the next statement, motivate the following definition for the homotopy invariant (compare also~\cite{de2022cohomology} for a cohomological study of the properties of projection-valued maps in this class).

\begin{definition}[DIII invariant] \label{def:DIII_invariant}
Let $P^{\pm} \colon \mathbb{T}^{d} \to \Proj_n(\C^{2n})$ be a pair of projection-valued maps in class DIII. Let $W \colon \mathbb{T}^{d} \to U(n)$ be as in~\eqref{chiral_symm_P+-P-}. We define the \emph{DIII invariant} of $P^{\pm}$ as follows:
\begin{itemize}
    \item if $d=0$, we set%
    \footnote{Compare the footnote to Definition~\ref{def:W_invariant}.}
    \[ \G(P^{\pm}) := 0; \]
    \item if $d=1$, we set
    \[ \G(P^{\pm}) := \TK(W(\cdot)) \in \mathbb{Z}_2; \]
    \item if $d=2$, we set
    \[ \G(P^{\pm}) := \left( \TK\left(W(\cdot) \big|_{\mathbb{T}^{1}_{1,\,\pi}}\right), \TK\left(W(\cdot) \big|_{\mathbb{T}^{1}_{2,\,\pi}}\right), \TK\left(W(\cdot) \big|_{\mathbb{T}^{1}_{\textup{diag}}}\right) \right) \in \mathbb{Z}_2 \times \mathbb{Z}_2 \times \mathbb{Z}_2 \]
    where $\mathbb{T}^{1}_{1,\,\pi}, \mathbb{T}^{1}_{2,\,\pi} \subset \mathbb{T}^{2}$ are as in~\eqref{eqn:subtori}, and where
    \[ \mathbb{T}^{1}_{\textup{diag}}:= \{ (k_1, k_2) \in \mathbb{T}^{2} : k_2 = -k_1 \} \subset \mathbb{T}^{2}. \]
\end{itemize}
\end{definition}

\begin{theorem}[Homotopies of class-DIII pairs of projections on minimal spaces] \label{thm:Homotopy_DIII_minimal}
Two pairs of finite-rank projection-valued maps in class DIII on the $d$-dimensional torus, $d \le 2$, acting on a minimal ambient Hilbert space are homotopically equivalent if and only if their DIII invariants $\G(\cdot)$ coincide.
\end{theorem}
\begin{proof}
As in the previous proof, it suffices to show that the corresponding unitary-valued maps $W_0, W_1 \colon \mathbb{T}^{d} \to U(n)$ from~\eqref{chiral_symm_P+-P-}, which satisfy the constraint~\eqref{eq:W^t=-W}, are homotopic, assuming the pairs of projection-valued maps have the same DIII invariants.

Let $d=0$. The decomposition~\eqref{eq:DIIId0} yields the fact that the space of skew-symmetric unitary matrices is path-connected, as one can connect the corresponding matrices $V \in U(n)$. Any two such matrices $W_0 = -W_0^{\mathrm{t}} \in U(n)$ and $W_1 = -W_1^{\mathrm{t}} \in U(n)$ are then homotopic.

In $d=1$, it was already recalled that the Teo--Kane invariant characterizes homotopy classes of maps $W \colon \mathbb{T}^{1} \to U(n)$ satisfying~\eqref{eq:W^t=-W}. By definition of the DIII invariant, this implies the existence of a homotopy between the corresponding pairs of projection-valued maps.

Finally, in $d=2$, one can perform a construction similar to the one from the previous proof. Notice first of all that the restrictions of the two maps $W_0, W_1 \colon \mathbb{T}^{2} \to U(n)$ to the sub-tori $\mathbb{T}^{1}_{1,\,\pi}$, $\mathbb{T}^{1}_{2,\,\pi}$ and $\mathbb{T}^{1}_{\textup{diag}}$, as in Definition~\ref{def:DIII_invariant}, all satisfy~\eqref{eq:W^t=-W}, as the three sub-tori are fixed by the involution $k \mapsto -k$. The condition for the restrictions of $W_0$ and $W_1$ to each sub-torus to be homotopic is exactly the equality of their Teo--Kane invariants, as proved in the discussion of the case $d=1$: this motivates the definition of the DIII invariant in dimension $d=2$. Assuming the equality of the DIII invariants of $P^{\pm}_0$ and $P^{\pm}_1$, one can therefore define a map $W \colon \partial \Sigma \to U(n)$, where
\[\Sigma := [0,1] \times \triangle \quad \text{with} \quad \triangle := \{ (k_1, k_2) \in \mathbb{T}^{2} : k_2 \ge -k_1 \}. \]
Notice how $\triangle$ gives a set of representatives for the quotient $\mathbb{T}^{2} / (k \sim -k)$ of the $2$-torus under the usual involution. Since the boundary of the prism $\Sigma$ is topologically a $2$-sphere, one can extend the map $W$ to its interior as in the previous proof, and then further to the whole $[0,1] \times \mathbb{T}^{2}$ by setting
\[ W(t,-k) := -W(t, k)^{\mathrm{t}}, \quad (t,k) \in \Sigma = [0,1] \times \triangle. \]
Continuity and periodicity of this extension are guaranteed by the properties of the homotopies among the restrictions of $W_0$ and $W_1$ to the sub-tori $\mathbb{T}^{1}_{1,\,\pi}$, $\mathbb{T}^{1}_{2,\,\pi}$ and $\mathbb{T}^{1}_{\textup{diag}}$. This produces the required homotopy between $W_0$ and $W_1$, and concludes the proof.
\end{proof}

\begin{remark} \label{rmk:DIII_no_invariants}
In class DIII, it is possible to compute the even-particle-hole invariant for a pair of projection-valued-maps, as in Definition~\ref{def:P_invariant}, and one may ask whether they have any role in the homotopy classification of such maps. We will see in Section~\ref{sec:Properties_of_the_invariants} that these invariants actually yield the same value for all pairs of projection-valued maps, and thus do not contain any relevant topological information.
\end{remark}

\section{Properties of the homotopy invariants} \label{sec:Properties_of_the_invariants}

In this Section, we discuss the properties of the homotopy invariants arising from the classification of pairs of projection-valued maps on minimal ambient Hilbert spaces: the chiral invariant (Definition~\ref{def:W_invariant}), the even-particle-hole invariant (Definition~\ref{def:P_invariant}) and the DIII invariant (Definition~\ref{def:DIII_invariant}). Specifically, we will first of all elucidate the possible redundancies among them and with the invariants arising in the corresponding classification up to Murray--von Neumann equivalences, in order to identify in each AZC class a minimal set of \emph{complete} homotopy invariants. Incidentally, the discussion below will also determine the ``fate'' of these invariants when one ``forgets'' about the presence of some symmetry, and regards an element of an AZC class with all three symmetries as a specific element of a different class with less symmetries. One-dimensional invariants (or two-dimensional invariants which are defined from one-dimensional restrictions) will also be compared with the well-known geometric object given by the \emph{Zak phase}~\cite{berry1984quantal, zak1989berry}, which for our purposes we define as follows.   

\begin{definition}[Zak phase]\label{def:I_invariant}
Let $P^{\pm} \colon \mathbb{T}^{1} \to \Proj_n(\mathcal{H})$ be a (particle-hole and/or chiral symmetric) pair of projection-valued maps acting on the minimal ambient space $\mathcal{H} \simeq \C^{2n}$. Let $Z \colon \mathbb{T}^{1} \to U(2n)$ be a (symmetric) unitary-valued map such that $Z(k)P^\pm(0)=P^\pm(k)Z(k)$ for all $k \in \mathbb{T}^{1}$. We define the \emph{Zak phase} as
\begin{equation} \label{eq:I_invariant} 
\I(P^{\pm}):= (-1)^{[\det(Z(\cdot))]} \in \{\pm 1\}\simeq\Z_2.
\end{equation}
\end{definition}
In~\cite{monaco2023z2, manzonimonacopeluso2026zak} it is proved that such $Z$ can always be exhibited (compare~\eqref{eq:modified_parallel}), and that the Zak phase can be computed using a (smooth) periodic basis $\{v_j(k)\}_{j \in \{1,\cdots, n\}}$ for $P^-(k)$ (compare Section~\ref{sec:Bases} and Remark~\ref{rmk:regularity_required}) through the customary expression involving the Berry connection:
\begin{equation}\label{eq:I_as_Berry_phase}
    \I(P^{\pm}) = e^{i \pi \mathcal{A}} \quad \text{with} \quad \mathcal{A} := \frac{1}{i\pi} \int_{\mathbb{T}^1} \sum_{j=1}^n \inn{v_j(k)}{\partial_k v_j(k)} \, dk.
\end{equation}
In particular, this representation allows to show that $\I(P^{\pm})$ is invariant under changes of gauge which preserve the symmetries, which in turn implies that the Zak phase does not depend on the choice of $Z$; moreover, it provides a homotopy invariant for the projection-valued maps.
In~\cite{manzonimonacopeluso2026zak} the role of the Zak phase as a topological marker was investigated in all AZC classes via the above formula~\eqref{eq:I_as_Berry_phase}, showing in particular how it vanishes in presence of any antiunitary symmetry operator squaring to minus the identity.

After having established the relations among the various invariants, we will pass to the discussion on how these quantities behave under unitary conjugation of the corresponding pairs of projection-valued maps. This will clarify the role of the homotopy invariants which appeared in the previous Section~\ref{sec:Homotopy_minimal} as \emph{relative invariants}, as anticipated in the Introduction. Further properties of these quantities, relating in particular to how they behave under direct sums and to how they can be computed from periodic bases  of the underlying projections, are presented in Appendix~\ref{app:Further_properties}.

\subsection{Relations among invariants}

Throughout this Section, $P^{\pm} \colon \mathbb{T}^{d} \to \Proj_n(\mathbb{C}^{2n})$ will be a particle-hole and/or chiral symmetric pair of projection-valued maps. We'll proceed in our analysis going through all seven relevant AZC symmetry classes one by one.

\subsubsection{Class AIII}

When $d=1$, chiral symmetric pairs of projection-valued maps on a minimal ambient Hilbert space admit both a chiral invariant $\W(P^{\pm}) \in \mathbb{Z}$ and the Zak phase $\I(P^{\pm}) \in \mathbb{Z}_2$: they are related by the following relation.

\begin{proposition}\label{lem:relation_W-I}
For a pair of projection-valued maps $P^\pm \colon \T^1 \to \Proj_n(\C^{2n})$ in  class AIII, it holds that
\begin{equation}\label{eq:relation_W-P}
        (-1)^{\W(P^{\pm})}=\I(P^{\pm}) \in \mathbb{Z}_2.
\end{equation}
\end{proposition}
\begin{proof}
We appeal to Principle~\ref{pri:chiral_classes_min_dim}. It suffices to notice that the unitary-valued map $Z$ from Definition~\ref{def:I_invariant} that intertwines 
\[
P^\pm(k) = \frac{1}{2}
\begin{pmatrix}
    \Id_n & \pm W(k) \\ 
    \pm W(k)^* & \Id_n
\end{pmatrix} 
\quad \text{and} \quad  
P^\pm(0) = \frac{1}{2} 
\begin{pmatrix}
    \Id_n & \pm W(0) \\ 
    \pm W(0)^* & \Id_n
\end{pmatrix}
\]
can be exhibited as 
\begin{equation} \label{eq:explicit_unitary_equivalence_in_chiral_symmetry}
Z(k) = 
\begin{pmatrix} 
W(k)W(0)^* & 0 \\ 
0 & \Id_n
\end{pmatrix}.
\end{equation}
In particular, $[\det(Z(\cdot))]=[\det(W(\cdot))] \in \mathbb{Z}$.   
\end{proof}

In dimension $d=2$, the above result still allows to compare the chiral invariant with Zak phases, as the former is defined upon restriction of the projection-valued maps to the sub-tori $\mathbb{T}^{1}_{1}, \mathbb{T}^{1}_{2} \subset \mathbb{T}^{2}$. We have also already discussed in Remark~\ref{rmk:Hom_vs_MvN_in_chiral} how the Chern numbers of the two elements of a pair of projection-valued maps in class AIII necessarily vanish, so that the invariant characterizing Murray--von Neumann equivalence plays no role in the homotopy classification.

\subsubsection{Class C}

Theorem~\ref{thm:Homotopy_C_minimal} states that, for pairs of projection-valued maps $P^{\pm} \colon \mathbb{T}^{d} \to \Proj_n(\mathbb{C}^{2n})$ in class C, the homotopy classification is equivalent to the classification up to Murray--von Neumann equivalences in $d \le 2$; in particular, the only non-trivial topological invariant present in this class is the Chern number $\Ch(P^-) \in \mathbb{Z}$ when $d=2$.

One can still compute the Zak phase $\I(P^{\pm})$ if $d=1$, and more generally on every 1-dimensional sub-torus of $\T^d$ if $d > 1$. However, due to the presence of an odd antiunitary symmetry,~\cite[Theorem~5.1]{manzonimonacopeluso2026zak} gives that

\begin{proposition}\label{lem:I_trivial_class_c}
For any pair of projection-valued maps $P^{\pm} \colon \mathbb{T}^{1} \to \Proj_n(\mathbb{C}^{2n})$ in class C, the Zak phase is trivial: $\I(P^{\pm}) = 1 \in \mathbb{Z}_2$.
\end{proposition}

Alternatively, one can argue that the odd particle-hole constraint $CZ(k)=Z(-k)C$, where $Z \colon \mathbb{T}^{1} \to U(2n)$ is as in Definition~\ref{def:I_invariant}, reads as $J\overline{Z(k)}=Z(-k)J$ in an appropriate basis (compare~\eqref{eq:TRS_normal_form_odd}), and thus implies $\overline{\det(Z(k))}=\det(Z(-k))$ for all $k \in \mathbb{T}^{1}$. It is easy to see that a map with this property has even winding number: at $k_{\star} \in \{0,\pm\pi\}$ the matrix $Z(k_{\star})$ is symplectic, and thus has unit determinant; the symmetry constraint yields that the two ``halves'' of the torus give the same integer contribution to the winding number.

\subsubsection{Class D}

When an even particle-hole symmetry is present, then one can compute the even-particle-hole invariant $\PI(P^{\pm})$ as the collection of the Pfaffians $\Pf(L(k_{\star})) \in \mathbb{Z}_2$, where $L(k) = i \left(P^{+}(k) - P^{-}(k) \right)$ is as in~\eqref{eq:L(k)} and where $k_{\star} \in \mathbb{T}^{d}$ ranges over the fixed points of the involution $k \mapsto -k$ (see Definition~\ref{def:P_invariant}). Let us denote by $F^d$ the set of such fixed points: explicitly
\begin{equation} \label{eq:Fixed}
F^0 = \{0\} \subseteq \mathbb{T}^{0}, \quad F^1 = \{0,\pi\} \subset \mathbb{T}^{1}, \quad F^2 = \{ (0,0), (0,\pi), (\pi,0), (\pi, \pi) \} \subset \mathbb{T}^{2}.
\end{equation}

In $d=1$, one can also compute the Zak phase $\I(P^{\pm}) \in \mathbb{Z}_2$ according to Definition~\ref{def:I_invariant}. In this setting, the two invariants are related by the following result.

\begin{proposition}\label{lem:relation_I-P}
For a pair of projection-valued maps $P^\pm \colon \T^1 \to \Proj_n(\C^{2n})$ in  class D, it holds that
\begin{equation}\label{eq:relation_I-P}
    \I(P^{\pm}) = \prod_{k_{\star} \in F^1} \Pf(L(k_{\star})) .
\end{equation}
\end{proposition}
\begin{proof}
Let $Z:\T^1 \to U(2n)$ be a unitary-valued map such that $CZ(k)=Z(-k)C$ and $Z(k)P^\pm(0)=P^\pm(k)Z(k)$, as in Definition~\ref{def:I_invariant}. By the definition of $L(k)$ we have that
\[
L(0) = Z(0) L(0) Z(0)^{-1} \quad \text{and} \quad L(\pi) = Z(\pi) L(0) Z(\pi)^{-1}.
\]
Assume a basis for $\mathcal{H} \simeq \C^{2n}$ is chosen so that $C=\mathcal{K}$ (compare~\eqref{eq:TRS_normal_form_even}). Then at $k_{\star} \in \{0,\pi\} = F^1$ we have $\overline{Z(k_{\star})}=Z(k_{\star})$, and consequently $Z(k_{\star})^{-1} = Z(k_{\star})^{\mathrm{t}}$. So, the previous identities imply that
\[
\Pf(L(0)) = \det(Z(0)) \, \Pf(L(0)) \quad \text{and} \quad \Pf(L(\pi)) = \det(Z(\pi)) \, \Pf(L(0))
\]
in view of how the Pfaffian changes under congruences. Since $\Pf(L(k_{\star})) \in \{\pm1\}$, we conclude that $\det(Z(0))=1$ and $\det(Z(\pi))= \Pf(L(0)) \, \Pf(L(\pi))$. 

Observe now that the particle-hole constraint for $Z$ yields $\det(Z(k))=\overline{\det (Z(-k))}$, $k \in \mathbb{T}^{1}$. Choose a continuous map $\mu:[-\pi,\pi] \to \R$ such that $\det(Z(k))=e^{2\pi i \mu(k)}$ with $\mu(k)=-\mu(-k)$; we can also fix $\mu(0)=0$ as $\det(Z(0))=1$ as we already proved. With this choice, we have 
\[ [\det(Z(\cdot))]=\mu(\pi)-\mu(-\pi)=2\mu (\pi) \]
and therefore
\[ \I(P^{\pm}) = e^{2 \pi i \mu(\pi)} = \det(Z(\pi)) = \prod_{k_{\star} \in F^1} \Pf(L(k_{\star})) \]
as wanted.
\end{proof}

Something interesting happens in $d=2$, where the homotopy and Murray--von Neumann classifications become intertwined, as stated in Theorem~\ref{thm:Homotopy_D_minimal}. The next result clarifies the relation between the even-particle-hole homotopy invariant $\PI(P^{\pm}) \in (\Z_2)^4$ and the Murray--von Neumann invariant given by the Chern number $\Ch(P^{-}) \in \Z$. 

\begin{proposition}\label{lem:Class_D_relation_Chern-P}
For a pair of projection-valued maps $P^\pm \colon \T^2 \to \Proj_n(\C^{2n})$ in class D, it holds that
\[
    (-1)^{\Ch(P^-)} = \prod_{k_{\star} \in F^2} \Pf(L(k_{\star})) .
\]
\end{proposition}
\begin{proof}
Recall from Section~\ref{sec:Bases} that, using a pseudo-periodic and symmetric basis for the pair $P^\pm$, we can define a unitary-valued map $U \colon [-\pi,\pi] \times \mathbb{T}^{1} \to U(2n)$, $(t,k_2) \mapsto U(t,k_2)$, such that $U(t,k_2)P^\pm(0,0)=P^\pm(t,k_2)U(t,k_2)$ and $CU(t,k_2)=U(-t,-k_2)C$. We write here $t$ and not $k_1$ to stress the fact that such $U$ is periodic only in the second variable, but in general not in the first. With that and a basis $\{v_j\}_{j\in \{1,\cdots ,n\}}$ for the range of the projection $P^-(0,0)$, we can introduce the $n \times n$ unitary matching matrix 
\[ [\alpha(k_2)]_{i,j}:=\inn{v_i}{U(-\pi,k_2)^{-1}U(\pi,k_2)v_j}, \quad i,j \in \{1,\ldots,n\},\]
as in \eqref{eq:standard_alpha}; in turn, this allows us to compute the Chern number $\Ch (P^-)=[\det(\alpha(\cdot))]$.
    
This definition of the Chern number concerns matrices of size $n\times n$, while we aim at establishing a relation with Pfaffians of $(2n\times 2n)$-matrices. To relate them more easily, we reinterpret the winding number defining the Chern number as the one produced from the determinant of a bigger matrix, namely we write $\det(\alpha(k_2))=\det(X(k_2))$ where
\begin{align*}
X(k_2) &:= P^+(\pi,0) + P^-(\pi,0) U(-\pi,0) U(-\pi,k_2)^{-1} U(\pi,k_2) U(\pi,0)^{-1} P^-(\pi,0) \\
    & = \begin{cases}
        \Id_n  & \text{on } \Imm(P^+(\pi,0)), \\
        U(-\pi,0) U(-\pi,k_2)^{-1} U(\pi,k_2) U(\pi,0)^{-1} & \text{on } \Imm(P^-(\pi,0)),
    \end{cases}
\end{align*}
for $k_2 \in \mathbb{T}^{1}$. Notice that by construction $X(k_2)=U(-\pi,0)U(-\pi,k_2)^{-1} Z_\pi(k_2)$, where
\[
Z_\pi(k_2):=\begin{cases}
    U(-\pi,k_2)U(-\pi,0)^{-1} & \text{on } \Imm(P^+(\pi,0)), \\
    U(\pi,k_2)U(\pi,0)^{-1} & \text{on } \Imm(P^-(\pi,0)).
\end{cases}
\]

The continuity of $U$ in $t \in [-\pi,\pi]$ implies that $[\det(U(-\pi,\cdot))]=[\det(U(0,\cdot))]$, as the winding number is a homotopy invariant, so we have 
\[
\Ch(P^-) = [\det(X(\cdot))] = [\det(Z_\pi(\cdot))]-[\det(U(0,\cdot))]
\]
and consequently
\begin{equation} \label{eallipichern}
(-1)^{\Ch(P^-)} = (-1)^{[\det(Z_\pi(\cdot))]} (-1)^{-[\det(U(0,\cdot))]}.
\end{equation}

Notice now that $Z_0(k_2) := U(0,k_2)$ defines a periodic unitary intertwiner between $P^{-}(0,0)$ and $P^{-}(0,k_2)$ in view of the defining properties of $U$; similarly
\begin{align*} 
Z_\pi(k_2) P^{-}(\pi,0) & = U(\pi,k_2) U(\pi,0)^{-1} P^{-}(\pi,0) = U(\pi,k_2) P^{-}(0,0) U(\pi,0)^{-1} \\
& = P^{-1}(\pi,k_2) U(\pi,k_2) U(\pi,0)^{-1} = P^{-1}(\pi,k_2) Z_\pi(k_2),
\end{align*}
so that $Z_\pi(k_2)$ gives a periodic unitary intertwiner between $P^{-}(\pi,0)$ and $P^{-}(\pi,k_2)$. A similar argument conducted on the `plus' component of the pair, spanning the orthogonal complement of $P^{-}$, allows to conclude that $Z_0, Z_\pi \colon \mathbb{T}^{1} \to U(2n)$ qualify as unitary-valued maps of the type described in Definition~\ref{def:I_invariant}, and can be used to compute the Zak phases of the restrictions $P^{\pm}(0,\cdot), P^{\pm}(\pi,\cdot) \colon \mathbb{T}^{1} \to \Proj_n(\mathbb{C}^{2n})$, respectively. Proposition~\ref{lem:relation_I-P} applies to both such restrictions, and yields
\begin{align*}
(-1)^{[\det(Z_\pi(\cdot))]} & = \I(P^{\pm}(\pi,\cdot)) = \Pf(L(\pi,0)) \, \Pf(L(\pi,\pi)) \\
(-1)^{-[\det(Z_0(\cdot))]} & = \I(P^{\pm}(0,\cdot))^{-1} = \Pf(L(0,0)) \, \Pf(L(0,\pi))
\end{align*}
(notice that quantities in $\Z_2 =  \{\pm1\}$ are equal to their inverses). Together with~\eqref{eallipichern}, the above identities give the desired statement.  
\end{proof}

The above result allows to slightly refine the statement of Theorem~\ref{thm:Homotopy_D_minimal}. The condition for two pairs of projection-valued maps in class D over the $2$-dimensional torus and acting on a minimal ambient Hilbert space to be homotopic is that of being Murray--von Neumann equivalent, and for three of the four Pfaffians appearing in the definition of the even-particle-hole invariant to agree; the value of the fourth one is then automatically derived from the knowledge of the Chern number of the `minus' component of the pair, which already determines the Murray--von Neumann equivalence of the two pairs.

\subsubsection{Class BDI}

The class features commuting chiral and even particle-hole symmetry operators: the considerations from classes AIII and D therefore apply here as well. In $d=1$ one can combine Propositions~\ref{lem:relation_W-I} and~\ref{lem:relation_I-P} and obtain the following relations among the chiral invariant, the even-particle-hole invariant, and the Zak phase of a pair of projection-valued maps $P^{\pm} \colon \mathbb{T}^{1} \to \Proj_n(\C^{2n})$:
\[
(-1)^{\W(P^{\pm})} = \I(P^{\pm}) = \prod_{k_{\star} \in F^1} \Pf(L(k_{\star})).
\]
A similar conclusion holds in $d=2$, since the chiral and even-particle-hole invariants are defined via the restrictions to the 1-dimensional sub-tori along cardinal directions in $\mathbb{T}^{2}$. These considerations allow to refine the statement of Theorem~\ref{thm:Homotopy_BDI_minimal}, and a minimal set of complete homotopy invariants in this class is given, for example, by the chiral invariant together with one of the Pfaffians defining the even-particle-hole invariant per direction in the torus.

\subsubsection{Class CII}

This class displays two commuting antiunitary symmetries of odd type: as mentioned above,~\cite[Theorem~5.1]{manzonimonacopeluso2026zak} yields that

\begin{proposition}\label{lem:I_trivial_class_cII}
For any pair of projection-valued maps $P^{\pm} \colon \mathbb{T}^{1} \to \Proj_n(\mathbb{C}^{2n})$ in class CII, the Zak phase is trivial: $\I(P^{\pm}) = 1 \in \mathbb{Z}_2$.
\end{proposition}

For an alternative argument, one can invoke the fact, already proven in Theorem~\ref{thm:Homotopy_CII_minimal}, that the chiral invariant takes values in even integers for such projection-valued maps, and conclude the triviality of the Zak phase with Proposition~\ref{lem:relation_W-I}.

We compare now also the classification of pairs of projection-valued maps in this class and in $d=2$ up to homotopy and up to Murray--von Neumann equivalence, as the latter invokes the Fu--Kane--Mele invariant (Section~\ref{sec:Bases_class_BDIetal}) which is absent in the former (Theorem~\ref{thm:Homotopy_CII_minimal}), rather involving the chiral invariant. This absence is explained by the following statement (compare Remark~\ref{rmk:Hom_vs_MvN_in_chiral} for a similar conclusion in class AIII).

\begin{proposition}\label{lem:Class_CII_vanishing_delta}
For any pair of projection-valued maps $P^{\pm} \colon \mathbb{T}^{2} \to \Proj_n(\mathbb{C}^{2n})$ in class CII, the Fu--Kane--Mele invariant is trivial: $\FKM(P^{-}) = 1 \in \mathbb{Z}_2$. In particular, any two such pairs of projection-valued maps acting on a minimal ambient Hilbert space are Murray--von Neumann equivalent.
\end{proposition}

\begin{proof}
The unitary operator $Z$ defined in~\eqref{eq:explicit_unitary_equivalence_in_chiral_symmetry} provides a symmetric unitary equivalence between $P^\pm$ and the constant pair $Q^{\pm}(\cdot) \equiv P^\pm(0)$. Since the invariant $\FKM(\cdot)$ is preserved under symmetric unitary equivalences (Remark~\ref{rmk:FKM}), we find that it must be trivial.
\end{proof}

\subsubsection{Class CI}

There is not much to be discussed: Section~\ref{sec:Bases_class_BDIetal} and Theorem~\ref{thm:Homotopy_CI_minimal} allow to conclude that any two pairs of projection-valued maps in this class are both Murray--von Neumann equivalent and homotopic to each other.

\subsubsection{Class DIII}

Finally, class DIII is probably the richest, as the homotopy classification prompts to consider {\it a priori} chiral, even-particle-hole, and DIII invariants; in this Section we further introduced the Zak phase on 1-dimensional tori; and finally, the Fu--Kane--Mele invariant enters the classification up to Murray--von Neumann equivalence (Section~\ref{sec:Bases_class_BDIetal}). The first statement of this Section shows that the chiral and even-particle-hole invariants, and consequently Zak phases, play no role in this class: compare Remark~\ref{rmk:DIII_no_invariants}.

\begin{proposition}
For any pair of projection-valued maps $P^{\pm} \colon \mathbb{T}^{d} \to \Proj_n(\mathbb{C}^{2n})$ in class DIII, $d \le 2$, the Pfaffians entering the Definition~\ref{def:P_invariant} of the even-particle-hole invariant $\PI(P^{\pm})$ are all equal and independent on $P^{\pm}$:%
\footnote{Recall that $n$ is necessarily an even positive integer.}
\[ \Pf(L(k_{\star})) \equiv (-1)^{n/2} \quad \forall \, k_{\star} \in F^d. \]
Moreover, both the chiral invariant and the Zak phase, whenever defined, are trivial: $\W(P^{\pm}) = 0 \in \mathbb{Z}$ if $d=1$ (respectively $\W(P^{\pm}) = (0,0) \in \mathbb{Z}^2$ if $d=2$), and $\I(P^{\pm}) = 1 \in \mathbb{Z}_2$ along any $1$-dimensional (sub-)torus of $\mathbb{T}^{d}$.
\end{proposition}

\begin{proof}
The considerations before the statement of Theorem~\ref{thm:Homotopy_DIII_minimal} exhibit $L(k)$ in the block-off-diagonal form~\eqref{eq:L(k)_block}, for a unitary-valued map $W \colon \mathbb{T}^{d} \to U(n)$ satisfying~\eqref{eq:W^t=-W}. However, for the Pfaffian $L(k_{\star})$ to be well-defined, Definition~\ref{def:P_invariant} requires this operator to be represented as an orthogonal and skew-symmetric matrix (compare also Remark~\ref{rmk:P_invariant}): this happens if the particle-hole symmetry operator $C$ takes the form of a conjugation operator, $C=\mathcal{K}$, while the choice adopted in class DIII has rather the operator $C$ in the form~\eqref{C_DIII}. We therefore change the basis of $\mathcal{H} \simeq \C^{2n}$ through the matrix
\[ B := \frac{1}{\sqrt{2}} \begin{pmatrix} 
\Id_n & \Id_n \\ 
i \Id_n & -i\Id_n
\end{pmatrix}. \]
Observe that
\[ B^{-1} = \frac{1}{\sqrt{2}} \begin{pmatrix} 
\Id_n & -i\Id_n \\ 
\Id_n & i\Id_n
\end{pmatrix} = \frac{1}{\sqrt{2}} \begin{pmatrix}
0 & \Id_n \\ \Id_n & 0
\end{pmatrix} \begin{pmatrix} 
\Id_n & i\Id_n \\ 
\Id_n & -i\Id_n
\end{pmatrix} = \begin{pmatrix}
0 & \Id_n \\ \Id_n & 0
\end{pmatrix} B^{\mathrm{t}}. \]
With this, and with the well-known properties of the Pfaffian, we can compute 
\begin{align*}
\Pf(L(k_\star)) & = \Pf\left[ B \begin{pmatrix}
            0 & i W(k_\star) \\ i W(k_\star)^* & 0 
        \end{pmatrix} B^{-1} \right] = \Pf\left[ B \begin{pmatrix}
            0 & i W(k_\star) \\ i W(k_\star)^* & 0 
        \end{pmatrix} \begin{pmatrix}
0 & \Id_n \\ \Id_n & 0
\end{pmatrix} B^{\mathrm{t}} \right] \\
        &= \det(B) \Pf\left[ \begin{pmatrix}
            iW(k_\star) &0 \\0& iW(k_\star)^*
        \end{pmatrix}\right] =\det(B) \Pf(i W(k_\star)) \Pf(iW(k_\star)^*) \\
        &= (-1)^{n/2} \Pf(iW(k_\star))\overline{\Pf(iW(k_\star))} = (-1)^{n/2} | \Pf(iW(k_\star)) |^2 \\
        &= (-1)^{n/2} | \det(iW(k_\star)) | = (-1)^{n/2}
\end{align*}
independently of $k_{\star} \in F^{d}$. Indeed, the condition~\eqref{eq:W^t=-W} on the unitary matrix $W(k_{\star})$, i.e.\ $W(k_\star)^{\mathrm{t}}=-W(k_\star)$, implies that $i W(k_{\star})$ is skew-symmetric (and thus has a Pfaffian) and that $i W(k_\star)^*=\overline{i W(k_\star)}$. Since $i W(k_{\star})$ is a unitary matrix, the absolute value of its determinant is $1$. This allows to conclude the first equality stated in the Theorem.

We have already argued below~\eqref{eq:W^t=-W} how the chiral invariant vanishes. The fact that the Zak phases are trivial follows then from Proposition~\ref{lem:relation_W-I}: compare also~\cite[Theorem~5.1]{manzonimonacopeluso2026zak}. 
\end{proof}

Next, we move to $d=2$. As was recalled above, Section~\ref{sec:Bases_class_BDIetal} identifies the Fu--Kane--Mele invariant $\FKM(P^{-}) \in \mathbb{Z}_2$ as the quantity classifying Murray--von Neumann equivalence classes of pairs of projection-valued maps $P^{\pm} \colon \mathbb{T}^{2} \to \Proj_n(\mathcal{H})$, while Theorem~\ref{thm:Homotopy_DIII_minimal} shows that their homotopy classes are characterized by the DIII invariant $\G(P^{\pm}) \in (\mathbb{Z}_{2})^{3}$ from Definition~\ref{def:DIII_invariant} if $\mathcal{H} \simeq \C^{2n}$. The relation among these two invariants is explained in the following result (compare~\cite[Equation~(3.69)]{chiu2016classification}).

\begin{proposition}\label{lem:relation_delta_gamma}
For a pair of projection-valued maps $P^\pm \colon \T^2 \to \Proj_n(\C^{2n})$ in class DIII, it holds that
\[
    \FKM(P^{-}) = \TK\left(W(\cdot) \big|_{\mathbb{T}^{1}_{1,\,\pi}}\right) \cdot \TK\left(W(\cdot) \big|_{\mathbb{T}^{1}_{2,\,\pi}}\right) \cdot \TK\left(W(\cdot) \big|_{\mathbb{T}^{1}_{\textup{diag}}}\right)
\]
where $\mathbb{T}^{1}_{1,\,\pi}, \mathbb{T}^{1}_{2,\,\pi}, \mathbb{T}^{1}_{\textup{diag}} \subset \mathbb{T}^{2}$ are as in Definition~\ref{def:DIII_invariant}.
\end{proposition}

\begin{proof}
We resort to the original definition from~\cite{FuKane2006} of the Fu--Kane--Mele invariant for a projection-valued map $P^{-} \colon \mathbb{T}^{2} \to \Proj_n(\mathcal{H})$ satisfying a time-reversal symmetry constraint of odd type (compare~\cite{Cornean_2017, Peluso2026} for the equivalence with Definition~\ref{def:FMK}). Since $\Ch(P^{-})=0$~\cite{panati2007triviality,monaco2015symmetry}, Proposition~\ref{prop:Class_A_frames} grants the existence of a periodic (but possibly not time-reversal symmetric) basis $\{\tilde{v}_1(k), \ldots, \tilde{v}_n(k)\}$ for $P^{-}$. Define the \emph{sewing matrix} $w \colon \mathbb{T}^{2} \to U(n)$ as
\[ w(k)_{i,j} := \inn{\tilde{v}_i(-k)}{T \tilde{v}_j(k)}, \quad i,j \in \{1,\ldots, n\}, \; k \in \mathbb{T}^{2}, \]
where $T$ is the time-reversal symmetry operator. An immediate check shows that $w(k)^{\mathrm{t}} = -w(-k)$ for $k \in \mathbb{T}^{2}$. The Fu--Kane--Mele invariant of $P^{-}$ can then be computed from the sewing matrix as~\cite[Equation~(3.26)]{FuKane2006}
\begin{equation} \label{eq:TRIM}
\FKM(P^{-}) = \TK\left(w(\cdot) \big|_{\mathbb{T}^{1}_{1,\,0}}\right) \cdot \TK\left(w(\cdot) \big|_{\mathbb{T}^{1}_{1,\,\pi}}\right) = \prod_{k \in F^2} \frac{\sqrt{\det(w(k_{\star}))}}{\Pf(w(k_{\star}))},
\end{equation}
where $\mathbb{T}^{1}_{1,\,0}, \mathbb{T}^{1}_{1,\,\pi} \subset \mathbb{T}^{2}$ are as in~\eqref{eqn:subtori}, $F^2$ is the set of fixed points (or \emph{time-reversal invariant momenta}) in~\eqref{eq:Fixed}, and the square-root on the right-hand side is chosen in a continuous way%
\footnote{That this can be done can be argued as in Remark~\ref{def:gamma_invariant} and the corresponding footnote. Indeed, the condition $w(k)^{\mathrm{t}} = -w(-k)$, $k \in \mathbb{T}^{2}$, makes the unitary-valued map $w$ homotopically trivial, as the winding numbers of its determinant along the cardinal directions in the torus vanish: compare Theorem~\ref{sec:Homotopies_unitary}. This in turn guarantees the existence of a continuous and periodic function $\mu \colon \mathbb{T}^{2} \to \mathbb{R}$ such that $\det(w(k)) = e^{i \mu(k)}$, $k \in \mathbb{T}^{2}$. A choice for the square-root that is continuous across the whole torus can be performed setting $\sqrt{\det(w(k))} := e^{i \mu(k)/2}$, $k \in \mathbb{T}^{2}$.}.

To compare this definition of the Fu--Kane--Mele invariant with the DIII invariant, we first of all appeal to Principle~\ref{pri:chiral_classes_min_dim}, and write
\[ P^{-}(k) = \frac{1}{2} \begin{pmatrix} \Id_n & - W(k) \\ - W(k)^* & \Id_n \end{pmatrix} \]
where $W \colon \mathbb{T}^{2} \to U(n)$ is as in~\eqref{chiral_symm_P+-P-}. Here, $W$ is understood as a unitary-valued map between $\mathcal{H}^{\downarrow} = \ker(S + \Id_{\mathcal{H}})$ and $\mathcal{H}^{\uparrow} = \ker(S - \Id_{\mathcal{H}})$, where $S$ is the chiral symmetry operator. Choose now an orthonormal basis $\{e_1, \ldots, e_n\}$ for $\mathcal{H}^{\uparrow}$, which we identify as the vectors
\[ e_j^{\uparrow} := \begin{pmatrix} e_j \\ 0 \end{pmatrix} \in \mathcal{H}, \quad j \in \{1,\ldots,n\}. \]
Observe that, if we set
\[ e_j^{\downarrow} := T e_j^{\uparrow} =: \begin{pmatrix} 0 \\ f_j \end{pmatrix} = \begin{pmatrix} 0 \\ - \mathcal{K} e_j \end{pmatrix}, \]
then the collection $\{ f_1, \ldots, f_n \}$ gives an orthonormal basis for $\mathcal{H}^{\downarrow}$: the above identities are a consequence of the fact that $T$ and $S$ anticommute, together with the form $T = \mathcal{K} J$ chosen for the odd time-reversal symmetry operator (compare~\eqref{eq:TRS_normal_form_odd}). It is now easily checked that
\[
\tilde{v}_j(k) := \frac{1}{\sqrt{2}} \begin{pmatrix} -e_j \\ W(k)^* e_j \end{pmatrix}, \quad j \in \{1,\ldots,n\}, \; k \in \mathbb{T}^{2},
\]
exhibits a continuous and periodic basis for $P^{-}$: indeed, for $j \in \{1,\ldots,n\}$ and $k \in \mathbb{T}^{2}$,
\begin{equation} \label{eq:basis_from_W} 
\begin{aligned} 
P^{-}(k) \tilde{v}_j(k) & = \frac{1}{2 \sqrt{2}} \begin{pmatrix} \Id_n & - W(k) \\ - W(k)^* & \Id_n \end{pmatrix} \begin{pmatrix} -e_j \\ W(k)^* e_j \end{pmatrix} \\
& = \frac{1}{\sqrt{2}} \begin{pmatrix} \dfrac{-e_j - W(k) W(k)^* e_j}{2} \\ \dfrac{W(k)^* e_j + W(k)^* e_j}{2} \end{pmatrix} = \tilde{v}_j(k).
\end{aligned}
\end{equation}

To compute the sewing matrix from $\{\tilde{v}_j(k)\}_{j \in \{1,\ldots,n\}}$, let us first observe that
\begin{align*}
T \tilde{v}_j(k) &= \mathcal{K} J \tilde{v}_j(k) = \frac{1}{\sqrt{2}} \mathcal{K} \begin{pmatrix} 0 & \Id_n \\ -\Id_n & 0 \end{pmatrix} \begin{pmatrix} -e_j \\ W(k)^* e_j \end{pmatrix} = \frac{1}{\sqrt{2}} \mathcal{K} \begin{pmatrix} W(k)^* e_j \\ e_j \end{pmatrix} \\
&= \frac{1}{\sqrt{2}} \begin{pmatrix} W(k)^{\mathrm{t}} \mathcal{K} e_j \\ \mathcal{K} e_j \end{pmatrix} = \frac{1}{\sqrt{2}} \begin{pmatrix} -W(-k) \mathcal{K} e_j \\ \mathcal{K} e_j \end{pmatrix} = \frac{1}{\sqrt{2}} \begin{pmatrix} W(-k) f_j \\ -f_j \end{pmatrix}
\end{align*}
where we used the symmetry condition~\eqref{eq:W^t=-W} for $W$ in the second-to-last equality. We can then compute for $i,j \in \{1,\ldots,n\}$ and $k \in \mathbb{T}^{2}$
\begin{align*} 
w(k)_{i,j} &= \inn{\frac{1}{\sqrt{2}} \begin{pmatrix} -e_i \\ W(-k)^* e_i \end{pmatrix}}{\frac{1}{\sqrt{2}} \begin{pmatrix} W(-k) f_j \\ - f_j \end{pmatrix}} = \frac{1}{2} \left( -\inn{e_i}{W(-k) f_j} - \inn{W(-k)^* e_i}{f_j} \right) \\
&= -\inn{e_i}{W(-k) f_j}
\end{align*}
or, more compactly, $w(k) = - W(-k)$, $k \in \mathbb{T}^{2}$. This relation implies 
\begin{gather*} 
\det(w(k)) = (-1)^{n} \det(W(-k)) = \det(W(-k)) \quad \forall\:k \in \mathbb{T}^{2}, \\
\Pf((w(k_{\star})) = (-1)^{n/2} \Pf(W(k_{\star})) \quad \forall \:k_{\star} \in F^2,
\end{gather*}
since $n$ is necessarily even: plugging the above identities in~\eqref{eq:TRIM} we obtain
\begin{equation} \label{eq:TRIM_W} 
\FKM(P^{-}) = \prod_{k \in F^2} \frac{\sqrt{\det(W(k_{\star}))}}{\Pf(W(k_{\star}))}. 
\end{equation}

To conclude the proof, it suffices to compute the product of the three components of the $2$-dimensional DIII invariant, as defined in Definition~\ref{def:DIII_invariant}. We have, by the definition of the Teo--Kane invariant in Remark~\ref{def:gamma_invariant},
\begin{align*}
\TK\left(W(\cdot) \big|_{\mathbb{T}^{1}_{1,\,\pi}}\right) &= \frac{\sqrt{\det(W(\pi,0))}}{\Pf(W(\pi,0))} \, \frac{\sqrt{\det(W(\pi,\pi))}}{\Pf(W(\pi,\pi))}, \\
\TK\left(W(\cdot) \big|_{\mathbb{T}^{1}_{2,\,\pi}}\right) &= \frac{\sqrt{\det(W(0,\pi))}}{\Pf(W(0,\pi))} \, \frac{\sqrt{\det(W(\pi,\pi))}}{\Pf(W(\pi,\pi))}, \\
\TK\left(W(\cdot) \big|_{\mathbb{T}^{1}_{\textup{diag}}}\right) &= \frac{\sqrt{\det(W(0,0))}}{\Pf(W(0,0))} \, \frac{\sqrt{\det(W(\pi,-\pi))}}{\Pf(W(\pi,-\pi))} = \frac{\sqrt{\det(W(0,0))}}{\Pf(W(0,0))} \, \frac{\sqrt{\det(W(\pi,\pi))}}{\Pf(W(\pi,\pi))},
\end{align*}
where in the last equality we used the periodicity of $W(\pi, \cdot)$. Taking the product of these three expressions, we see that the ratio $\sqrt{\det(W(\pi,\pi))}/\Pf(W(\pi,\pi))$ appears thrice: however,
\[ \left(\frac{\sqrt{\det(W(\pi,\pi))}}{\Pf(W(\pi,\pi))}\right)^2 = \frac{\det(W(\pi,\pi))}{\Pf(W(\pi,\pi))^2} = 1. \]
Therefore, we deduce that
\[ \TK\left(W(\cdot) \big|_{\mathbb{T}^{1}_{1,\,\pi}}\right) \cdot \TK\left(W(\cdot) \big|_{\mathbb{T}^{1}_{2,\,\pi}}\right) \cdot \TK\left(W(\cdot) \big|_{\mathbb{T}^{1}_{\textup{diag}}}\right) = \prod_{k \in F^2} \frac{\sqrt{\det(W(k_{\star}))}}{\Pf(W(k_{\star}))} \]
and the comparison with~\eqref{eq:TRIM_W} concludes the proof.
\end{proof}

The above Proposition, together with the conclusions of Section~\ref{sec:Bases_class_BDIetal}, allows to rephrase Theorem~\ref{thm:Homotopy_DIII_minimal} as follows: Two pairs of projection-valued maps in class DIII acting on a minimal ambient Hilbert space are homotopically equivalent if and only if they are Murray--von Neumann equivalent and any two of the three components of their DIII invariants agree. The value of the third component is indeed determined, via the previous statement, by the other two and the Fu--Kane--Mele invariant which characterizes the Murray--von Neumann equivalence classes.

\medskip

In conclusion, a minimal set of homotopy invariants for all AZC classes of (pairs of) projection-valued maps on minimal Hilbert spaces is presented in the following Table~\ref{tabular:AZC_classes_Homotopy_minimal2}.

\begin{table}[!ht]
\centering
\caption{A minimal set of homotopy invariants for (pairs of) rank-$n$ projection-valued maps acting on minimal Hilbert spaces; the rank is omitted. In the columns labeled $d=1$ and $d=2$, we list the set of possible labels, together with the independent topological invariants that characterize the homotopy classes. The last column collects all non-trivial relations and constraints among the possible topological invariants defined in each class.} 
\label{tabular:AZC_classes_Homotopy_minimal2}

\renewcommand{\arraystretch}{1.15}
\begin{tabularx}{\textwidth}{
    ||>{\centering\arraybackslash}m{2.5em}
    ||>{\centering\arraybackslash}m{10.5em}
    ||>{\centering\arraybackslash}m{10.5em}
    ||>{\centering\arraybackslash}X||
}
    \hline
    AZC class
    & $d=1$
    & $d=2$
    & Relations and constraints
    \\
    \hline\hline

    A
    &
    $\text{trivial}$
    &
    $\displaystyle \begin{matrix} \mathbb Z \\
    [\operatorname{Ch}(P)] \end{matrix}$
    &
    --
    \\ \hline

    AIII
    &
    $\displaystyle \begin{matrix} \mathbb Z\\
    [\W(P^\pm)] \end{matrix}$
    &
    $\displaystyle \begin{matrix} \mathbb Z^2 \\
    [\W(P^\pm)] \end{matrix}$
    &
    $\operatorname{Ch}(P^\pm)=0$
    \\ \hline

    AI
    &
    $\text{trivial}$
    &
    $\text{trivial}$
    &
    $\operatorname{Ch}(P)=0$
    \\ \hline

    BDI
    &
    $\displaystyle \begin{matrix} 2\Z\times\Z_2 \\
    [\W(P^\pm)\text{, one component}\\
    \text{of }\PI(P^\pm)] \end{matrix}$
    &
    $\displaystyle \begin{matrix} (2\Z)^2\times\Z_2 \\
    [\W(P^\pm)\text{, one component} \\
    \text{of }\PI(P^\pm)] \end{matrix}$
    &
    The other components of $\PI(P^\pm)$ depend on $\W(P^\pm)$;
    $\operatorname{Ch}(P^\pm)=0$
    \\ \hline

    D
    &
    $\displaystyle \begin{matrix} (\Z_2)^2 \\
    [\PI(P^\pm)] \end{matrix}$
    &
    $\displaystyle \begin{matrix} (\Z_2)^3\times\mathbb Z \\
    [\text{three components of} \\
    \PI(P^\pm), \operatorname{Ch}(P^-)]\end{matrix}$ 
    &
    In $d=2$, the remaining component of $\PI(P^\pm)$
    depends on the other invariants
    \\ \hline

    DIII
    &
    $ \displaystyle \begin{matrix}\Z_2
    \\ [\G(P^\pm)] \end{matrix}$
    &
    $\displaystyle \begin{matrix}(\Z_2)^3 \\
    [\G(P^\pm)] \end{matrix}$
    &
    In $d=2$, $\FKM(P^\pm)$ depends on $\G(P^\pm)$;
    all other invariants are $0$
    \\ \hline

    AII
    &
    $\text{trivial}$
    &
    $\displaystyle \begin{matrix}\Z_2 \\
    [\FKM(P)] \end{matrix}$
    &
    $\operatorname{Ch}(P)=0$
    \\ \hline

    CII
    &
    $\displaystyle \begin{matrix} 2\Z \\
    [\W(P^\pm)] \end{matrix}$
    &
    $ \displaystyle \begin{matrix}(2\Z)^2 \\
    [\W(P^\pm)] \end{matrix}$
    &
    $\FKM(P^\pm)=1$;
    $\operatorname{Ch}(P^\pm)=0$
    \\ \hline

    C
    &
    $\text{trivial}$
    &
    $\displaystyle \begin{matrix} \mathbb Z \\
    [\operatorname{Ch}(P^-)] \end{matrix}$
    &
    --
    \\ \hline

    CI
    &
    $\text{trivial}$
    &
    $\text{trivial}$
    &
    $\operatorname{Ch}(P^\pm)=0$;
    $\W(P^\pm)=0$
    \\

    \hline
\end{tabularx}
\end{table}

\subsection{Change of the homotopy invariants under unitary conjugation}\label{section:Dimerization_ambiguity}

As we discussed in Sections~\ref{sec:Bases} and~\ref{sec:unitary_equivalences}, the invariants that characterize Murray--von Neumann and unitary equivalence classes are \emph{absolute invariants}, i.e.\ they do not change under homotopies and unitary equivalences of the corresponding projection-valued maps. In contrast, we will now show that the homotopy invariants introduced in Section~\ref{sec:Homotopy_minimal} for projections on a minimal Hilbert space \emph{do} change in general under unitary conjugation, albeit in a prescribed way which will be detailed below: this qualifies these homotopy invariants as \emph{relative}.

In what follows, we thus consider two pairs of projection-valued maps $P_0^\pm, P_1^\pm \colon \mathbb{T}^{d} \to \Proj_n(\C^{2n})$, $d \le 2$, together with a symmetric unitary-valued map intertwining them: $U (k)P_0^\pm(k)=P_1^\pm(k)U(k)$, $k \in \mathbb{T}^{d}$. We will express the homotopy invariants of $P_1^\pm$ in terms of the corresponding invariants for $P_0^\pm$ and appropriate homotopy invariants for $U$.

\subsubsection{Chiral invariant} 

Let us start with the chiral invariant $\W(P^{\pm})$ from Definition~\ref{def:W_invariant}. Since the invariant in $d=2$ is obtained from restriction to $1$-dimensional subtori in $\mathbb{T}^{2}$, we focus on projection-valued maps over the $1$-dimensional torus.

Let $P_0^\pm, P_1^\pm \colon \T^1 \to \Proj_n(\mathcal{H})$ be two pairs of projection-valued maps acting on a minimal Hilbert space $\mathcal{H} \simeq \C^{2n}$ which are symmetric under a chiral symmetry operator $S$. Assume that $U \colon \mathbb{T}^{1} \to U(2n)$ is a chiral symmetric unitary-valued map intertwining them. The chiral symmetry commutation relation $S U(k) = U(k) S$ implies that, in the decomposition $\C^{2n} = \mathcal{H}^{\uparrow} \oplus \mathcal{H}^{\downarrow} \equiv \ker(S - \Id_{2n}) \oplus \ker(S + \Id_{2n})$,
\begin{equation} \label{eq:UupUdown} 
U(k) = \begin{pmatrix}
U^{\uparrow}(k) & 0 \\ 0 & U^{\downarrow}(k)
\end{pmatrix}
\end{equation}
where $U^{\uparrow} \colon \mathbb{T}^{1} \to \mathcal{U}(\mathcal{H}^{\uparrow}) \simeq U(n)$ (respectively $U^{\downarrow} \colon \mathbb{T}^{1} \to \mathcal{U}(\mathcal{H}^{\downarrow}) \simeq U(n)$). Then we have

\begin{proposition}\label{prop:W_dimerization}
With the notation above,
\[ \W(P_1^{\pm})= \W(P_0^{\pm}) + [\det(U^\uparrow(\cdot))] - [\det(U^\downarrow(\cdot))]. \]
\end{proposition}
\begin{proof}
We recall that $\W(P_s^{\pm}) = [\det(W_s(\cdot))]$, where $s \in \{0,1\}$ and $W_s \colon \mathbb{T}^{1} \to U(n)$ is as in~\eqref{chiral_symm_P+-P-}. The intertwining relation $U(k)P_0^\pm(k)=P_1^\pm(k)U(k)$ readily implies that $W_1(k)=U^\uparrow(k)W_0(k)U^\downarrow(k)^*$ for all $k \in \mathbb{T}^{1}$, and the thesis follows.
\end{proof}

\subsubsection{Even-particle-hole invariant}

Let us pass now to the analysis of the even-particle-hole invariant $\PI(P^{\pm})$ from Definition~\ref{def:P_invariant}. This time, $P_0^\pm, P_1^\pm \colon \T^d \to \Proj_n(\mathcal{H})$, $d \le 2$, are pairs of projection-valued maps acting on a minimal Hilbert space $\mathcal{H} \simeq \C^{2n}$ which are symmetric under an even particle-hole symmetry operator $C = \mathcal{K}$ (compare~\eqref{eq:TRS_normal_form_even}), and $U \colon \mathbb{T}^{d} \to U(2n)$ is a particle-hole symmetric unitary-valued map intertwining them. The even-particle-hole invariant requires considering the Pfaffian of the matrix $L_s$ defined in~\eqref{eq:L(k)} from $P_s^{\pm}$, $s \in \{0,1\}$, when this is evaluated at the fixed points $k_{\star} = - k_{\star} \in F^{d} \subset \mathbb{T}^{d}$ (compare~\eqref{eq:Fixed}). We state therefore the following

\begin{proposition}
With the notation above,
\[ \Pf(L_1(k_{\star})) = \Pf(L_0(k_{\star})) \det(U(k_{\star})) \quad \forall \: k_{\star} \in F^d. \]
\end{proposition} 
\begin{proof}
We clearly have $U(k) L_0(k) = L_1(k) U(k)$ for all $k \in \mathbb{T}^{2}$. The even particle-hole symmetry of $U$, namely $\mathcal{K} U(k) = U(-k) \mathcal{K}$ for $k \in \mathbb{T}^{2}$, implies that, at the fixed points $k_{\star} \in F^d$, the matrix $U(k_{\star})$ is real-valued and orthogonal: in particular $U(k_{\star})^{-1} = U(k_{\star})^*=U(k_{\star})^{\mathrm{t}}$. The conclusion follows from the behaviour of the Pfaffian under congruences.
\end{proof}

\subsubsection{DIII invariant} 

Finally, we study the change of the DIII invariant from Definition~\ref{def:DIII_invariant} under unitary conjugation. Also the DIII invariant is defined in $d=2$ by restriction to $1$-dimensional subtori, so we focus our attention to the case $d=1$. The setting is then that of two pairs of projection-valued maps $P_0^\pm, P_1^\pm \colon \T^1 \to \Proj_n(\C^{2n})$ in class DIII, together with a unitary-valued map $U \colon \mathbb{T}^{1} \to U(2n)$ which intertwines them and is time-reversal, particle-hole and chiral symmetric. With the customary choices of normal forms for the odd time-reversal symmetry operator $T$, even particle-hole symmetry operator $C$ and chiral operator~$S$, we can deduce a specific block-decomposition of the matrix $U(k)$, $k \in \mathbb{T}^{1}$: commutation with the chiral symmetry operator dictates the form~\eqref{eq:UupUdown}, while the time-reversal symmetry constraint $T U(k) = U(-k) T$ with $T = \mathcal{K} J$ yields, for $k \in \mathbb{T}^{1}$, 
\begin{equation} \label{eq:U_DIII} 
\begin{pmatrix}
0 & \Id_n \\ -\Id_n & 0
\end{pmatrix} \begin{pmatrix}
\overline{U^{\uparrow}(k)} & 0 \\ 0 & \overline{U^{\downarrow}(k)} 
\end{pmatrix} =  \begin{pmatrix}
U^{\uparrow}(-k) & 0 \\ 0 & U^{\downarrow}(-k)
\end{pmatrix} \begin{pmatrix}
0 & \Id_n \\ -\Id_n & 0
\end{pmatrix} \iff \overline{U^{\downarrow}(k)} = U^{\uparrow}(-k). 
\end{equation}

After these considerations, we are able to state
\begin{proposition}\label{prop:gamma_dimerization}
With the notation above,
\[ \G(P_1^{\pm}) = \G(P_0^{\pm}) \cdot (-1)^{-[\det(U^{\uparrow}(\cdot))]} = \G(P_0^{\pm}) \cdot (-1)^{[\det(U^{\downarrow}(\cdot))]}. \]
\end{proposition}
\begin{proof}
Let $W_0, W_1 \colon \mathbb{T}^{1} \to U(n)$ be as in~\eqref{chiral_symm_P+-P-}. As in the proof of Proposition~\ref{prop:W_dimerization}, we have that $W_0$ and $W_1$ are related as $W_1(k) = U^\uparrow(k) W_0(k) U^\downarrow(k)^* = U^\uparrow(k) W_0(k) U^\uparrow(-k)^{\mathrm{t}}$. In particular, this relation yields
\[ \det(W_1(k)) = \det(U^{\uparrow}(k)) \det(U^{\uparrow}(-k)) \det(W_0(k)), \quad k \in \mathbb{T}^{1}, \]
and at the fixed points
\[ \Pf(W_1(k_{\star})) = \det(U^{\uparrow}(k_{\star})) \Pf(W_0(k_{\star})), \quad k_{\star} \in \{0,\pi\}. \]

To compute the Teo--Kane invariants of $W_0$ and $W_1$, we proceed as indicated in Remark~\ref{def:gamma_invariant}. Consider $W_0$ first, and choose $\mu_0 \colon [-\pi,\pi] \to \R$ such that $\det(W_0(k)) = e^{i \mu_0(k)}$, as well as $\lambda_0(0), \lambda_0(\pi) \in \R$ so that $\Pf(W_0(k_{\star})) = e^{i \lambda_0(k_{\star})}$, $k_{\star} \in \{0,\pi\}$. Choose also a continuous map $\alpha \colon [-\pi,\pi] \to \R$ such that $\det(U^\uparrow(k))=e^{i\alpha(k)}$. The relations shown above imply that the following are appropriate choices for the computation of the Teo--Kane invariant of $W_1$:
\begin{gather*}
\mu_1(k) := \mu_0(k) + \alpha(k) + \alpha(-k) \quad \text{is such that} \quad \det(W_1(k)) = e^{i \mu_1(k)}, \quad k \in \mathbb{T}^{1}, \\
\lambda_1(k_{\star}) := \lambda_0(k_{\star}) + \alpha(k_{\star}) \quad \text{is such that} \quad \Pf(W_1(k_{\star})) = e^{i \lambda_1(k_{\star})}, \quad k_{\star} \in \{0,\pi\}.
\end{gather*}
By definition, we have now
\begin{align*}
\TK(W_1(\cdot)) & = \exp \left(i [\mu_1(0) - 2 \lambda_1(0)]/2 \right) \exp \left(i [\mu_1(\pi) - 2 \lambda_1(\pi)]/2 \right) \\
& = \exp \left(i [\mu_0(0) + 2 \alpha(0) - 2 \lambda_0(0) - 2 \alpha(0)]/2 \right) \\
& \quad \exp \left(i [\mu_0(\pi) + \alpha(\pi) + \alpha(-\pi) - 2 \lambda_0(\pi) - 2 \alpha(\pi)]/2 \right) \\
&= \exp \left(i [\mu_0(0) - 2 \lambda_0(0)]/2 \right) \exp \left(i [\mu_0(\pi) - 2 \lambda_0(\pi)]/2 \right) \exp \left(- i [\alpha(\pi) - \alpha(-\pi)]/2 \right) \\
& = \TK(W_0(\cdot)) \exp \left(- i \pi [\det(U^{\uparrow}(\cdot))] \right)
\end{align*}
as wanted. The last identity in the statement follows from the relation~\eqref{eq:U_DIII} between $U^{\uparrow}$ and $U^{\downarrow}$, which implies that the determinant of the latter has an opposite winding number with respect to the one of the former. (Notice that the distinction is anyway immaterial, as only the parity of the winding number contributes to the sign in the stated equality.)
\end{proof}

\section{Conclusion and outlook}

In this work we have developed a unified and explicit framework for the classification of low-dimensional projection-valued maps in the presence of the symmetries appearing in the Altland--Zirnbauer--Cartan classification. The guiding idea throughout the paper has been that the classification of topological phases should not be reduced to a single equivalence relation, but rather decomposed into a hierarchy of geometric problems. Symmetric bases, Murray--von Neumann equivalence, unitary equivalence and homotopy equivalence describe different levels of rigidity of the same underlying object. Each of them captures a different geometric feature of the family of projections and each of them has a distinct physical interpretation in terms of localized bases, changes of representation, and gap-preserving deformations (compare Appendix~\ref{hamepro}).

The main outcome of this analysis is the distinction between two types of topological information. The first one is represented by what we have called \emph{absolute invariants}. These invariants belong intrinsically to a projection-valued map, or to one distinguished component of a symmetric pair, and they measure the obstruction to constructing globally periodic symmetric bases. In dimension $2$, these obstructions are encoded by the familiar integer-valued Chern numbers and $\Z_2$-valued Fu--Kane--Mele invariants. These invariants then also govern Murray--von Neumann equivalence and unitary equivalence.

The most delicate and conceptually interesting part of the classification emerges only when one compares unitary equivalence with homotopy equivalence. In non-minimal Hilbert spaces, the two notions coincide for \(d \leq 2\): the unused part of the ambient Hilbert space provides enough ``room'' to ``unwind'' the topology of the intertwining unitary and to construct the required homotopy without breaking the symmetry constraints. In this regime, the classification is therefore completely controlled by the absolute invariants. The situation changes substantially in minimal Hilbert spaces, particularly in the chiral and particle-hole symmetric classes. When the pair \(P^\pm(k)\) spans the entire Hilbert space, there is no longer an auxiliary orthogonal sector available to absorb the topology of the unitary intertwiner. A second layer of topological data emerges: the \emph{relative} homotopy invariants. The fact that these topological quantities arise only in minimal, finite-dimensional ambient Hilbert spaces makes them amenable to be observed only in tight-binding models of topological insulators (which have a finite number of degrees of freedom per unit cell), and challenges their derivation from the continuum~\cite{shapiro2022continuum}. 

This clear distinction between absolute and relative invariants is one of the central conceptual contributions of the paper. Absolute invariants characterize the pair of projection-valued maps in itself. Relative invariants, by contrast, are intrinsically relational: the topological label does not describe one pair of projection in isolation. Rather, different ``observers'' may choose different unitary representation for the same model, and thus disagree on which phases are ``topological'' and which are not: they will only agree upon the \emph{change} of ``topological charge'' when the model switches from one phase to the other. This phenomenology is akin to that of macroscopic electric polarization in solids~\cite{resta1994}, where the polarization observable is ill-defined, but \emph{differences} of polarization are measurable and carry a geometric interpretation.

The basis construction developed in the paper also has an important interpretative value. In dimension $1$, symmetric periodic bases can always be constructed in the relevant cases. In dimension $2$, the obstruction to full periodicity is localized in a controlled pseudo-periodic behaviour. This means that the topology of the family can be seen directly at the level of the boundary mismatch of a basis over the fundamental domain of the torus. This provides a concrete bridge between the abstract topological classification and the construction of Wannier-type bases. In particular, the failure of periodic symmetric bases is closely related to the impossibility of constructing globally well-behaved localized orbitals compatible with the symmetries of the system. Conversely, the pseudo-periodic basis construction gives an explicit and computable way to represent non-trivial topology. The discussion of Parseval frames from Section~\ref{sec:Parseval} further enriches this picture. The obstruction to constructing an orthonormal periodic basis can be traded for a controlled loss of orthonormality. In this sense, pseudo-periodic orthonormal bases and periodic Parseval frames represent two complementary ways of managing the same topological obstruction. The former preserve orthonormality but sacrifice full periodicity; the latter preserve periodicity but relax the strict orthonormal basis condition. This dual perspective suggests that the classification results may be useful not only at the level of abstract topology but also in computational approaches to topological materials, where one often needs bases or frames that are sufficiently localized, symmetric and numerically stable.

As for possible developments, the methods used here rely heavily on the fact that the topology of the torus and of the relevant unitary groups is sufficiently controlled in low dimension. Nevertheless, the structure of the proof suggests a possible strategy for higher-dimensional generalizations. One should expect that the construction of pseudo-periodic bases can still be performed inductively, but that the matching data will no longer be exhausted by one-dimensional winding numbers or by the two-dimensional Fu--Kane--Mele obstruction. Higher homotopy classes of unitary groups, higher Chern classes and more refined equivariant invariants should enter the classification. Extending the present framework to dimensions higher than $2$ would therefore require new technical ingredients, but the conceptual architecture should remain the same: one first constructs symmetric pseudo-periodic bases, then identifies the matching invariants and finally compares Murray--von Neumann, unitary, and homotopy equivalence. Alternatively, at last for top-dimensional invariants, one could exploit the ``Bott-clock'' relations in the periodic table of topological insulators to track how the invariants arising in one AZC class in a certain dimension can be transferred to a ``neighbouring'' symmetry class in the next dimension~\cite{StoneChiuRoy2011, KennedyZirnbauer2015, SantiThesis}.

In conclusion, the paper provides a detailed geometric classification of low-dimensional AZC symmetric projection-valued maps and clarifies the relations between several notions of equivalence which are often implicitly conflated. The classification shows that absolute invariants govern bases, Murray--von Neumann equivalence and unitary equivalence, while relative invariants appear precisely when one studies homotopy in minimal symmetric settings. This separation gives a clearer mathematical picture of topological phases of matter. The resulting framework unifies explicit basis constructions, matching matrix techniques, symmetry-preserving intertwiners and homotopy-theoretic arguments into a single coherent picture.

The broader message is that the topology of quantum phases is not encoded only in stable labels but also in the geometry of how projections are represented, based, intertwined and deformed. By making these distinctions explicit, the present work provides both a concrete classification in dimensions \(d \leq 2\) and a conceptual foundation for future extensions to higher dimensions, crystalline symmetries, and systems with broken periodicity.

\appendix

\section{Geometry in gapped quantum systems} \label{app:A}

This Appendix is devoted to recalling how the mathematical description of topological quantum matter is formulated. As this topic is well-studied, we will be brief and highlight only some points which are relevant for the discussion in the main text: we refer the reader to~\cite{reed1978iv, Panati_2013, kuchment2016overview, Monaco_2018, lewin2024spectral} for further details and considerations.

\subsection{From gapped periodic Hamiltonians to projection-valued maps}\label{hamepro}

In condensed matter physics, the configuration space for the $d$-dimensional crystalline solid in which a quantum particle is moving is described by a set $X\subset \R^d$, which we'll call the \emph{crystal}, that is left invariant by the action of translations in a \emph{Bravais lattice} $\Lambda\simeq \Z^d$, i.e.\ such that $X=\lambda+X$ for all $\lambda\in\Lambda$. There are two typical classes of models that are used: \emph{continuum models}, in which $X=\mathbb{R}^d$, and \emph{tight-binding models}, in which $X$ is a discrete set of atomic sites, periodic under the action of $\Lambda$. We will treat both classes on the same footing.

The Hamiltonian $H$ for the quantum particle moving in the crystal acts on the space of square-integrable wavefunctions $L^2(X) \otimes \C^N$, where $\C^N$ encodes possible internal degrees of freedom, like spin. The action of translations by Bravais lattice vectors on the crystal lifts to a suitable family of pairwise-commuting translation operators $\{T_\lambda\}_{\lambda \in \Lambda} \subset \mathcal{U}(L^2(X) \otimes \C^N)$, with which $H$ is required to commute: we will say that the Hamiltonian $H$ is \emph{periodic}. One can exploit this form of discrete translation invariance and pass to a (\emph{quasi-}, or \emph{crystal}, or \emph{Bloch}) \emph{momentum} representation. In order to describe this passage, let us first introduce some further piece of notation. The choice of a \emph{Wigner--Seitz cell} $\mathbb{W}$, i.e.\ of a set of representatives for the quotient $X / \Lambda$, induces an isomorphism
\[
L^2(X) \otimes \C^N \simeq \ell^2(\Lambda) \otimes [L^2(\mathbb{W}) \otimes \C^N],
\]
which can be called a \emph{dimerization} of the crystal. Notice that in continuum models the Wigner--Seitz cell is still continuous (typically a box in $\mathbb{R}^{d}$), while in tight-binding models it consists of a finite set of points. We stress how different choices of this cell will produce unitarily equivalent representations of the quantum system at hand~\cite{yang2020unit}: this could be a first source of `unitary equivalence' whose relevance to topological properties of the material under scrutiny has been investigated in the main body of the paper (see Sections~\ref{sec:unitary_equivalences} and~\ref{section:Dimerization_ambiguity}). Let us also introduce the \emph{dual Bravais lattice}
\[
    \Lambda^*:=\{\lambda^*\in \R^d : \lambda^*\cdot \lambda \in 2\pi \Z \; \forall\lambda \in \Lambda\},
\]
and the corresponding periodicity cell $\mathbb{B} := \R^d / \Lambda^{*}$, the \emph{Brillouin zone}, which we regard topologically as a $d$-dimensional torus.

The momentum representation which aids the study of periodic Hamiltonians is achieved via the \emph{Bloch--Floquet--Zak transform}
\begin{equation} 
(\UZ \psi) (k,y) := \sum_{\lambda\in\Lambda} e^{-ik\cdot (y-\lambda)} (T_\lambda \psi) (y) = \hat{\psi}_k(y), \quad k \in \mathbb{R}^{d}, \; y \in \mathbb{R}^{d},
\label{eq:BFZ}
\end{equation}   
which is seen to satisfy the following pseudo-periodicity conditions:
\begin{equation} \label{tau_lambda}
(T_{\lambda} \hat{\psi}_k)(y)=\hat{\psi}_k(y) \text{ and } \hat{\psi}_{k-\lambda^*}(y)=e^{i\lambda^*\cdot y} \hat{\psi}_k(y) =: (\tau_{\lambda^*} \hat{\psi}_k)(y) \quad \forall \lambda\in\Lambda, \: \lambda^* \in \Lambda^*, \: k , \, y \in \R^d. 
\end{equation}
At fixed quasi-momentum $k \in \mathbb{R}^{d}$, the first equality above establishes the function $\hat{\psi}_k$ as an element of the \emph{fiber Hilbert space} $\mathcal{H} := L^2_{\textup{per}}(\mathbb{W}) \otimes \C^N$, which accounts for all degrees of freedom inside the Wigner--Seitz cell, while the second equality defines the set of commuting unitary operators $\{\tau_{\lambda^*}\}_{\lambda^* \in \Lambda^*}$ and states that the mapping $\mathbb{R}^{d} \ni k \mapsto \hat{\psi}_k \in \mathcal{H}$ is \emph{$\tau$-equivariant}: in particular, its values are completely determined by the ones it attains over the Brillouin zone $\mathbb{B}$. With these considerations, we can reinterpret $\UZ$ as a unitary operator
\[ \UZ \colon L^2(X) \otimes \C^{N} \to \int_{\mathbb{B}}^{\oplus} \mathcal{H} dk \]
where the arrival space is a constant-fiber direct integral of Hilbert spaces.

The usefulness of $\UZ$ consists in the fact that any (linear or anti-linear) operator that commutes with the translations $T_\lambda$ admits a fiber decomposition.
In particular, the Hamiltonian $H$ gets decomposed into a family of \emph{fiber Hamiltonians}, namely self-adjoint operators $\tilde{H}(k)$, $k \in\mathbb{B}$, each acting in $\mathcal{H}$:
\[
    \UZ H \UZ^{-1} = \int_\mathbb{B}^\oplus \tilde{H}(k) dk.
\]
Under quite general and mild assumptions on $H$, the fiber Hamiltonians satisfy a number of properties, listed below (compare the notion of \emph{analytic family of type (A)} from~\cite{kato2013perturbation}):
\begin{enumerate}[label={\it \roman*.}, leftmargin=*]
    \item they are a \emph{$\tau$-covariant} family of operators:
\begin{equation}
\label{eq:taucov}
    \tilde{H}(k-\lambda^*) = \tau_{\lambda^*} \tilde{H}(k) \tau_{\lambda^*}^{-1} \quad \forall k\in \R^d, \lambda^*\in\Lambda^*,
\end{equation}
    and as such, the knowledge of the operators $\tilde{H}(k)$ for $k \in \mathbb{B}$ is sufficient to reconstruct the whole family;
    \item the domain of $\tilde{H}(k)$ within $\mathcal{H}$ does not depend on $k \in \mathbb{B}$;
    \item the set $\mathcal{R}= \{ (k,z)\in \R^d\times \C \ | \ z\in \rho (\tilde{H}(k)) \}$ is open and the resolvent map $(k,z)\mapsto (\tilde{H}(k)-z \Id_{\mathcal{H}})^{-1}$ is analytic on $\mathcal{R}$, with values in the algebra of compact operators on~$\mathcal{H}$; in particular, each $\tilde{H}(k)$ has only eigenvalues, called \emph{Bloch bands}, accumulating at infinity;
    \item the spectrum of $\tilde{H}$ is reconstructed as the union of the spectra of the fiber Hamiltonians as 
    \[ \sigma(H) = \bigcup_{k\in\mathbb{B}} \sigma (\tilde{H}(k)); \]
    the ranges of the Bloch bands as functions of $k \in \mathbb{B}$ define spectral bands for $H$ which are possibly interspersed with \emph{spectral gaps}.
\end{enumerate}

In view of the last point, we'll assume that the Hamiltonian models a \emph{topological insulator}, that is, that the set $\Omega$ of relevant energy bands is a spectral island isolated by open spectral gaps from the rest of the spectrum. If $\mathcal{C}$ is a closed and simple complex curve around $\Omega$, then the \emph{Riesz formula} 
\[\tilde{P}_\Omega (k) := \frac{i}{2\pi}\int_{\mathcal{C}} \left(\tilde{H}(k) - z\Id_{\mathcal{H}} \right)^{-1} dz, \quad k \in \mathbb{B}, \] 
computes the spectral eigenprojection of $\tilde{H}(k)$ corresponding to Bloch bands inside $\Omega$. By functional calculus, $\tilde{P}_\Omega(k)$ inherits the regular and $\tau$-covariant dependence on $k$ from the resolvent map $(\tilde{H}(k)-z\Id_{\mathcal{H}})^{-1}$. In particular, we find that the rank of this projection is constant in $k$, since $\operatorname{rank}(\tilde{P}_{\Omega}(k)) = \tr_{\mathcal{H}}(\tilde{P}_{\Omega}(k)) \in \N$ defines a continuous map of $k \in \mathbb{R}^{d}$ towards a discrete set. Setting $n := \operatorname{rank}(\tilde{P}_{\Omega}(k))$, we therefore have $\tilde{P}_{\Omega}(k) \in \Proj_n (\Hi)$, $k \in \mathbb{R}^{d}$.

The property of $\tau$-covariance for fiber operators coming from the Bloch--Floquet--Zak representation can be traded for standard $\Lambda^*$-periodicity~\cite[Section~2.1]{cornean2016construction}. Indeed, choose a lattice basis $\{\lambda_j^*\}_{j\in\{1,\ldots, d\}}$ for $\Lambda^*$. We can use the spectral theorem and obtain bounded self-adjoint operators $L_1,\ldots, L_d$ such that $\tau_{\lambda_j^*}=e^{iL_j}$, $j \in \{1,\ldots, d\}$: since the $\tau_{\lambda_j^*}$'s commute among each other, the choice of the operators $L_j$ can be performed so to have also $[L_i, L_j] = 0$ for all $i, j \in \{1,\ldots, d\}$. For $k \in \mathbb{R}^{d}$, write now $k = k_1 \lambda_1^* +\cdots+k_d \lambda_d^*$ in the chosen lattice basis, and define the unitary operator
\begin{equation} \label{eq:tau(k)}
\tau(k):=\prod_{j=1}^{d} e^{ik_jL_j} = e^{i(k_1 L_1 + \cdots + k_d L_d)} \in \mathcal{U}(\mathcal{H}), \quad k \in \R^d. 
\end{equation} 
The dependence of $\tau(k)$ on $k$ is analytic. Then
\begin{equation} \label{eq:periodic_H(k)}
    H(k):= \tau(k)^{-1} \tilde{H}(k) \tau(k)
\end{equation}
defines a \emph{$\Lambda^*$-periodic} family of operators, which satisfies the same regularity conditions on $k$ as the ones of $\tilde{H}(k)$; correspondingly, $P_{\Omega}(k) := \tau(k)^{-1} \tilde{P}_{\Omega}(k) \tau(k)$ defines a periodic projection-valued map $P_{\Omega} \colon \mathbb{T}^{d} \to \Proj_n(\mathcal{H})$. This is exactly the type of mathematical object that has been investigated in this paper.

Let us highlight how this procedure of reduction of $\tau$-covariance to periodicity is another possible source of `unitary equivalence'. Indeed, the choice of $\tau(k)$ in~\eqref{eq:tau(k)} is not unique, and any other choice is related to it via a $\Lambda^*$-periodic unitary-valued map $U \colon \mathbb{T}^{d} \to \mathcal{U}(\mathcal{H})$.

\begin{remark}[Bloch bundle] \label{rmk:Bloch_bundle}
By the Serre--Swan theorem (see~\cite{swan1962vector} and~\cite[Section~3.3.7]{rordam2000introduction}), the above projection-valued map $P_{\Omega}$ corresponds to the \emph{Bloch bundle}, namely a vector sub-bundle $E_\Omega \xrightarrow{\pi} \mathbb{T}^{d}$ of the trivial bundle over the $d$-dimensional torus with typical fiber $\mathcal{H}$, whose fiber over $k \in \mathbb{T}^{d}$ is exactly $\Imm(P_\Omega(k)) \subset \mathcal{H}$: a concrete definition of the Bloch bundle can be found in~\cite{panati2007triviality, monaco2015symmetry} starting directly from the $\tau$-covariant projections $\tilde{P}_{\Omega}(k)$, $k \in \mathbb{R}^{d}$. Let us stress that, in these geometric terms, the notion of Murray--von Neumann equivalence for projection-valued maps $P, Q \colon \mathbb{T}^{d} \to \Proj_n(\mathcal{H})$ corresponds to that of \emph{isomorphism} for the corresponding bundles (compare e.g.~\cite[Theorem~6]{fiorenza2016z}): indeed, for $k \in \mathbb{T}^{d}$, a partial-isometry-valued intertwiner $V(k)$ yields a fiberwise unitary isomorphism between the fiber $E_k := \Imm(P(k))$ of the bundle associated to $P$ and the fiber $F_k := \Imm(Q(k))$ of the bundle associated to $Q$. Since fiberwise isomorphisms are bundle isomorphisms~\cite[Chapter~3, Theorem~2.5]{Husemoller1994}, the conclusion follows.

The results of Section~\ref{sec:Bases} can thus be seen as a characterization of isomorphism invariants for the Bloch bundles arising from spectral eigenprojections of gapped periodic Hamiltonians (endowed with further symmetries, see the next Section).
\end{remark}

\begin{remark}[On the regularity of projection-valued maps]\label{rmk:regularity_required}
As we have just reviewed, pro\-jec\-tion-valued maps arising from gapped quantum systems are usually \emph{analytic} in $k$. However, in this paper, we have just asked for \emph{continuity}, also in order to keep a wider generality. At any rate, no information would be lost or gained if one were to retain the original analytic regularity: in view of the previous Remark, this observation is a consequence of the \emph{Oka--Grauert principle}~\cite{oka1939fonctions, grauert1958analytische}, which roughly speaking states that the topological classification of vector bundles yields the same results as their analytic classification. More concretely, one can also invoke explicit algorithms which allow, for example, to promote a continuous basis for a projection-valued map to an analytic one (possibly preserving any symmetric constraints of the type detailed below): see e.g.~\cite[Lemma~2.3]{cornean2016construction},~\cite[Appendix~A]{fiorenza2016construction} and~\cite[Appendix~A]{fiorenza2016z}.
\end{remark}

\subsection{Symmetries in topological insulators}

The quantum system at hand could present some additional symmetries of the type described in the main body of the paper: time-reversal symmetry $\Theta$, particle-hole symmetry $\Gamma$, or chiral symmetry $\Sigma$. These are (anti)unitary operators commuting or anticommuting with the Hamiltonian: in particular,
\[ \Theta H = H \Theta, \quad \Gamma H = - H \Gamma, \quad \Sigma H = - H \Sigma. \]
We will make the following customary assumption on these symmetry operators:
\begin{enumerate}[leftmargin=*]
    \item each symmetry operator commutes with the translation operators $\{T_\lambda\}_{\lambda \in \Lambda}$, and as such admits a Bloch--Floquet--Zak decomposition:
    \[ \UZ \Theta \UZ^{-1} = \int_{\mathbb{B}}^{\oplus} \theta(k) dk, \quad \UZ \Gamma \UZ^{-1} = \int_{\mathbb{B}}^{\oplus} \gamma(k) dk, \quad \UZ \Sigma \UZ^{-1} = \int_{\mathbb{B}}^{\oplus} \sigma(k) dk; \]
    the fibers $\theta(k)$ and $\gamma(k)$ define antiunitary operators on $\mathcal{H}$, while the fibers $\sigma(k)$ define unitary operators on $\mathcal{H}$, $k \in \mathbb{B}$;
    \item the fibers of the symmetry operators are constant in $k$:
    \[ \theta(k) \equiv T, \quad \gamma(k) \equiv C, \quad \sigma(k) \equiv S. \]
\end{enumerate}

Observe that antiunitary operators conjugate the phase inside~\eqref{eq:BFZ} of $\UZ$: therefore, the fibers of antiunitary symmetry operators link the fiber over $k \in \mathbb{B}$ of the direct-integral arrival space of the Bloch--Floquet--Zak transform with the fiber over $-k$. Moreover, by a similar argument these fiber operators interact with the operators $\{\tau_{\lambda^*}\}_{\lambda^* \in \Lambda^*}$ as  follows: for $k\in \R^d$ and $\lambda^* \in \Lambda^*$,
\[
\theta(k-\lambda^*) \tau_{\lambda^*} = \tau_{-\lambda^*} \theta(k), \quad \gamma(k-\lambda^*) \tau_{\lambda^*} = \tau_{-\lambda^*} \gamma(k), \quad \sigma(k-\lambda^*) \tau_{\lambda^*} = \tau_{\lambda^*} \sigma(k).
\]

These observations, together with the assumption of constancy of the fibers of the symmetry operators, yield the following two conclusions. First of all, at the level of Bloch--Floquet--Zak fibers, the defining symmetry constraints read, for $k \in \mathbb{R}^{d}$, 
\begin{equation}\label{eq:fiber_symmetric_condition}
T \tilde{H}(k)=\tilde{H}(-k)T, \quad C \tilde{H}(k)=- \tilde{H}(-k)C, \quad S \tilde{H}(k)=-\tilde{H}(k)S,
\end{equation}
and the same relations are inherited by the spectral eigenprojections $\tilde{P}_{\Omega}(k)$. Secondly, the commutation relations among the symmetries' fibers and the operators $\{\tau_{\lambda^*}\}_{\lambda^* \in \Lambda^*}$ allow to refine the construction of the operators $\tau(k)$ in~\eqref{eq:tau(k)} in order to have
\[
    T \tau(k)=\tau(-k)T, \quad C\tau(k)=\tau(-k)C \quad S\tau(k)=\tau(k)S.
\]
The combination of the two sets of identities above allow to conclude that the periodic fibers $H(k)$ in~\eqref{eq:periodic_H(k)} still satisfy the symmetry constraints in~\eqref{eq:fiber_symmetric_condition}, and that any other choice of periodization of the family $\tilde{H}(k)$ is unitarily related to $H(k)$ via a unitary-valued map satisfying the symmetry constraints~\eqref{MvNsymmetry}.

These symmetry constraints have also natural implications at the level of the projection-valued map $P_\Omega \colon \mathbb{T}^{d} \to \Proj_n(\mathcal{H})$. Indeed, if the Hamiltonian is time-reversal symmetric, the Riesz formula gives
\[ TP_\Omega(k)=P_\Omega(-k)T \quad \forall k \in \T^d. \]
Instead, if the Hamiltonian is particle-hole and/or chiral symmetric, the anticommutation relations with the symmetry operators force the spectrum of $H$ to be symmetric around zero. This readily implies that
\[ C P_\Omega (k) = P_{-\Omega}(-k)C \quad \text{and/or} \quad SP_\Omega(k) =P_{-\Omega}(k)S \quad \forall k \in \T^d. \]
This gives that, on the one hand, if the spectral band $\Omega$ contains $0$ and $\Omega=-\Omega$, then $P := P_{\Omega} = P_{-\Omega}$ is a covariant projection-valued map. On the other hand, if the spectral island does not contain $0$, then the symmetries intertwine the elements of the pair of projection-valued maps $P^{\pm} := P_{\pm \Omega}$. We have thus recovered the notions presented in Definition~\ref{def:Symmetries}.

\section{Time-reversal symmetric periodic bases in infinite rank} \label{app:XYZ}

The following result details the construction of symmetric periodic bases for (time-reversal symmetric) infinite-rank projection-valued maps on the torus of any dimension.

\begin{proposition}
Any infinite-rank projection-valued map $P:\T^d \to \Proj_\infty (\Hi)$ admits a periodic basis for all $d\in\N$. If $P$ is time-reversal symmetric, the basis can also be taken to be time reversal symmetric.
\end{proposition}

\begin{proof}
Let $\{v_j\}_{j\in\N}$ be an orthonormal basis for the range of $P(0)$: the latter is an infinite-dimensional subspace of $\mathcal{H}$, which is moreover fixed by the time-reversal symmetry operator~$T$ in case this is present. As such, the restriction of $T$ onto $\Imm(P(0))$ defines a time-reversal symmetry on this subspace, and as argued in Section~\ref{sec:Bases_class_AIAII} we can impose $Tv_j =v_j$ if the symmetry is even or $v_{2j+2}=-Tv_{2j+1}$ if the symmetry is odd, for all $j\in\N$.

To extend the basis to the whole torus, we first use the construction sketched in the proof of Principle \ref{pri:unitary_eq_to_homotopy} coordinate-wise, which produces a symmetric Kato--Nagy intertwining unitary $U \colon [-\pi,\pi]^d \to \U (\Hi)$ such that $P(t)U(t)=U(t)P(0)$, possibly satisfying the symmetry constraint $TU(t)=U(-t)T$: we write $t \in [-\pi,\pi]^d$ and not $k \in \mathbb{T}^{d}$ to emphasize that $U$ does not in general depend periodically on its parameters. Setting
\[ \tilde{v}_j(t) := U(t) v_j, \quad j \in \mathbb{N}, \]
defines a (symmetric) non-periodic basis for the range of $P(t)$, $t \in [-\pi,\pi]^{d}$. The choice of this basis can be used to define a family of unitary operators $I(t) \colon \Imm (P(t))\to \ell^2(\N)$ specified by the relation $I(t)\tilde{v}_j(t)=\delta_j$, $t \in [-\pi,\pi]^{d}$, where $\delta_j = \big(\delta_{j,l}\big)_{l \in \mathbb{N}} \in \ell^2(\mathbb{N})$ and $\delta_{j,l}$ is the Kronecker delta. The time-reversal symmetry operator on $\mathcal{H}$, whenever present, is mapped to a corresponding normal form on $\ell^2(\mathbb{N})$: denoting by $\mathcal{K}$ the standard complex conjugation on $\ell^2(\mathbb{N})$ and $\mathcal{J}\in\U(\ell^2(\Hi))$ the unitary such that 
\[ \mathcal{J}(\delta_{2j+2})=-\delta_{2j+1}, \quad \mathcal{J}(\delta_{2j+1})=\delta_{2j+2}, \quad j \in \mathbb{N}, \]
then if $T$ is even we have $TI(t)=I(-t)\mathcal{K}$, while if $T$ is odd we have $TI(t)=I(-t)\mathcal{KJ}$ (compare respectively~\eqref{eq:TRS_normal_form_even} and~\eqref{eq:TRS_normal_form_odd} in finite-dimensions). In the following, we therefore denote
\[ \mathcal{T} \colon \ell^2(\mathbb{N}) \to \ell^2(\mathbb{N}), \quad \mathcal{T} := \begin{cases}
\mathcal{K} & \text{if $T$ is even}, \\
\mathcal{K} \mathcal{J} & \text{if $T$ is odd}.
\end{cases} \]

In order to modify the basis $\{\tilde{v}_j(t)\}_{j \in \mathbb{N}}$ and enforce periodicity, we observe that the choice of a different basis for $\Imm(P(t))$ is equivalent to selecting a (non-periodic) unitary-valued map $J \colon [-\pi,\pi]^d\to \U(\ell^2(\N))$ and setting 
\[ v_j(t):=I(t)^{-1} J(t)\delta_j, \quad j \in \mathbb{N}, \; t \in [-\pi,\pi]^d . \]
This new basis will be periodic if and only if for every pairs of points $t\sim t'\in[-\pi,\pi]^d$, that is, points whose coordinates differ by integer multiples of $2\pi$, we have \begin{equation}\label{eq:lemma_inf_periodic_condition} J(t')^{-1}I(t')^{-1} I(t)J(t)=\Id_{\ell^2(\N)}. \end{equation}
Correspondingly, the new basis is symmetric if and only if \begin{equation}\label{eq:lemma_inf_symmetric_condition} 
\mathcal{T}J(t)=J(-t)\mathcal{T}.
\end{equation}
The thesis will then follow if we show that a unitary-valued map $J$ satisfying~\eqref{eq:lemma_inf_periodic_condition} and~\eqref{eq:lemma_inf_symmetric_condition} can be constructed.

To this end, we use a triangulation of the cube $[-\pi,\pi]^d$ into $d$-simplexes defined iteratively to be compatible with the inversion $t \mapsto -t$. The 0-skeleton of this triangulation contains the points $t_* \in [-\pi,\pi]^{d}$ with coordinates equal to $-\pi$, $0$ or $\pi$; the $1$-skeleton is obtained by joining these points, that is, each edge $[-\pi,\pi]$ is divided into $[-\pi,0]\cup [0,\pi]$; in general, every $(n-1)$-simplex of the triangulation is promoted to a $n$-simplex that has the center as an additional vertex, for $n \le d$. Notice that the points $t_*$ in the $0$-skeleton are the only ones on which both conditions~\eqref{eq:lemma_inf_periodic_condition} and~\eqref{eq:lemma_inf_symmetric_condition} are to be enforced simultaneously, since $t_* \sim -t_*$. We thus start the construction of $J$ on these points, and then show how to move to higher-dimensional simplexes.

In absence of time-reversal symmetry, for every orbit of the points $t_*$ as above under the equivalence relation $\sim$ (i.e.\ for every $k_{\star} = [t_*]_{\sim} \in \mathbb{T}^{d} = [-\pi,\pi]^{d} / \sim$), we choose a representative $t_\star$ and fix the value $J(t_\star)=\Id_{\ell^2(\N)}$. The condition~\eqref{eq:lemma_inf_periodic_condition} is enforced if, on the other points $t_*' \sim t_{\star}$ in the orbit, we set $J(t_*'):=I(t_*')^{-1} I(t_{\star})$.

In presence of time-reversal symmetry, we want to enforce both of the above conditions on $J$ on the points $t_*$. Let us then choose again a representative $t_\star$ in each orbit $[t_*]_{\sim}$. Writing $J(t_{\star}) = \mathcal{T} J(t_{\star}) \mathcal{T}^{-1}$ from~\eqref{eq:lemma_inf_symmetric_condition} and plugging this into~\eqref{eq:lemma_inf_periodic_condition} with $t' = - t_{\star}$, after some algebraic manipulations we see that $J(t_{\star})$ must solve the equation
\begin{equation} \label{eq:J(t_star)} 
\mathcal{I}(t_{\star}) = J(t_{\star}) \mathcal{T} J(t)^{-1} \mathcal{T}^{-1}, \quad \text{where} \quad \mathcal{I}(t_{\star}) := I(t_{\star})^{-1} T I(t) \mathcal{T}^{-1}.
\end{equation}
The spectral theorem allows us to write $\mathcal{I}(t_{\star})= e^{iX}$ with $X=X^*$ on $\ell^2(\mathbb{N})$. Since $\mathcal{T} \mathcal{I}(t_{\star}) \mathcal{T}^{-1} = \mathcal{I}(t_{\star})^{-1}$, as can be easily checked, the above $X$ can be chosen to satisfy $X= \mathcal{T} X \mathcal{T}^{-1}$, by antiunitarity of $\mathcal{T}$. It is easily seen that $J(t_*):=e^{iX/2}$ solves~\eqref{eq:J(t_star)}; Equation~\eqref{eq:lemma_inf_periodic_condition} then forces the values of $J$ at the other points $t_*' \sim t_{\star}$ in the orbit to be $J(t_*'):=I(t_*')^{-1} I(t_*) e^{iX/2}$.

The remainder of the construction of $J$ can be done inductively on the $n$-skeleton of $[-\pi,\pi]^d$, $n \in \{1, \ldots, d\}$, with repeated applications of Kuiper's theorem~\cite{kuiper1965homotopy}: the weak contractibility of $\U(\ell^2(\N))$ always allows to extend the value of a map $J \colon \partial \triangle\to\U(\ell^2(\N)) $, defined on the boundary of a $n$-simplex $\triangle$, to a continuous map $J \colon \triangle\to\U(\ell^2(\N))$. In order to enforce the symmetry constraints~\eqref{eq:lemma_inf_periodic_condition} and~\eqref{eq:lemma_inf_symmetric_condition}, one has to proceed carefully: if $J$ is defined on a simplex $\triangle$, the extension on the symmetric counterpart $-\triangle$ is governed by Equation~\eqref{eq:lemma_inf_symmetric_condition}, while the extension to any periodic counterpart $\triangle'$ is governed by Equation~\eqref{eq:lemma_inf_periodic_condition}. The combination of the two procedures forces the extension to the simplex $-(\triangle')$. The properties of $I(t)$ ensure compatibility between~\eqref{eq:lemma_inf_periodic_condition} and~\eqref{eq:lemma_inf_symmetric_condition} in defining these extensions: if $t'\sim t$ and $J(t)$ is given, then $J(-t')$ can be defined equivalently as
\begin{align*}
J(-t')&=I(-t')^{-1}I(-t)J(-t)=\mathcal{T}I(t')^{-1}I(t)J(t)\mathcal{T}^{-1} \quad \text{or} \\
J(-t')&=\mathcal{T} J(t') \mathcal{T}^{-1} = \mathcal {T } I(t')^{-1}I(t)J(t) \mathcal{T}^{-1}.
\end{align*} 
This concludes the proof.
\end{proof}

\section{Further properties of the relative invariants} \label{app:Further_properties}

This Appendix explores further properties of the relative invariants introduced in Section~\ref{sec:Homotopy_minimal}, namely the chiral invariant (Definition~\ref{def:W_invariant}), the even-particle-hole invariant (Definition~\ref{def:P_invariant}), and the DIII invariant (Definition~\ref{def:DIII_invariant}).

\subsection{Expressions for the relative invariants via symmetric bases}

First we present formul\ae\ to compute the relative invariants through symmetric periodic bases for the $1$-dimensional projection-valued maps involved, in the spirit of the expression~\eqref{eq:I_as_Berry_phase} for the Zak phase. The restriction to $d=1$ is needed to ensure that these symmetric periodic bases actually exist, in view of the results of Section~\ref{sec:Bases}; at any rate, the $2$-dimensional chiral and DIII invariants are anyway defined via restriction to $1$-dimensional sub-tori. The formul\ae\ obtained below enable the explicit calculation of relative invariants e.g.\ in existing codes for computational modeling of condensed matter systems~\cite{mostofi2008wannier90, marrazzo2024wannier}. As a byproduct, the statements below show that these homotopy invariant are also gauge-invariant under changes of gauge which preserve the appropriate symmetries.

Throughout this Section, we assume that the involved projection-valued maps and symmetric periodic bases are smooth (meaning at least continuously differentiable), rather than just continuous: in applications, this comes at no loss of generality in view of Remark~\ref{rmk:regularity_required}.

\subsubsection{Chiral invariant}

Let us start from the chiral invariant $\W(P^{\pm})$ of a pair of projection-valued maps.

\begin{proposition}\label{prop:W_frame_description}
Let $P^\pm \colon \T^1 \to \Proj_n(\C^{2n})$ be a chiral symmetric pair of smooth projection-valued maps. If $\{v_1(k), \ldots, v_{2n}(k)\}$ is a smooth symmetric periodic basis for $P^{\pm}$, then 
\begin{equation} \label{eq:chiral_simil_Zak}
\W(P^\pm) = \frac{1}{\pi i} \int_{\T^1} \sum_{j=1}^{n} \inn{v_j(k)}{S \partial_k v_j(k)} dk \in \Z. 
\end{equation}
\end{proposition}
\begin{proof}
Let 
\begin{equation} \label{eq:tildevj}
\tilde{v}_j(k) := \dfrac{1}{\sqrt{2}} \begin{pmatrix}
        - W(k) e_j \\ e_j
    \end{pmatrix}, \quad j \in \{1,\ldots, n\}, \; k \in \mathbb{T}^{d},
\end{equation}
where $\{e_1, \ldots, e_n\}$ is an orthonormal basis for the eigenspace $\mathcal{H}^{\downarrow} = \ker(S+\Id_{2n}) \simeq \C^{n}$, $W \colon \mathbb{T}^{1} \to U(n)$ is as in~\eqref{chiral_symm_P+-P-}, and the decomposition of $\mathcal{H} \simeq \C^{2n}$ is the usual one into $\ker(S-\Id_{2n}) \oplus \ker(S+\Id_{2n})$. Then
$\{\tilde{v}_1(k), \ldots, \tilde{v}_n(k)\}$ is a smooth periodic basis for $P^{-}$, by a computation similar to the one presented in~\eqref{eq:basis_from_W}. Notice how a symmetric basis $\{\tilde{v}_1(k), \ldots, \tilde{v}_{2n}(k)\}$ for the pair $P^{\pm}$, as in Definition~\ref{def:Symmetric_basis}, is obtained by defining
\[ \tilde{v}_j(k) := S \tilde{v}_{2n-j+1}(k) \quad \forall \: j \in \{n+1, \ldots, 2n\}. \]
If we compute the integrand on the right-hand side of~\eqref{eq:chiral_simil_Zak}, we obtain, for $j \in \{1, \ldots, n\}$,
\begin{align*}
\sum_{j=1}^{n} & \inn{\tilde{v}_j(k)}{\partial_k S \tilde{v}_j(k)} = \sum_{j=1}^{n} \frac{1}{2} \inn{ \begin{pmatrix}
        - W(k) e_j \\ e_j
    \end{pmatrix} }{ \partial_k \begin{pmatrix}
        - W(k) e_j \\ -e_j
    \end{pmatrix} } \\
&= \sum_{j=1}^{n} \frac{1}{2} \inn{ \begin{pmatrix}
        - W(k) e_j \\ e_j
    \end{pmatrix} }{ \begin{pmatrix}
        - [\partial_k W(k)] e_j \\ 0
    \end{pmatrix} } = \sum_{j=1}^{n} \frac{1}{2} \inn{e_j}{W(k)^*\,[\partial_k W(k)] e_j} \\
& = \frac{1}{2} \tr_{\mathcal{H}^{\downarrow}} \left( W(k)^*\,[\partial_k W(k)] \right)
\end{align*}
yielding in particular (compare e.g.~\cite[Proposition~A.1]{monaco2023z2})
\[
\frac{1}{\pi i} \int_{\mathbb{T}^{1}} \sum_{j=1}^{n} \inn{\tilde{v}_j(k)}{S \partial_k \tilde{v}_j(k)} dk =  \frac{1}{2\pi i} \int_{\mathbb{T}^{1}} \tr_{\mathcal{H}^{\downarrow}} \left( W(k)^*\,[\partial_k W(k)] \right) dk = [\det(W(\cdot))] = \W(P^{\pm}).
\]

To conclude the proof, we need to show that the right-hand side of~\eqref{eq:chiral_simil_Zak} is also gauge-invariant, i.e., invariant under changes of the periodic basis chosen for $P^{-}$: notice that, in fact, a change in the basis for $P^{-}$ immediately dictates how to correspondingly change the basis for $P^{+}$, that is, by imposing the chiral symmetry constraint. Therefore, assume that
\[ v_j(k) := \sum_{i=1}^{n} G(k)_{i,j} \tilde{v}_i(k), \quad j \in \{1,\ldots,n\}, \]
defines a new smooth and periodic basis for the range of $P^{-}(k)$, for some unitary-valued map $G \colon \mathbb{T}^1 \to U(n)$. Since
\[ \partial_k v_j(k) = \sum_{i=1}^{n} [\partial_k G(k)_{i,j}] \tilde{v}_i(k) + G(k)_{i,j} [\partial_k \tilde{v}_i(k)], \]
we can then compute
\begin{align*}
\inn{v_j(k)}{\partial_k S v_j(k)} & = \sum_{i,\ell=1}^{n} \inn{G(k)_{\ell,j} \tilde{v}_\ell(k)}{S [\partial_k G(k)_{i,j}] \tilde{v}_i(k)} + \inn{G(k)_{\ell,j} \tilde{v}_\ell(k)}{S G(k)_{i,j} [\partial_k \tilde{v}_i(k)]} \\
& = \sum_{i,\ell=1}^{n} \overline{G(k)_{\ell,j}} \, [\partial_k G(k)_{i,j}] \inn{ \tilde{v}_\ell(k)}{S  \tilde{v}_i(k)} + \overline{G(k)_{\ell,j}} G(k)_{i,j} \inn{\tilde{v}_\ell(k)}{S  [\partial_k \tilde{v}_i(k)]}.
\end{align*}
The first term vanishes, as $S\tilde{v}_i(k)$ is in $\Imm(P^{+}(k)) = \Imm(P^{-}(k))^{\perp}$. As for the second term, we observe that $\sum_{j=1}^{n} \overline{G(k)_{\ell,j}} G(k)_{i,j} = \delta_{i,\ell}$ due to the unitarity of $G(k)$, $k \in \mathbb{T}^{1}$. We obtain that
\[ \sum_{j=1}^{n} \inn{v_j(k)}{\partial_k S v_j(k)} =  \sum_{i=1}^{n} \inn{\tilde{v}_i(k)}{\partial_k S \tilde{v}_i(k)} \]
and the desired conclusion follows.
\end{proof}

\subsubsection{Even-particle-hole invariant}

The $1$-dimensional even-particle-hole invariant is a collection of two $0$-dimensional quantities computed at the fixed points of the involution $k \mapsto -k \in \mathbb{T}^{1}$. Still, Proposition~\ref{lem:relation_I-P} states that the product of these two quantities coincides with the Zak phase $\I(P^{\pm})$, which can be computed from a smooth periodic basis for $P^-$ via~\eqref{eq:I_as_Berry_phase}.

\subsubsection{DIII invariant}

For pairs of projection-valued maps $P^{\pm} \colon \mathbb{T}^{1} \to \Proj_{n}(\C^{2n})$ in class DIII, a symmetric basis consists of a time-reversal symmetric basis $\{v_1, \ldots, v_n\}$ for $P^{-}$ (hence made up of Kramers pairs, since the time-reversal symmetry operator is odd), together with the vectors $\{S v_1, \ldots, S v_n\}$ spanning the projection $P^{+}$ on the orthogonal complement. Write $n =: 2m \in 2 \mathbb{N}$. Coherently with~\eqref{eq:TRS_normal_form_odd}, we arrange then the vectors in the basis for $P^{-}$ so that the second $m$ vectors are the time-reversal symmetric images of the first $m$ vectors.

\begin{proposition}\label{prop:gamma_frame_description}
Let $P^\pm \colon \T^1 \to \Proj_n(\C^{2n})$ be a pair of smooth projection-valued maps in class DIII. If $\{v_1(k), \ldots, v_{2n}(k)\}$ is a smooth symmetric periodic basis for $P^{\pm}$, then 
\begin{equation} \label{eq:DIII_simil_Zak}
\G(P^\pm) = e^{i \pi \mathcal{A}^S} \in \Z_2 \quad \text{where } \mathcal{A}^S := \frac{1}{\pi i} \int_{\T^1} \sum_{j=1}^{m} \inn{v_j(k)}{S \partial_k v_j(k)} dk. 
\end{equation}
\end{proposition}
\begin{proof}
We first make an observation concerning the computation of $\mathcal{A}^S$ as in the statement. Assume that $\{v_1, \ldots, v_n\}$ is a time-reversal symmetric basis for $P^{-}$, i.e., $v_{j+m}(k) = T v_{j}(-k)$ for all $j \in \{1, \ldots, m\}$. This relation implies in particular
\[ \partial_k v_{j+m}(k) = \partial_k [T v_{j}(-k)] = - T (\partial_k v_{j})(-k), \quad j \in \{1, \ldots, m\}, \; k \in \mathbb{T}^{1}. \]
Notice that, for $g \colon \mathbb{T}^{1} \to U(1)$,
\begin{equation} \label{g(-k)}
\int_{0}^{\pi} g(-k) dk = - \int_{0}^{-\pi} g(k) dk = \int_{-\pi}^{0} g(k) dk.
\end{equation}
Using the antiunitarity of the time-reversal symmetry operator $T$ and the anticommutation relation $TS = - ST$, we can then compute, for $j \in \{1,\ldots, m\}$
\begin{align*}
\int_{-\pi}^{0} \inn{v_j(k)}{S \partial_k v_j(k)} dk & = \int_{0}^{\pi} \inn{v_j(-k)}{S (\partial_k v_j)(-k)} dk = \int_{0}^{\pi} \inn{TS (\partial_k v_j)(-k)}{T v_j(-k)} dk \\
& = - \int_{0}^{\pi} \inn{S T (\partial_k v_j)(-k)}{T v_j(-k)} dk = \int_{0}^{\pi} \inn{S \partial_k v_{j+m}(k)}{v_{j+m}(k)} dk \\
& = - \int_{0}^{\pi} \inn{v_{j+m}(k)}{S \partial_k v_{j+m}(k)} dk
\end{align*}
where in the last equality we used the identity
\[ \inn{v_{j+m}(k)}{S v_{j+m}(k)} \equiv 0 \quad \text{(and hence } \partial_k \inn{v_{j+m}(k)}{S v_{j+m}(k)} \equiv 0 \text{)} \]
as $S v_{j+m}(k) \in \Imm(P^{+}(k)) = \Imm(P^{-}(k))^{\perp}$. We conclude that
\begin{equation} \label{eq:AS}
\mathcal{A}^S = \frac{1}{\pi i} \int_{0}^{\pi} \sum_{j=1}^{m} \big( \inn{v_j(k)}{S \partial_k v_j(k)} - \inn{v_{j+m}(k)}{S \partial_k v_{j+m}(k)} \big) dk .
\end{equation}

After this general consideration, valid for any symmetric basis for $P^{-}$, we follow a similar strategy to that of Proposition~\ref{prop:W_frame_description} and specify to a particular basis. To this end, let us first of all fix an orthonormal basis $\{e_1, \ldots, e_n\}$ for $\mathcal{H}^{\downarrow} = \ker(S+\Id_{2n})$; without loss of generality, we choose the basis to be real, i.e., $\mathcal{K} e_j = e_j$ for all $j \in \{1, \ldots, n\}$. Acting with the time-reversal symmetry operator $T$ on the vector $\begin{pmatrix} 0 & e_j \end{pmatrix}^{\mathrm{t}} \in \{0\} \oplus \mathcal{H}^{\downarrow}$ yields the vector $\begin{pmatrix} \mathcal{K} e_j & 0 \end{pmatrix}^{\mathrm{t}} \in \mathcal{H}^{\uparrow} \oplus \{0\}$ where $\mathcal{H}^{\uparrow} = \ker(S-\Id_{2n})$, in view of the anticommutation relation $T S = - S T$; as in the proof of Proposition~\ref{lem:relation_delta_gamma}, upon this identification we can regard the collection $\{f_j := \mathcal{K} e_j = e_j\}_{j \in \{1, \ldots, n\}}$ as an orthonormal basis for $\mathcal{H}^{\uparrow}$ as well. With this choice of bases, the usual map $W$ from~\ref{chiral_symm_P+-P-} can be represented as a family of matrices $W(k) \in U(n)$, $k \in \mathbb{T}^1$, which further satisfy the symmetry relation~\eqref{eq:W^t=-W}. We now borrow the result of~\cite[Lemma~5.1]{Cornean_2017}, which states the existence of a unitary-valued map $\gamma \colon \mathbb{T}^{1} \to U(n)$ such that
\[ W(k) = \gamma(-k)^{\mathrm{t}} J_{n} \gamma(k), \quad k \in \mathbb{T}^{1}, \]
where $J_n$ is the $n \times n$ standard symplectic matrix, or, up to a redefinition of $\gamma$ by multiplication with a constant matrix, any unitary and skew-symmetric matrix (compare~\eqref{eq:DIIId0}). Without loss of generality, we can therefore assume that the chosen bases for $\mathcal{H}^{\downarrow}$ and $\mathcal{H}^{\uparrow}$ are tailored to the standard matrix $J_n$, i.e.\ that $J_n e_j = -e_{j+m}$ for $j \in \{1, \ldots, m\}$. With $\gamma$ as above, it is immediate to check that
\begin{equation} \label{eq:UPU=PJ}
\begin{pmatrix} \overline{\gamma(-k)} & 0 \\ 0 & \gamma(k) \end{pmatrix} 
\left[ \frac{1}{2} \begin{pmatrix} \Id_n & -W(k) \\ -W(k) & \Id_n \end{pmatrix} \right] 
\begin{pmatrix} \gamma(-k)^{\mathrm{t}} & 0 \\ 0 & \gamma(k)^* \end{pmatrix} = \frac{1}{2} \begin{pmatrix} \Id_n & -J_n \\ -J_n & \Id_n \end{pmatrix},
\end{equation}
that is, that $P^{-}(k)$ is unitarily conjugated, via 
\[ U(k) = \begin{pmatrix} U^{\uparrow}(k) & 0 \\ 0 & U^{\downarrow}(k) \end{pmatrix} := \begin{pmatrix} \overline{\gamma(-k)} & 0 \\ 0 & \gamma(k) \end{pmatrix}, \] 
to a constant projection-valued map. Notice also that this $U(k)$ has the structure required in~\eqref{eq:U_DIII}. Thus, Proposition~\ref{prop:gamma_dimerization} applies and yields
\begin{equation} \label{eq:DIII_gamma} 
\G(P^{\pm}) = (-1)^{[\det(U^{\downarrow}(\cdot))]} = (-1)^{[\det( \gamma(\cdot))]}.
\end{equation}
The relation~\eqref{eq:UPU=PJ} and a computation similar to~\eqref{eq:basis_from_W} allow also to check that the vectors
\[ \tilde{v}_j(k) := \frac{1}{\sqrt{2}} \begin{pmatrix} -\gamma(-k)^{\mathrm{t}} J_n e_j \\ \gamma(k)^* e_j \end{pmatrix} = \frac{1}{\sqrt{2}} \begin{pmatrix} \gamma(-k)^{\mathrm{t}} e_{j+m} \\ \gamma(k)^* e_j \end{pmatrix}, \quad j \in \{1,\ldots,m\}, \; k \in \mathbb{T}^{1}, \]
are orthonormal and in the range of $P^{-}(k)$. This collection of vectors can be complemented to a time-reversal symmetric basis for the projection-valued map $P^{-}$ setting
\[
\tilde{v}_{j+m}(k) := T \tilde{v}_{j}(-k) = \frac{1}{\sqrt{2}} \begin{pmatrix} \gamma(-k)^{\mathrm{t}} f_{j} \\ -\gamma(k)^* f_{j+m} \end{pmatrix}.
\]
In turn, the collection $\{\tilde{v}_1, \ldots, \tilde{v}_{2m}\}$ can be extended to a symmetric basis for the pair $P^{\pm}$ by applying the chiral symmetry operator $S$.

Let us now compute $\mathcal{A}^S$ via~\eqref{eq:AS} in this basis. We obtain
\begin{align*}
i \pi \mathcal{A}^S & = \frac{1}{2} \int_{0}^{\pi} \sum_{j=1}^{m} \left( \inn{\begin{pmatrix} \gamma(-k)^{\mathrm{t}} e_{j+m} \\ \gamma(k)^* e_j \end{pmatrix}}{\partial_k \begin{pmatrix} \gamma(-k)^{\mathrm{t}} e_{j+m} \\ -\gamma(k)^* e_j \end{pmatrix}} \right. \\
& \qquad\qquad - \left. \inn{\begin{pmatrix} \gamma(-k)^{\mathrm{t}} f_{j} \\ -\gamma(k)^* f_{j+m} \end{pmatrix}}{\partial_k \begin{pmatrix} \gamma(-k)^{\mathrm{t}} f_{j} \\ \gamma(k)^* f_{j+m} \end{pmatrix}} \right) dk \\
& = \frac{1}{2} \int_{0}^{\pi} \sum_{j=1}^{m} \big( \inn{\gamma(-k)^{\mathrm{t}} e_{j+m}}{[-(\partial_k \gamma)(-k)]^{\mathrm{t}} e_{j+m}} + \inn{\gamma(k)^* e_j}{[-(\partial_k \gamma)(k)]^* e_j}  \\
& \qquad\qquad - \inn{\gamma(-k)^{\mathrm{t}} f_{j}}{[-(\partial_k \gamma)(-k)]^{\mathrm{t}} f_{j}} - \inn{-\gamma(k)^* f_{j+m}}{[(\partial_k \gamma)(k)]^* f_{j+m}} \big) dk \\
& = \frac{1}{2} \int_{0}^{\pi} \sum_{j=1}^{m} \left( \inn{e_j}{[-\gamma(k)(\partial_k \gamma)(k)^*] e_j} + \inn{e_{j+m}}{[-\overline{\gamma(-k)}(\partial_k \gamma)(-k)^{\mathrm{t}}] e_{j+m}} \right.  \\
& \qquad\qquad - \left. \inn{[-\overline{(\partial_k \gamma)(-k)}\gamma(-k)^{\mathrm{t}}] \mathcal{K} e_{j}}{\mathcal{K} e_{j}} - \inn{[-(\partial_k \gamma)(k) \gamma(k)^*] \mathcal{K} e_{j+m}}{ \mathcal{K} e_{j+m}} \right) dk \\
& = \frac{1}{2} \int_{0}^{\pi} \sum_{j=1}^{m} \left( \inn{e_j}{[(\partial_k\gamma)(k) \gamma(k)^*] e_j} + \inn{e_{j+m}}{[\overline{(\partial_k \gamma)(-k)}\gamma(-k)^{\mathrm{t}}] e_{j+m}} \right.  \\
& \qquad\qquad + \left. \inn{e_{j}}{[(\partial_k \gamma)(-k) \gamma(-k)^{*}] e_{j}} + \inn{e_{j+m}}{[\overline{(\partial_k \gamma)(k)} \gamma(k)^{\mathrm{t}}] e_{j+m}} \right) dk \\
\end{align*}
where in the last equality we used the identity $(\partial_k \gamma)(k) \gamma(k)^* + \gamma(k) (\partial_k \gamma)(k)^* \equiv 0$, which follows upon differentiation of $\gamma(k) \gamma(k)^* \equiv \Id_n$, and the antiunitarity of the complex conjugation operator $\mathcal{K}$. With \eqref{g(-k)}, the above expression can be simplified to
\begin{align*} 
\mathcal{A}^S & = \frac{1}{2\pi i} \int_{\mathbb{T}^{1}} \sum_{j=1}^{m} \left( \inn{e_j}{[(\partial_k\gamma)(k) \gamma(k)^*] e_j} + \overline{\inn{e_{j+m}}{[(\partial_k \gamma)(k)\gamma(k)^{*}] e_{j+m}}} \right) dk \\
& = \frac{1}{2\pi i} \int_{\mathbb{T}^{1}} \sum_{j=1}^{m} \left( \inn{e_j}{[(\partial_k\gamma)(k) \gamma(k)^*] e_j} - \inn{e_{j+m}}{[(\partial_k \gamma)(k)\gamma(k)^{*}] e_{j+m}} \right) dk \\
\end{align*}
since, as was already shown, $(\partial_k\gamma)(k) \gamma(k)^* = -\gamma(k) (\partial_k\gamma)(k)^* = -[(\partial_k\gamma)(k) \gamma(k)^*]^*$ is a skew-adjoint operator. Notice now that
\begin{align*}
\sum_{j=1}^{m} & \inn{e_j}{[(\partial_k\gamma)(k) \gamma(k)^*] e_j} - \inn{e_{j+m}}{[(\partial_k \gamma)(k)\gamma(k)^{*}] e_{j+m}} \\
& = \sum_{j=1}^{2m} \inn{e_j}{[(\partial_k\gamma)(k) \gamma(k)^*] e_j} - 2 \sum_{j=1}^{m} \inn{e_{j+m}}{[(\partial_k \gamma)(k)\gamma(k)^{*}] e_{j+m}} \\
& = \tr_{\mathcal{H}^{\downarrow}} \left((\partial_k\gamma)(k) \gamma(k)^*\right) - 2 \tr_{\mathcal{H}^{\downarrow}} \left(\Pi_> (\partial_k\gamma)(k) \gamma(k)^* \Pi_> \right),
\end{align*}
where 
\[ \Pi_> := \begin{pmatrix} 0 & 0 \\ 0 & \Id_m \end{pmatrix} = \Pi_>^* = \Pi_>^2 \]
is the projection onto $\operatorname{Span}\{e_{m+1}, \ldots, e_{2m}\} \subset \mathcal{H}^{\downarrow}$. Again by~\cite[Proposition~A.1]{monaco2023z2} we have
\[ \tr_{\mathcal{H}^{\downarrow}} \left(\partial_k [\Pi_> \gamma(k)] [\Pi_>  \gamma(k)]^*\right) = \overline{\det(\Pi_> \gamma(k))} \, \partial_k [\det(\Pi_> \gamma(k))] \equiv 0, \quad k \in \mathbb{T}^{1}, \]
as $\det(\Pi_>\gamma(k)) = \det(\Pi_>) \det(\gamma(k)) \equiv 0$. Resuming the above computation thus yields
\[ 
\mathcal{A}^S = \frac{1}{2\pi i} \int_{\mathbb{T}^{1}} \tr_{\mathcal{H}^{\downarrow}} \left((\partial_k\gamma)(k) \gamma(k)^*\right) dk = [\det(\gamma(\cdot))]
\]
and comparing the above with~\eqref{eq:DIII_gamma} allows to conclude that $\G(P^{\pm}) = e^{i \pi \mathcal{A}^S}$, as desired.

The same argument as in the proof of Proposition~\ref{prop:W_dimerization} shows that the integrand in $\mathcal{A}^S$ is gauge-invariant, or, more specifically, that it does not change under unitary reshuffling of the first $m$ vectors in the basis for the projection-valued map $P^{-}$ (the changes of the remaining vectors in the range of $P^{-}$ and in the range of $P^{+}$ constituting a symmetric basis can be inferred from time-reversal and chiral symmetry, respectively). Therefore, the above computation actually holds in any symmetric basis for $P^{\pm}$, and the proof is complete.
\end{proof}

\begin{remark}
It follows from Proposition~\ref{prop:Class_AII_frames} that a projection-valued map $P \colon \mathbb{T}^{2} \to \Proj_n(\mathcal{H})$, $n = 2m \in \mathbb{N}$, in class AII can be decomposed as $P = P^{\uparrow} + P^{\downarrow}$, with $P^{\uparrow}, P^{\downarrow} \colon \mathbb{T}^{2} \to \Proj_{m}(\mathcal{H})$ related by $T P^{\uparrow}(k) = P^{\downarrow}(-k) T$, $k \in \mathbb{T}^{2}$, and such that $\FKM(P) = (-1)^{\Ch(P^{\uparrow})}$: see also~\cite{Peluso2026} and the discussion in \cite{kellendonk2019cyclic} after Corollary 7.4 for a similar statement in the context of torsion invariants in $K$-theory. The above Proposition suggests a similar connection between the DIII invariant and a decomposition of $P^{-}$ into $P^{-, \uparrow} + P^{-, \downarrow}$, in which $P^{-, \uparrow}$ is spanned by the first $m$ vectors in a time-reversal symmetric basis and, together with $P^{+, \uparrow}(\cdot) := S P^{-, \uparrow}(\cdot) S$, enters in a pair of projection-valued maps over $\mathbb{T}^{1}$ in class AIII (albeit acting on a non-minimal Hilbert space): the DIII invariant is expressed as the parity of the quantity $\mathcal{A}^S$, which appears also in the statement of Proposition~\ref{prop:W_frame_description}.
\end{remark}

\subsection{Direct sum and relative invariants}\label{sec:Stable_equivalence}
We now discuss how the relative invariants behave under direct sums of projection-valued maps. This is a natural operation from the physical point of view, since it corresponds to stacking two independent quantum systems, and from the geometric point of view, since it corresponds to taking the direct (or Whitney) sum of the associated vector bundles (compare Remark~\ref{rmk:Bloch_bundle}). The following Proposition shows that the relative invariants behave in a natural way: integer-valued invariants are additive, while $\mathbb Z_2$-valued invariants are multiplicative (recall that all through the main text we identified $\Z_2 = (\{\pm 1\}, \cdot)$ as a group).

\begin{proposition}\label{prop:realtive_invariants_additive}
Let $P^\pm \colon \T^{d} \to \Proj_n(\C^{2n})$ and $Q^\pm \colon \T^{d} \to \Proj_m(\C^{2m})$, $d \le 2$, be two pairs of projection-valued maps. Define their direct sum as
\[ (P \oplus Q)^{\pm}(k) := \begin{pmatrix}
        P^\pm(k) & 0 \\ 0 & Q^\pm(k)
    \end{pmatrix}, \quad k \in \mathbb{T}^{d}, \]
which constitutes a pair of projection-valued maps $(P \oplus Q)^{\pm} \colon \T^{d} \to \Proj_{n+m}(\C^{2n+2m})$. Then
\begin{itemize}
    \item if $P^\pm$ and $Q^\pm$ are chiral symmetric, then so is $(P \oplus Q)^\pm$ and 
    \[ \W((P \oplus Q)^\pm) = \W(P^\pm) + \W(Q^\pm); \]
    \item if $P^\pm$ and $Q^\pm$ are even-particle-hole symmetric, then so is $(P \oplus Q)^\pm$ and 
    \[ \PI((P\oplus Q)^\pm) = \PI(P^\pm) \cdot \PI(Q^\pm); \]
    \item if $P^\pm$ and $Q^\pm$ are in class DIII, then so is $(P \oplus Q)^\pm$ and
    \[ \G((P \oplus Q)^\pm) = \G(P^\pm) \cdot \G(Q^\pm). \]
    \end{itemize}
\end{proposition}

\begin{proof}
Throughout the proof, we'll denote $R^{\pm} := (P \oplus Q)^{\pm}$ for brevity.

Let us start from the chiral case. We can assume that $\C^{2n}$, the Hilbert space on which $P^{\pm}$ acts, admits a chiral symmetry of the form
\[ S_P = \begin{pmatrix}
\Id_n & 0 \\ 0 & -\Id_n
\end{pmatrix} \]
and similarly that $\C^{2m}$, on which $Q^{\pm}$ acts, is endowed with the chiral symmetry operator
\[ S_Q=\begin{pmatrix}
\Id_m & 0 \\ 0 & -\Id_m
\end{pmatrix}. \]
Then, according to Principle~\ref{pri:chiral_classes_min_dim},
\[P^\pm(k)=\frac{1}{2}\begin{pmatrix}
\Id_n & \pm W_P(k) \\ \pm W_P(k)^* & -\Id_n
\end{pmatrix}, \quad Q^\pm(k)=\frac{1}{2}\begin{pmatrix}
\Id_m & \pm W_Q(k) \\ \pm W_Q(k)^* & -\Id_m
\end{pmatrix}, \] 
for some unitary-valued maps $W_P \colon \mathbb{T}^{d} \to U(n)$ and $W_Q \colon \mathbb{T}^{d} \to U(m)$. Clearly
\[S_P \oplus S_Q =\begin{pmatrix}
\Id_n & 0 & 0 & 0  \\ 0 & -\Id_n & 0 &0 \\ 0 &0 &\Id_m & 0 \\ 0&0&0 & -\Id_m
\end{pmatrix}\] 
defines a chiral symmetry operator on $\C^{2n} \oplus \C^{2m}$, but in order to compute the chiral invariant of $R$ it is more convenient to bring it again to the form prescribed by Principle~\ref{pri:chiral_classes_min_dim}. This can be achieved by a change of basis performed via conjugation with the following matrix:
\[U:=\begin{pmatrix}
\Id_n & 0 & 0 & 0  \\ 0 &0 &\Id_n  &0 \\ 0 &\Id_m &0 & 0 \\ 0&0&0 & \Id_m
\end{pmatrix}.\]
Then
\[U^{-1}(S_P \oplus S_Q) U = \begin{pmatrix}
\Id_{n+m} & 0  \\ 0 & -\Id_{n+m}
\end{pmatrix}\] 
and a straightforward computation also reveals that
\begin{equation} \label{eq:W_direct_sum}
U^{-1}(R^+(k)-R^-(k))U=\begin{pmatrix}
0 & W_R(k) \\ W_R(k)^* &0
\end{pmatrix} \quad \text{with} \quad W_R(k) := W_P(k) \oplus W_Q(k).
\end{equation}
Since the change-of-basis matrix $U$ is constant in $k$, Proposition~\ref{prop:W_dimerization} guarantees that the chiral invariant of $U^{-1} R^{\pm} U$ and of $R^{\pm}$ coincide. We can therefore compute, if $d=1$,
\[\W(R^\pm)= \left[\det((W_P \oplus W_Q)(\cdot)\right] = [\det(W_P(\cdot)) \cdot \det(W_Q(\cdot))] = \W(P^\pm) + \W (Q^\pm) \in \Z.\]
Since the chiral invariant in $d=2$ is computed upon restriction to $1$-dimensional sub-tori, the above allows to conclude the desired statement.

Next, we move to the case of even-particle-hole symmetric pairs of projection valued maps. Comparing with~\eqref{eq:L(k)}, it suffices to notice that
\[ i(R^+(k)-R^-(k))=
\begin{pmatrix}
i(P^+(k)-P^-(k)) & 0 \\ 0 & i(Q^+(k)-Q^-(k))
\end{pmatrix}, \quad k \in \mathbb{T}^{d}, \] 
and the multiplicativity of the Pfaffian under direct sums ensures the claimed multiplicativity of the even-particle-hole invariant.

Finally, if the pairs are in class DIII, we can assume that the chiral and time-reversal symmetries on $\C^{2n}$ and $\C^{2m}$ are respectively
\[ S_P=\begin{pmatrix}
\Id_n & 0 \\ 0 & -\Id_n
\end{pmatrix}, \quad S_Q=\begin{pmatrix}
\Id_m & 0 \\ 0 & -\Id_m
\end{pmatrix}, \quad T_P= \mathcal{K}\begin{pmatrix}
0 & \Id_n  \\ -\Id_n & 0
\end{pmatrix},\quad T_Q= \mathcal{K} \begin{pmatrix}
0& \Id_m \\ -\Id_m & 0
\end{pmatrix}.\]
In order to bring also $S_P \oplus S_Q$ and $T_P \oplus T_Q$ to the same normal form, we conjugate through the change-of-basis matrix $U$ given by
\[ U:=\begin{pmatrix}
0 & \Id_{n+m} \\ -\Id_{n+m} & 0
\end{pmatrix}. \]
Once again,~\eqref{eq:W_direct_sum} holds, and, similarly to the chiral case, Proposition~\ref{prop:gamma_dimerization} guarantees that this constant change of basis does not affect the value of the DIII invariant of $R^{\pm}$. Upon inspection of Definition~\ref{def:DIII_invariant}, we are lead to conclude that the desired multiplicativity property of the DIII invariant holds if we can show that the Teo--Kane invariant of $1$-dimensional unitary-valued maps from Remark~\ref{def:gamma_invariant} is itself multiplicative:
\[ \TK(W_P(\cdot) \oplus W_Q(\cdot)) = \TK(W_P(\cdot)) \cdot \TK(W_Q(\cdot)) \in \mathbb{Z}_2. \]
But since both the determinant and the Pfaffian behave multiplicatively under direct sums, this is certainly the case.
\end{proof}

The above Proposition allows to promote the relative invariants to $K$-theoretic invariants. Indeed, by Remark~\ref{rmk:Bloch_bundle} a pair of projection-valued maps over the $d$-dimensional torus determines a pair of vector bundles over $\mathbb T^d$, linked by some further symmetries which act fiberwise. As mentioned previously, the direct sum of projection-valued maps corresponds to the Whitney sum of the associated bundles. This is the fundamental operation in the definition of the topological $K_0$-group: one starts from the semigroup of vector bundles under direct sum and then passes to its Grothendieck group. The statement above shows that the relative invariants behave as semigroup homomorphisms from the appropriate semigroup of symmetric projection-valued maps to the corresponding group of topological invariants (integers, integers modulo $2$, or products thereof), and can be thus naturally extended to group homomorphisms in $K$-theory. Notice that, while the relative invariants define complete homotopy invariants in minimal ambient Hilbert spaces, they probably lose their ability to characterize $K$-theory classes, as the latter is inherently related to the stabilization of (virtual) vector bundles, and thus to the addition of further (topologically trivial) dimensions to the ambient Hilbert space. 

\section*{Conflict of Interest, Funding and Data Availability}
The authors declare no conflict of interest. No new data were created or analyzed in this study. D.\ M.\ acknowledges financial support from  Sapienza Università di Roma within Progetto di Ricerca di Ateneo 2023, 2024, and 2025. G.\ P.\ acknowledges support from Grant DFF 5281-00046B of the Independent Research Fund Denmark | Natural Sciences. This work has been carried out under the auspices of the GNFM-INdAM (Gruppo Nazionale per la Fisica Matematica --- Istituto Nazionale di Alta Matematica).

\bibliographystyle{plain}
\bibliography{biblio}
\end{document}